%% file: Pre-Post-Lie-Final-Niels-Abel2015.tex
\documentclass[graybox]{svmult}

\usepackage{amsfonts,amsmath, amssymb,amscd,mathtools,mathrsfs,xspace, stmaryrd,
graphicx,wasysym,enumerate,color}
\usepackage[all]{xy}
\usepackage{mathptmx}       % selects Times Roman as basic font
\usepackage{helvet}         % selects Helvetica as sans-serif font
\usepackage{courier}        % selects Courier as typewriter font
\usepackage{type1cm}        % activate if the above 3 fonts are
\usepackage{makeidx}         % allows index generation
\usepackage{graphicx}        % standard LaTeX graphics tool
\usepackage{multicol}        % used for the two-column index
\usepackage[bottom]{footmisc}% places footnotes at page bottom

\makeindex             % used for the subject index
\begin{document}
%%%%%%%%%%%%%%%%%%%%%%%%%%%%%%%%%%%%%%%%%%%%%%%%%%%

\title*{Post-Lie Algebras, Factorization Theorems and Isospectral Flows}
\author{Kurusch Ebrahimi-Fard and Igor Mencattini}
\institute{Kurusch Ebrahimi-Fard \at Department of Mathematical Sciences, Norwegian University of Science and Technology -- NTNU, 7491 Trondheim, Norway, \email{kurusch.ebrahimi-fard@ntnu.no}. On leave from UHA, Mulhouse, France.
%Instituto de Ciencias Matem\'aticas -- ICMAT, Consejo Superior de Investigaciones Cient\'{i}ficas -- CSIC , C/ Nicol\'as Cabrera 13-15, 28049 Madrid, Spain, \email{kurusch@icmat.es}. On leave from UHA, Mulhouse, France.
\and Igor Mencattini \at Instituto de Ci\^encias Matem\'aticas e de Computa\c{c}\~ao, Universidade de S\~ao Paulo (USP), S\~ao Carlos, SP, Brazil, \email{igorre@icmc.usp.br}}
\maketitle

%%%%%%%%%%%%%%%%%%%%%%%%%%%%%

\abstract*{In these notes we review and further explore the Lie enveloping algebra of a post-Lie algebra. From a Hopf algebra point of view, one of the central results, which will be recalled in detail, is the existence of a second Hopf algebra structure. By comparing group-like elements in suitable completions of these two Hopf algebras, we derive a particular map which we dub post-Lie Magnus expansion. These results are then considered in the case of Semenov-Tian-Shansky's double Lie algebra, where a post-Lie algebra is defined in terms of solutions of modified classical Yang--Baxter equation. In this context, we prove a factorization theorem for group-like elements. An explicit exponential solution of the corresponding Lie bracket flow is presented, which is based on the aforementioned post-Lie Magnus expansion.}

\abstract{In these notes we review and further explore the Lie enveloping algebra of a post-Lie algebra. From a Hopf algebra point of view, one of the central results, which will be recalled in detail, is the existence of a second Hopf algebra structure. By comparing group-like elements in suitable completions of these two Hopf algebras, we derive a particular map which we dub post-Lie Magnus expansion. These results are then considered in the case of Semenov-Tian-Shansky's double Lie algebra, where a post-Lie algebra is defined in terms of solutions of modified classical Yang--Baxter equation. In this context, we prove a factorization theorem for group-like elements. An explicit exponential solution of the corresponding Lie bracket flow is presented, which is based on the aforementioned post-Lie Magnus expansion.}

%%%%%%%%%%%%%%%%%%%%%%%%%%%%%
\noindent {\footnotesize{.}}
%%%%%%%%%%%%%%%%%%%%%%%%%%%%%

\section{Introduction}
\label{sec:intro}

These notes are based on recent joint work \cite{EFLMK,EFLIM,EFIM} by the authors together with A.~Lundervold and H.~Z.~Munthe-Kaas. They present an extended summary of a talk given by the first author at the Instituto de Ciencias Matem\'aticas (ICMAT) in Madrid\footnote{Brainstorming Workshop on ``New Developments in Discrete Mechanics, Geometric Integration and Lie-Butcher Series", May 25-28, 2015, ICMAT, Madrid, Spain. Supported by a grant from Iceland, Liechtenstein and Norway through the EEA Financial Mechanism as well as the project ``Mathematical Methods for Ecology and Industrial Management" funded by Ayudas Fundaci\'on BBVA a Investigadores, Innovadores y Creadores Culturales.}. The main aim is to explore a certain factorization problem for Lie groups in the framework of universal enveloping algebra from the perspective offered by the relatively new theory of post-Lie algebras. The latter provides another viewpoint on the notion of (finite dimensional) double Lie algebra, which is a Lie algebra $\mathfrak g$ over the ground field $\mathbb{F}$ endowed with a solution $R_+ \in \operatorname{End}_\mathbb{F}(\mathfrak g)$ of the modified classical Yang--Baxter equation:
\begin{equation}
\label{R-matrix}
	[R_+ x, R_+ y] = R_+ ([R_+ x, y] + [x,R_+ y] - [x,y]).
\end{equation}
The identity implies that the bracket
$$
	[x,y]_{R_+} := [R_+ x, y] - [R_+ y,x] - [x,y]
$$
satisfies the Jacobi identity and therefore yields another Lie algebra, denoted $\mathfrak g_R$, on the vector space underlying $\mathfrak g$. Thanks to the seminal work of Semenov-Tian-Shansky \cite{STS1}, solutions of \eqref{R-matrix}, known as classical $r$-matrices, play an important role in studying solutions of Lax equations, which in turn are intimately related to a factorization problem in the Lie group corresponding to $\mathfrak g$. In the framework of the universal enveloping algebra of the Lie algebra $\mathfrak g$, this factorization problem has been studied in \cite{RSTS,STS3}. In these works it is shown, among other things, that every solution of the modified classical Yang--Baxter equation gives rise to a factorization of group-like elements in (a suitable completion of) the universal enveloping algebra of $\mathfrak g$. On the other hand, in \cite{GuoBaiNi} it was shown that in a Lie algebra $\mathfrak g$ every solution of \eqref{R-matrix} gives rise to a post-Lie algebra. 

\medskip

A post-Lie algebra \cite{LMK1,LMK2,Vallette}, which we denote by the triple $(V, \triangleright, [\cdot,\cdot])$, consists of a vector space $V$ which is endowed with two bilinear operations, the Lie bracket $[ \cdot , \cdot ]: V \otimes V \to V$ and the magmatic post-Lie product $\triangleright: V \otimes V \to V$. The relations that the latter is supposed to satisfy with respect to the Lie bracket imply that
\begin{equation}
\label{def:post-LieRel}
	 \llbracket x,y \rrbracket := x \triangleright y - y \triangleright x - [x,y]
\end{equation}
yields another Lie bracket on $V$. The complete definition will given further below. However, the following geometric example \cite{LMK1,LMK2} may provide some insight into the interplay between the post-Lie product and the Lie bracket in post-Lie algebra. Recall that a \emph{linear connection} is a $\mathbb F$-bilinear application $\nabla \colon \mathfrak X_M \times \mathfrak X_M  \rightarrow \mathfrak X_M$ on $\mathfrak X_M$, the vector space of smooth vector fields on the manifold $M$, satisfying the Leibniz rule $\nabla_X(fY)=X(f) Y+f\nabla_X Y,$ for all $f\in C^\infty(M)$ and all $X,Y\in\mathfrak X_M$. Clearly, a linear connection endows $\mathfrak X_M$ with a product, defined simply as $(X,Y) \mapsto X \curvearrowright  Y := \nabla_X Y.$ The \emph{torsion} of $\nabla$ is a skew-symmetric tensor $\mathrm{T} \colon  TM \wedge TM \rightarrow  TM$
\begin{equation}
	\mathrm{T}(X,Y) := X\curvearrowright Y -Y\curvearrowright X - [X,Y],
	\label{eq:torsion}
\end{equation}
where $[\cdot,\cdot ]$ denotes the \emph{Jacobi--Lie bracket} of vector fields, defined by $[X,Y](f) = X(Y(f))-Y(X(f))$, for every $X,Y\in\mathfrak X_M$ and every $f\in C^\infty(M)$. The \emph{curvature} tensor $\mathrm{R}\colon TM \wedge TM \rightarrow \mbox{End}(TM)$ satisfies the identity
\begin{align}
	\mathrm{R}(X,Y)Z &= X\curvearrowright (Y\curvearrowright Z)-(X\curvearrowright Y)\curvearrowright Z \nonumber \\
	& \quad - Y\curvearrowright (X\curvearrowright Z)+(Y\curvearrowright X)\curvearrowright Z
		+\mathrm{T}(X,Y)\curvearrowright Z.
	\label{eq:curvature}
\end{align}
Torsion and curvature are related by the \emph{Bianchi identities}
\begin{align}
	\sum_{\circlearrowleft}(\mathrm{T}(\mathrm{T}(X,Y),Z)+(\nabla_X\mathrm{T})(Y,Z)) 
	&= \sum_{\circlearrowleft}\mathrm{R}(X,Y)Z\label{eq:bianchi1}\\
	\sum_{\circlearrowleft}((\nabla_X\mathrm{R})(Y,Z)+\mathrm{R}(\mathrm{T}(X,Y),Z)) 
	&= 0 ,\label{eq:bianchi2}
\end{align}
where $\sum_{\circlearrowleft}$ denotes the sum over the three cyclic permutations of $(X,Y,Z)$. If a connection is \emph{flat}, $\mathrm{R}=0$, and has \emph{constant torsion}, $\nabla_X \mathrm{T}=0$, then \eqref{eq:bianchi1} reduces to the Jacobi identity, such that the torsion defines a Lie bracket $[X,Y] _{\mathrm{T}}:= \mathrm{T}(X,Y)$, which is related to the Jacobi--Lie bracket by \eqref{eq:torsion}. The covariant derivation formula $\nabla_X (\mathrm{T}(Y,Z)) = (\nabla_X\mathrm{T})(Y,Z)+\mathrm{T}(\nabla_X Y,Z)+\mathrm{T}(Y,\nabla_X Z)$ together with $\nabla_X \mathrm{T}=0$ imply
\begin{equation}
	X\curvearrowright [Y,Z]_\mathrm{T} 
	= [X\curvearrowright Y, Z]_\mathrm{T} + [Y,X\curvearrowright Z]_\mathrm{T} .\label{eq:post-Lie1}
\end{equation}
On the other hand, \eqref{eq:curvature} together with $\mathrm{R}=0$ yield
\begin{equation}
	[X,Y]_\mathrm{T}\curvearrowright Z = a_\curvearrowright(X,Y,Z) - a_\curvearrowright(Y,X,Z) ,\label{eq:post-Lie2}
\end{equation}
where $a_\curvearrowright(X,Y,Z):= (X\curvearrowright Y)\curvearrowright Z - X\curvearrowright (Y\curvearrowright Z)$ is the usual associator with respect to the product $\curvearrowright$. Relations \eqref{eq:post-Lie1} and \eqref{eq:post-Lie2} define the post-Lie algebra $(\mathfrak X_M,\curvearrowright, [\cdot,\cdot]_\mathrm{T})$, see Proposition \ref{pro:ge}. 

Note that for a connection which is both flat  and torsion free ($\mathrm{T}= 0=\mathrm{R}$), equation \eqref{eq:curvature} implies $a_\curvearrowright(X,Y,Z) = a_\curvearrowright(Y,X,Z)$. This is the characterizing identity of a (left) pre-Lie algebra, which is Lie admissible, i.e., by skew-symmetrization one obtains a Lie algebra. We refer the reader to \cite{Burde,Cartier11,ChaLiv,Manchon} for details. 

\medskip

Returning to the abstract definition of a post-Lie algebra, $(V, \triangleright, [\cdot,\cdot])$, we consider the lifting of the post-Lie product to the universal enveloping algebra, $\mathcal U(\mathfrak g)$, of the Lie algebra $\mathfrak g:=(V,[ \cdot , \cdot ])$. It turns out that it allows to define another Hopf algebra, $\mathcal U_*(\mathfrak g)$, on the underlying vector space of $\mathcal U(\mathfrak g)$, which is isomorphic, as a Hopf algebra, to the universal enveloping algebra corresponding to the Lie algebra $\bar{\mathfrak g}:=(V,\llbracket \cdot,\cdot \rrbracket)$ defined in terms of the Lie bracket \eqref{def:post-LieRel}. The Hopf algebra isomorphism between $\mathcal U(\bar{\mathfrak g})$ and $\mathcal U_*(\mathfrak g)$ is an extension of the identity between the Lie algebras $\mathfrak{g}$ and $\bar{\mathfrak{g}}$. Moreover, for every $x \in \mathfrak g$ there exists a unique element $\chi(x) \in \mathfrak g$, such that $\exp(x) = \exp^*(\chi(x))$ with respect to (suitable completions of) $\mathcal U(\mathfrak g)$ respectively $\mathcal U_*(\mathfrak g)$. The map $\chi: \mathfrak g \to \mathfrak g$ is called {\it{post-Lie Magnus expansion}} and is defined as the solution of a particular differential equation. 

From \cite{GuoBaiNi} we know that every solution of \eqref{R-matrix} turns a double Lie algebra \cite{STS1} into a post-Lie algebra. The Lie bracket $[\cdot,\cdot]_{R_+}$ on $\mathfrak g_R $ is a manifestation of \eqref{def:post-LieRel}, and the aforementioned Hopf algebra isomorphism between $\mathcal U_\ast(\mathfrak g)$ and $\mathcal U(\bar{\mathfrak g})=\mathcal U(\mathfrak g_R)$ can be realized in terms of the solution $R_+$ and the Hopf algebra structures of these two universal enveloping algebras. The role of post-Lie algebra in the context of the factorization problem on $\mathcal U(\mathfrak g)$, mentioned above, becomes clear from the fact that any group-like element $\exp(x)$ in (a suitable completion of) $\mathcal U(\mathfrak g)$ factorizes into the product of two group-like elements, $ \exp(\chi_+(x))$ and $\exp(\chi_-(x))$, with $\chi_{\pm}(x):=\pm R_{\pm}\chi(x)$, where $R_-:=R_+ - \operatorname{id}$.

\medskip 

In what follows $\mathbb F$ denotes the ground field of characteristic zero over which all algebraic structures are considered. Unless stated otherwise,  $\mathbb F$ will be either the complex numbers $\mathbb C$ or the real numbers $\mathbb R$.

\smallskip
{\bf{Acknowledgements}}: The first author acknowledges support from the Spanish government under the project MTM2013-46553-C3-2-P and from FAPESP under the project 2015/06858-2 . We also thank the anonymous referees for pointed suggestions and remarks that helped to improve the manuscript.

%%%%%%%%%%%%%%%%%%%%%%%%%%%%%

\section{Universal enveloping algebras and Hopf algebras}
\label{sec:LieTheory}

In this section we present some background on classical Lie theory. We recall notions and thereby fix notations used in later sections. The construction of the so-called \emph{$I$-adic} completion of an augmented algebra will be discussed since it plays a central role. For details the reader is referred to \cite{D-K,Postnikov,Quillen,Warner}. 

%%%%%%%%%%%%%%%%%%%%%%%%%%%%%

\subsection{Lie groups and Lie algebras}
\label{subsect:LieGroAlg}

A Lie group $G$ is a smooth manifold endowed with the structure of an abstract group, which is compatible with the underlying differentiable structure of $G$. This means that both maps, the multiplication $m: G \times G \rightarrow G$ and the inversion $i: G \rightarrow G$ are smooth applications. For each element $g \in G$, one can define the diffeomorphisms $L_g: G \rightarrow G$ and $R_g : G \rightarrow G$, called the \emph{left}- and \emph{right-translations} by $g$, respectively. In the following we will use the $\ast$-notation to denote the differential\footnote{More precisely, for $\phi : M_1 \rightarrow M_2$, $M_2 \neq \mathbb F$, the differential of $\phi$ at $m \in M_1$ will be denoted as $\phi_{\ast,m}$. This is a linear map between $T_m M_1$ and $T_{\phi(m)}M_2$ such that $(\phi_{\ast,m}v)f = v(f\circ\phi)$, for all $v \in T_mM_1$ and all $f \in C^\infty(M_2)$. In case of $M_2=\mathbb F$, i.e., if $\phi = H : M \rightarrow \mathbb F$ is a smooth function, we will write its differential at the point $m \in M$ as $dH_m \in T^\ast_mM$.} of a smooth application. For each element $x$ in $T_eG$, the tangent space at the identity $e \in G$, let $X_x : G \rightarrow TG$ be the map defined for all $g \in G$ by $X_x : g \mapsto (L_g)_{\ast,e}x$. Then $X_x$ is smooth and it satisfies $\pi \circ X_x=\operatorname{id}_G$, where $\pi: TG \rightarrow G$ is the canonical projection. In other words, $X_x$ is a \emph{left-invariant, smooth vector field} on $G$. The set $\mathfrak X(G)^G \subset \mathfrak X(G)$ of all left-invariant vector fields forms a Lie algebra of $\mathfrak X(G)$ of dimension equal to the dimension of $G$, and $X \in \mathfrak X(G)$ is left-invariant if and only if $X=X_x$ for some $x \in T_eG$. This observation implies the existence of a bilinear, skew-symmetric bracket $[\cdot,\cdot]: T_eG \times T_eG \rightarrow T_eG$ satisfying the Jacobi identity. By definition, $[x,y]:=[X_x,X_y](e)$, for all $x,y\in T_eG$. The Lie algebra $(T_eG,[\cdot,\cdot])$ will be denoted by $\mathfrak g$. To every homomorphism of Lie groups corresponds a homomorphism between the corresponding Lie algebras, i.e., if $\psi : G_1 \rightarrow G_2$ is a homomorphism between groups $G_i$, $i=1,2$, with corresponding Lie algebras $\mathfrak g_i=(T_eG_i,[\cdot,\cdot]_i)$, $i=1,2$, then its differential evaluated at $e$ satisfies $\psi_{\ast,e}[x,y]_1=[\psi_{\ast,e}(x),\psi_{\ast,e}(y)]_2$, for all $x,y\in\mathfrak g_1$.  Since every left (right) invariant vector field is complete, for every $x \in\mathfrak g$ the integral curve of $X_x$, going through the identity $e \in G$ at $t=0$, defines a smooth map $\gamma^x_e : \mathbb{R} \rightarrow G$. It can be shown that $\gamma^{x}_e$ is a Lie group homomorphism from $(\mathbb R,+)$ to $G$, and every {\it continuous} group homomorphism between $(\mathbb R,+)$ and $G$ is of this form. For all $x \in \mathfrak g$ the curve $\gamma^x_{e}$ is called a {\it $1$-parameter group homomorphism} of $G$. Given the latter, one can define the \emph{exponential map}, $\exp:\mathfrak g\rightarrow G$, $x \in \mathfrak g \mapsto \operatorname{exp} x =\gamma^x_e(1),$ which is smooth and has the following properties:
\begin{enumerate}
 
 \item $(\operatorname{exp})_{\ast,0}=\operatorname{id}_{\mathfrak{g}}$. In particular, there exist $U\subset\mathfrak{g}$ and $V\subset G$, open neighborhoods of $0\in\mathfrak{g}$ respectively $e\in G$, such that $\operatorname{exp}\vert_{U}:U\rightarrow V$ is a diffeomorphism.

 \item Let $G_1$ and $G_2$ be two Lie groups and $\mathfrak{g}_1$ respectively $\mathfrak{g}_2$ the corresponding Lie algebras. If $\phi:G_1\rightarrow G_2$ is a Lie group morphism, then $\operatorname{exp}\circ\phi_{\ast,e}=\phi\circ\operatorname{exp}.$
\end{enumerate}

A closed formula for the differential of the exponential map at $x \in \mathfrak g$ is:
\begin{equation}
	\exp_{\ast,x}=(L_{\exp x}\big)_{\ast,e}\circ\frac{\operatorname{id}_{\mathfrak g}-\operatorname{e}^{-\operatorname{ad}_x}}{\operatorname{ad}_x}.
	\label{eq:difexp}
\end{equation}
Note that $\operatorname{ad}_xy:=[x,y]$, for all $x,y\in\mathfrak g$. The formal expression $\frac{\operatorname{id}_{\mathfrak g}-\operatorname{e}^{-\operatorname{ad}_x}}{\operatorname{ad}_x}$ represents the element of $\operatorname{End}_{\mathbb F}(\mathfrak g)$ defined by
\[
	\frac{\operatorname{id}_{\mathfrak g}-\operatorname{e}^{-\operatorname{ad}_x}}{\operatorname{ad}_x}
	=\int_0^1 \operatorname{e}^{-s\operatorname{ad}_x}\,ds.
\] 
One can prove that the exponential map is a local diffeomorphism in a neighborhood of $x \in \mathfrak g$ if and only if the linear operator $\operatorname{ad}_x$ has no eigenvalues in the set $2 \pi\,\imath\,\mathbb Z\backslash\{0\}$. Choosing $x,y\in\mathfrak g$ belonging to a sufficiently small open neighborhood $U$ of $0 \in \mathfrak g$, such that $\exp x\exp y\in V \subset G$, where $V$ is a small neighborhood of $e \in G$, one is able to find an element $\operatorname{BCH}(x,y)\in\mathfrak g$ such that $\exp\operatorname{BCH}(x,y)=\exp x\,\exp y,$ or, what is equivalent, such that $\operatorname{BCH}(x,y)=\log\big(\exp x\,\exp y\big)$, where $\log:V\rightarrow U$ denotes the inverse of the restriction of the exponential map. An explicit formula for $\operatorname{BCH}(x,y)$ is given by the so-called \emph{Baker--Campbell--Hausdorff} series ($\operatorname{BCH}$-series). See \cite{Postnikov} for details. The first few terms of this \emph{Lie series} in the variables $x,y \in \mathfrak g$ are:
\[
	\operatorname{BCH}(x,y) = x + y + \frac 12[x,y]+\frac 1{12}([x,[x,y]]+[y,[y,x]]) + \cdots.
\]
The reduced $\operatorname{BCH}$-series is defined by $\overline{\operatorname{BCH}}(x,y):=\operatorname{BCH}(x,y) - x - y$.

%%%%%%%%%%%%%%%%%%%%%%%%%%%%%

\subsection{Universal enveloping algebras and Hopf algebras}
\label{subsec:universalAlg}

In this subsection we follow references \cite{Postnikov, Reu}. Let $\mathfrak{g}$ be a finite dimensional Lie algebra and let $T(\mathfrak{g})=T_{\mathfrak{g}}=(T^{\bullet}\mathfrak{g},\otimes)$ be its \emph{tensor algebra}, which is a {\it graded}, {\it associative} and {\it non--commutative} algebra, whose homogeneous sub-space of degree $n$, $\mathfrak g^n:=\mathfrak g \otimes \mathfrak g^{n-1}$, $\mathfrak g^0:=1$,  is generated, as vector space, by monomials of the form $x_{i_1}\otimes\cdots\otimes {x_{i_n}}$. Consider the \emph{$2$-sided ideal} 
\[
	J=\langle x\otimes y-y\otimes x-[x,y]\rangle:=T_{\mathfrak g}\big(x\otimes y-y\otimes x-[x,y]\big)T_{\mathfrak g}.
\]
Note that in the following we will identify $x_{i_1}\otimes\cdots\otimes {x_{i_n}}$ with words $x_{i_1}\cdots x_{i_n}$. 

\begin{definition}[Universal enveloping algebra]\label{def:unienal}
The \emph{universal enveloping algebra} of $\mathfrak g$ is the algebra $\mathcal{U}(\mathfrak g):=T_{\mathfrak{g}}/J$. Its product is induced by the tensor product $\otimes$, i.e., if $\overline{X},\overline{Y}\in\mathscr{U}(\mathfrak g)$ are the classes of the monomials $X\in\mathfrak{g}^{k}$ and $Y\in\mathfrak{g}^{l}$, then $\overline{X}\cdot\overline{Y}$ is the class of the monomial $X\otimes Y\in\mathfrak{g}^{k+l}$.
\end{definition}

Note that $\mathcal{U}(\mathfrak{g})$ is a \emph{unital}, \emph{associative} algebra. In general it is not graded, since the ideal $J$ is non--homogeneous. However, $\mathcal U(\mathfrak g)$ is a \emph{filtered} algebra, that is, it is endowed with the filtration ${\mathbb F}=\mathcal{U}_0(\mathfrak{g})\subset\mathcal{U}_1(\mathfrak{g})\subset\cdots\subset\mathcal{U}_n(\mathfrak{g})\subset\cdots$, where $\mathcal{U}_n(\mathfrak{g})$ is the subspace of $\mathcal{U}(\mathfrak{g})$ generated by {\it monomials of length at most $n$}, i.e., monomials like $x_{i_1}\cdots x_{i_m}$, $m\leq n$, with $x_{i_1},\dots,x_{i_m} \in\mathfrak{g}$. Note that $\mathcal{U}_i(\mathfrak{g})\cdot\mathcal{U}_j(\mathfrak{g})\subset\mathcal{U}_{i+j}(\mathfrak{g})$, for all $i,j\geq 0$, and that $\mathcal{U}(\mathfrak{g})=\cup_{k\geq 0}\mathscr{U}_k(\mathfrak{g})$. Observe that $\mathcal{U}_1(\mathfrak{g})\simeq\mathfrak{g}$, so that there is a {\it natural} homomorphism of Lie algebras, $i:\mathfrak{g}\rightarrow\mathcal{U}(\mathfrak{g})_{\text{Lie}}$. The adjective universal emphasizes the fact that $\mathcal{U}(\mathfrak{g})$ has the following property: suppose that $\mathscr A$ is an associative algebra and that $j:\mathfrak{g} \rightarrow {\mathscr A}_{\text{Lie}}$ is a morphism of Lie algebras. Then there exists a {\it unique} morphism of unital associative algebras, $\phi:\mathcal{U}(\mathfrak{g})\rightarrow{\mathscr A}$, which makes the following diagram of Lie algebras commutative:  
\[
	\xy\xymatrix{
  	&\mathfrak{g}\ar[dr]^j\ar[d]_{i} &\\
  	&\mathcal{U}(\mathfrak{g})_{\text{Lie}}\ar[r]_{\phi_L}&{\mathscr A}_{\text{Lie}}}\endxy
\]
  
The {\it graded algebra} associated to $\mathscr{U}(\mathfrak{g})$ is: 
\[
	\operatorname{gr}(\mathcal{U}(\mathfrak{g}))
	=\bigoplus_{k\geq 0}\displaystyle\frac{\mathcal{U}_k(\mathfrak{g})}{\mathcal{U}_{k-1}(\mathfrak{g})},
	\quad\mathcal{U}_{-1}(\mathfrak{g})=\{0\}.
\]
Furthermore, note that $\mathfrak{g}\simeq \mathcal{U}_1(\mathfrak{g})/\mathcal{U}_{0}(\mathfrak{g})$, so that there exists a linear map $i:\mathfrak{g}\rightarrow \operatorname{gr}(\mathcal{U}(\mathfrak{g}))$ and $\operatorname{gr}(\mathcal{U}(\mathfrak{g}))$, endowed with the obvious multiplication, is a {\it commutative algebra}. In fact, for every $k\geq 0$, $x_{i_1}\cdots x_{i_k} - x_{\sigma(i_1)}\cdots x_{\sigma(i_k)} \in \mathscr{U}_{k-1}(\mathfrak{g})$, for all $\sigma$ in the permutation group $\Sigma_k$ of $k$ elements. This is clear when $\sigma$ is a transposition. For a general $\sigma$, the statement follows from the fact that every permutation is the product of transpositions. Observe that since a general $x_{i_1}\cdots x_{i_k}\in\mathcal U_k(\mathfrak g)$ can be written as
\begin{equation}\label{eq:dec}
	x_{i_1}\cdots x_{i_k}=\frac{1}{k!}\sum_{\sigma\in\Sigma_k}x_{\sigma(i_1)}\cdots x_{\sigma(i_k)}
	+\frac{1}{k!}\sum_{\sigma\in\Sigma_k}\big(x_{i_1}\cdots x_{i_k}-x_{\sigma(i_1)}\cdots x_{\sigma(i_k)}\big),
\end{equation}
where each summand of the second sum is an element of $\mathcal U_{k-1}(\mathfrak g)$, for each $k\geq 1$ one has the following exact sequence of vector spaces
\begin{equation}
\begin{CD}
0@>>>\mathcal U_{k-1}(\mathfrak g)@>i_{k-1}>>
\mathcal U_k(\mathfrak g)@>\sigma_k>>\frac{\mathcal U_k(\mathfrak g)}{\mathcal U_{k-1}(\mathfrak g)}@>>>0\end{CD}\label{eq:symb1}
\end{equation}
where $\sigma_k(x_{i_1}\cdots x_{i_k})$ is the class in $\mathcal U_k(\mathfrak g)/\mathcal U_{k-1}(\mathfrak g)$ of the sum $\frac{1}{k!}\sum_{\sigma\in\Sigma_k}x_{\sigma(i_1)}\cdots x_{\sigma(i_k)}$, see formula \eqref{eq:dec}.

Together with the universal enveloping algebra, one can introduce the \emph{symmetric algebra} of $\mathfrak g$,  $S(\mathfrak{g})=S_{\mathfrak{g}}:=T_{\mathfrak{g}}/{J'}$ where the $2$-sided ideal $J':=T_{\mathfrak{g}}\big(x\otimes y-y\otimes x\big)T_{\mathfrak{g}}$. $S_{\mathfrak{g}}$ is a {\it graded commutative algebra} endowed with a natural {\it injective} linear map $j:\mathfrak{g}\rightarrow S_{\mathfrak{g}}$, having the following {\it universal property}: given a commutative algebra $\mathscr{C}$ and a linear map $f:\mathfrak{g}\rightarrow\mathscr C$ there exists a \emph{unique} map of commutative algebras $\phi:S_{\mathfrak{g}}\rightarrow\mathscr C$ which closes the following to a commutative diagram:
\[
  \xy\xymatrix{
  &\mathfrak{g}\ar[dr]^f\ar[d]_{j} &\\
  &S_{\mathfrak{g}}\ar[r]_{\phi}&{\mathscr C}}\endxy
\]
 
For each $k\geq 0$, $S_k(\mathfrak g)$ denotes the homogeneous component of degree $k$ of $S_{\mathfrak g}$, and $S_{\mathfrak g}=\oplus_{k\geq 0}S_k(\mathfrak g),$ where $S_0(\mathfrak g):=\mathbb F$ and $S_{-1}(\mathfrak g)=\{0\}$. Letting $\mathscr C=\operatorname{gr}(\mathscr{U}(\mathfrak g))$ and $f=i:\mathfrak g\rightarrow \operatorname{gr}(\mathscr{U}(\mathfrak g))$, one can state the following important result. 
  
\begin{theorem}[Poincar\'e--Birkhoff--Witt]
The corresponding map $\phi:S_{\mathfrak{g}}\rightarrow\operatorname{gr}(\mathscr{U}(\mathfrak{g}))$ in the above diagram is an {\it isomorphism} of graded commutative algebras. In particular, for each $k\geq 0$ one has that 
\begin{equation}
	\phi_k:=\phi\vert_{S_k(\mathfrak g)}:S_k(\mathfrak g)\rightarrow
	\frac{\mathcal U_k(\mathfrak g)}{\mathcal U_{k-1}(\mathfrak g)}
	\label{eq:symm1}
\end{equation}
is an isomorphism of vector spaces.
\end{theorem}

Note that $\phi_k$ in \eqref{eq:symm1} maps every monomial $x_{i_1}\cdots x_{i_k}\in S_k(\mathfrak g)$ to the class of $x_{i_1}\cdots x_{i_k}$ in $\mathcal U_k(\mathfrak g)/\mathcal U_{k-1}(\mathfrak g)$, i.e., $\phi(x_{i_1}\cdots x_{i_k})=x_{i_1}\cdots x_{i_k}\,\text{mod}\,\mathcal U_{k-1}(\mathfrak g),$ where the product on the l.h.s.~is the one in the symmetric algebra while the product on the r.h.s.~is the one in the universal enveloping algebra. Since for each $x_{i_1}\cdots x_{i_k}\in\mathcal U_k(\mathfrak g)$
\[
	x_{i_1}\cdots x_{i_k}=\frac{1}{k!}\sum_{\sigma\in\Sigma_k}x_{\sigma(i_1)}\cdots 
	x_{\sigma(i_k)}\,\text{mod}\,\mathcal U_{k-1}(\mathfrak g),
\]
see \eqref{eq:dec}, for each $k$,  one can use the (inverse) of the map $\phi_k$ together with \eqref{eq:symb1}, to define the following exact sequence
\begin{equation}
	\begin{CD}
	0@>>>\mathcal U_{k-1}(\mathfrak g)@>i_{k-1}>>
	\mathcal U_k(\mathfrak g)@>s_k>>S_k(\mathfrak g)@>>>0\end{CD}\label{eq:symb2}
\end{equation}
where $s_k=\phi_k^{-1}\circ\sigma_k$ is defined by 
\begin{equation}
	s_k(x_{i_1}\cdots x_{i_k})=x_{i_1}\cdots x_{i_k}.\label{eq:symb3}
\end{equation}

The map $s_k$ defined above is called the (degree $k$) \emph{symbol map}. Since \eqref{eq:symb2} is an exact sequence of vector spaces it splits.  The linear map $\text{symm}_k:S_k(\mathfrak g)\rightarrow\mathcal U_k(\mathfrak g)$ defined by
\begin{equation}
	\text{symm}_k(x_{i_1}\cdots x_{i_k})=
	\frac{1}{k!}\sum_{\sigma\in\Sigma_k}x_{\sigma(i_1)}\cdots x_{\sigma(i_k)}\label{eq:symm3}
\end{equation}
and called the (degree $k$) \emph{symmetrization map}, is a \emph{section} of \eqref{eq:symb2}, i.e., for each $k$, $s_k\circ\text{symm}_k=\operatorname{id}_{S_k(\mathfrak g)}.$

Note that both products in \eqref{eq:symb3} and \eqref{eq:symm3} on the right and left side should be interpreted accordingly to the meaning of the monomials.

Observe that when $\mathfrak{g}$ is abelian $\mathscr{U}(\mathfrak{g})=S_{\mathfrak{g}}$, while for general $\mathfrak{g}$ the Poincar\'e--Birkhoff--Witt theorem tells us that we still have an isomorphism $\mathscr{U}(\mathfrak{g})\simeq S_{\mathfrak{g}}$ but {\it only} at the level of vector spaces. 
 
The universal enveloping algebra is an example of a {\it quasi-commutative algebra}, i.e., an associative, unital and filtered algebra $\mathscr A$, whose associated graded algebra $\operatorname{gr}(\mathscr A)$ is commutative. One can prove that if  $\mathscr A$ is a quasi-commutative algebra, then $\operatorname{gr}(\mathscr A)$ is a \emph{Poisson algebra}, see Section \ref{sec:PoisYB}. To define the Poisson bracket on $\operatorname{gr}(\mathscr A)$ it suffices to define it on the homogeneous components of the associated algebra. To this end, let:
\begin{equation}
	\{\cdot,\cdot\}:\frac{{\mathscr A}_i}{{\mathscr A}_{i-1}}\times\frac{{\mathscr A}_j}{{\mathscr A}_{j-1}}\rightarrow
	\frac{{\mathscr A}_{i+j-1}}{{\mathscr A}_{i+j-2}},
	\quad(\overline{x},\overline{y})\rightarrow (xy-yx)\:\text{mod}\:{\mathscr A}_{i+j-2},
	\label{eq:poisgrad}
\end{equation}
where $x\in\mathscr A_{i}$ and $y\in\mathscr A_{j}$ are two lifts of $\overline x$ respectively $\overline{y}$. The proof follows at once after showing that such a bracket is well defined, in particular, that given $x,y$ as above $xy-yx \in \mathscr A_{i+j-1}$, and that the result does not depend on the choice of the two lifts. Given that, the proof that the above bracket is Poisson follows from the fact that $\mathscr A$ is an associative algebra. Then, in particular, given a Lie algebra $\mathfrak{g}$, the graded algebra associated to ${\mathscr U}(\mathfrak{g})$ is a Poisson algebra. Using the Poincar\'e--Birkhoff--Witt theorem, such a Poisson structure can be transferred to the symmetric algebra $S_{\mathfrak{g}}$. In this framework it is worth to note that the Poisson bracket induced on $S_{\mathfrak{g}}$ by the one defined on $\operatorname{gr}({\mathscr U}(\mathfrak{g}))$ coincides with the linear Poisson structure of $\mathfrak g$, see \eqref{eq:linpoi} further below. To prove this statement, it suffices to check it on the restriction of (\ref{eq:poisgrad}) to the components of degree $1$:
\[
	\{\cdot,\cdot\}:\frac{{\mathscr U}_1(\mathfrak{g})}{{\mathscr U}_{0}(\mathfrak{g})}\times
	\frac{{\mathscr U}_1(\mathfrak{g})}{{\mathscr U}_{0}(\mathfrak{g})}\rightarrow
	\frac{{\mathscr U}_{1}(\mathfrak{g})}{{\mathscr U}_{0}(\mathfrak{g})},
	\quad(\overline{x},\overline{y})\rightarrow (xy-yx)\:\text{mod}\:\mathscr{U}_0(\mathfrak{g}),
\]
which shows that $\{\overline{x},\overline{y}\}=\overline{[x,y]}$ (remember that ${\mathscr U}_1(\mathfrak{g})/{\mathscr U}_{0}(\mathfrak{g})\simeq\mathfrak{g}$ and that ${\mathscr U}_{0}(\mathfrak{g})\simeq\mathbb F$). Furthermore, one can prove that if $\mathscr A$ is a positively filtered algebra, such that 
\begin{enumerate}
\item[i)] $\mathscr A_0=\mathbb F$, 
\item[ii)] $\mathscr A$ is generated as a ring by $\mathscr A_1$ and 
\item[iii)] $\mathscr A$ is almost-commutative, then
\end{enumerate}
there exists  a Lie algebra $\mathfrak g$ and an ideal $I$ of $\mathcal U(\mathfrak g)$, such that  $\mathscr A\simeq \mathcal U(\mathfrak g)/I$. 

\medskip

\noindent
\emph{Universal enveloping algebra as a Hopf algebra.} The universal enveloping algebra $\mathscr U(\mathfrak g)$ of a Lie algebra $\mathfrak g$ carries the structure of a \emph{Hopf algebra}. We follow \cite{Cartier07,Sweedler}. 

Recall the triple notation $(A,m,i)$ for an associative unital $\mathbb F$-algebra $A$, where the multiplication $m: A \otimes A \rightarrow A$ and the unit map $i: \mathbb F \rightarrow A$ satisfy
\begin{eqnarray*}
	m\circ (m \otimes \mathrm{id})
	=m\circ ( \mathrm{id} \otimes m): A \otimes A \otimes A \rightarrow A
	\quad&&\text{\it associativity}\\ 
	m\circ(i \otimes \mathrm{id})
	= \mathrm{id} =m\circ ( \mathrm{id} \otimes i): A \otimes \mathbb F \simeq  \mathbb F \otimes A \rightarrow A
	\quad &&\text{\it unit property} .
\end{eqnarray*}
If $\tau: A\otimes A\rightarrow A\otimes A$, $\tau(a\otimes b):=b\otimes a$, then $A$ is called \emph{commutative} if $m\circ\tau=m$. 

A \emph{co-algebra} is defined as a triple $(C,\Delta,\epsilon)$, where $C$ is a vector space, and $\Delta:C\rightarrow C\otimes C$, $\epsilon:C\rightarrow\mathbb F$ are two linear maps, the first is called \emph{co-product} and the second is called \emph{co-unit}. Co-product and co-unit satisfy the following properties:
\begin{eqnarray*}
	(\Delta\otimes \mathrm{id})\circ \Delta
	=(\mathrm{id}\otimes \Delta)\circ\Delta: C \rightarrow C \otimes C \otimes C 
	\quad&&\text{\it co-associativity}\\
	(\epsilon\otimes \mathrm{id})\circ\Delta=\mathrm{id}
	=(\mathrm{id}\otimes\epsilon)\circ\Delta: C\rightarrow C
	\quad &&\text{\it co-unit property} .
\end{eqnarray*}
A co-algebra is \emph{co-commutative} if $\tau\circ\Delta=\Delta$. Note that the notions of algebra and co-algebra are \emph{almost} dual to each other. More precisely, the dual of a co-algebra is an algebra whose multiplication and unit maps are obtained by reversing arrows of co-multiplication and co-unit. On the other hand, taking the dual of an algebra and reversing the arrows of the multiplication and of the unit maps, one obtains a co-algebra $(A^\ast,m^\ast,i^\ast)$ if $\text{dim}\,\,A<\infty$ but not if $\text{dim}\,\,A=\infty$. In fact, in this case, reversing the multiplication arrow, one does not obtain a map $m^\ast$ from $A^\ast$ to $A^\ast\otimes A^\ast$, but rather from $A^\ast$ to $(A\otimes A)^\ast$ which contains $A^\ast\otimes A^\ast$ as a proper vector sub-space. Finally, a {\it{bialgebra}} $(H,m,i,\Delta,\epsilon )$ consists of a vector space $H$ endowed with the maps:
\begin{eqnarray*}
	m:H\otimes H\rightarrow H\quad&&\text{\it multiplication}\\
	\Delta:H\rightarrow H\otimes H\quad&&\text{\it co-multiplication}\\
	i:\mathbb F\rightarrow H\quad&&\text{\it unit}\\ 
	\epsilon:H\rightarrow\mathbb F\quad &&\text{\it co-unit},
\end{eqnarray*}
such that $(H,m,i)$ is an algebra and $(H,\Delta,\epsilon)$ is a co-algebra, which are compatible
\begin{eqnarray}
	\Delta\circ m &=& (m \otimes m)\circ ( \mathrm{id} \otimes\tau\otimes \mathrm{id} )
	\circ\Delta\otimes\Delta \label{eq:copr}\\
	\epsilon\otimes\epsilon &=& \epsilon\otimes m.\label{eq:couni}
\end{eqnarray} 
Note that these conditions are equivalent to saying that $(\Delta,\epsilon)$ are algebra morphisms -- equivalently, $(m,i)$ are co-algebra morphisms. A \emph{Hopf algebra} $(H,m,i,\Delta,\epsilon,S)$ is a bialgebra with an {\it{antipode}} $S: H \rightarrow H$, a linear map satisfying:
\[
	m\circ  (\mathrm{id} \otimes S)\circ\Delta 	= i\circ\epsilon = m\circ (S\otimes \mathrm{id} )\circ\Delta.
\]
It is easy to show that $S$ is a co-algebra and algebra anti-homorphism, such that $S\circ i=i$ and $\epsilon\circ S=\epsilon$.  

An element $x\in H$ will be called \emph{primitive} if $\Delta x=x\otimes 1+1\otimes x$, while $g\in H$ will be called \emph{group-like} if $g\neq 0$ and $\Delta g=g\otimes g$. Let $\mathcal P(H)$ and $\mathcal G(H)$ be the sets of primitive respectively group-like elements in the Hopf algebra $(H,m,i,\Delta,\epsilon,S)$. Note that if $g_1,g_2\in\mathcal G(H)$, then $g_1\cdot g_2:=m(g_1,g_2)\in\mathcal G(H)$. If $x_1,x_2\in\mathcal P(H)$, then $[x_1,x_2]:=x_1\cdot  x_2 - x_2\cdot x_1 \in\mathcal P(H)$, i.e., $(\mathcal P(H),[\cdot,\cdot ])$ is a Lie algebra. Furthermore, defining $e:=i(1)$ and $g^{-1}:=S(g)$ for all $g\in\mathcal G(H)$, one can show that $g\cdot e=g=e\cdot g$, and $g^{-1}\cdot g=e=g\cdot g^{-1}$, for all $g\in\mathcal G(H)$. In other words, $(\mathcal G(H),\cdot)$ is a group whose identity element is $e$, such that for each $g\in\mathcal G(H)$, $g^{-1}=S(g)$. 

\begin{proposition}
Let $\mathfrak g$ be a Lie algebra. Its universal enveloping algebra $\mathcal U(\mathfrak g)$ is a co-commutative Hopf algebra.
\end{proposition}

\begin{proof}
To prove the first part of the statement it suffices to define the antipode and a co-algebra structure compatible with the algebra structure of $\mathcal U(\mathfrak g)$. Let $\mathfrak G=\mathfrak g\oplus\mathfrak g$ be endowed with the structure of direct product Lie algebra and let $\Delta:\mathfrak g\rightarrow\mathfrak G$ be the diagonal embedding, i.e., $\Delta (x)=(x,x)$, for all $x\in\mathfrak g$. Then, by the universal property, $\Delta$ extends uniquely to an associative algebra morphism $\Delta:\mathcal U(\mathfrak g)\rightarrow\mathcal U(\mathfrak G)$, which, composed with the canonical isomorphism $\mathcal U(\mathfrak G)\simeq \mathcal U(\mathfrak g)\otimes\mathcal U(\mathfrak g)$, yields a linear map $\Delta:\mathcal U(\mathfrak g)\rightarrow\mathcal U(\mathfrak g)\otimes\mathcal U(\mathfrak g)$, defined by

\begin{eqnarray}
	{\lefteqn{\Delta(x_1\cdots x_n)=x_1\cdots x_n\otimes 1+1\otimes x_1\cdots x_n}} \nonumber\\
	& \qquad +\sum\limits_{k=1}^{n-1}\sum\limits_{\sigma\in\Sigma_{k,n-k}}x_{\sigma(1)}
	\cdots x_{\sigma(k)}\otimes x_{\sigma(k+1)}\cdots x_{\sigma (n)}\label{eq:coprodU},
\end{eqnarray}
where for each $k=1,\dots, n-1$, $\Sigma_{k,n-k}$ is the subgroup of the $(k,n-k)$ shuffles in $\Sigma_n$. Starting now from the trivial map $\mathfrak g\rightarrow 0$ and using again the universal property of $\mathcal U(\mathfrak g)$, one can define the co-unit map $\epsilon:\mathcal U(\mathfrak g)\rightarrow\mathbb F$.
%, also called the \emph{augmentation map}, i.e., the $\mathbb F$-linear map whose kernel is the so-called \emph{augmentation ideal} $\mathcal I=\cup_{k>0}\mathcal U_k(\mathfrak g)$. 
It is again the universal property of the enveloping algebra, that permits to show that $(\mathcal U(\mathfrak g),\Delta,\epsilon)$ is a co-algebra. 

On the other hand, the map $S:\mathfrak g\rightarrow\mathfrak g$, defined by $S(x)=-x$ for all $x\in\mathfrak g$, is a Lie algebra anti-homomorphism, which extends in a unique way to an associative algebra homomorphism $S:\mathcal U(\mathfrak g)\rightarrow\mathcal U(\mathfrak g)$, such that $S(x_{i_1}\cdots x_{i_n})=(-1)^n x_{i_n}\cdots x_{i_1}$ for each monomial $x_{i_1}\cdots x_{i_n}$, and it satisfies the antipode property. Finally, the proof of co-commutativity follows at once from the universal property of $\mathcal U(\mathfrak g)$ and noticing that the maps $(\Delta\otimes \mathrm{id})\circ\Delta$ and $(\mathrm{id}\otimes\Delta)\circ\Delta:\mathcal U(\mathfrak g)\rightarrow\mathcal U(\mathfrak g)\otimes\mathcal U(\mathfrak g)\otimes\mathcal U(\mathfrak g)$ are both obtained from the embeddings of $\mathfrak g$ into $\mathfrak g\oplus\mathfrak G\simeq\mathfrak g\oplus\mathfrak g\oplus\mathfrak g$ respectively $\mathfrak g\rightarrow\mathfrak G\oplus \mathfrak g\simeq\mathfrak g\oplus\mathfrak g\oplus\mathfrak g$.
\end{proof}

Note that every $x\in\mathfrak g=\mathcal U_1(\mathfrak g)$ is a primitive element. Furthermore, it can be shown that if $\xi\in\mathcal U(\mathfrak g)$ is primitive then $\xi\in\mathfrak g$. In other words, one can prove that $\mathcal P(\mathcal U(\mathfrak g))=\mathfrak g$. On the other hand, it is simple to see that in $\mathcal U(\mathfrak g)$ there are no group-like elements of degree greater than zero, i.e., $\mathcal G(\mathcal U(\mathfrak g))=\mathbb F$. To associate non-trivial group-like elements to $\mathcal U(\mathfrak g)$ one needs to consider instead of $\mathcal U(\mathfrak g)$ a suitable \emph{completion} of it. 

\begin{remark}
Since $\mathfrak g=\mathcal P(\mathcal U(\mathfrak g))$, every Lie polynomial is still a primitive element of $\mathcal U(\mathfrak g)$.
\end{remark}

\emph{Complete Hopf algebras.} We follow reference \cite{Quillen}. In what follows all algebras are unital and defined over the field $\mathbb F$. $A$ will be called an \emph{augmented} algebra, if it comes endowed with an algebra morphism $\epsilon: A\rightarrow\mathbb F$ called the \emph{augmentation map}. In this case its kernel $\text{ker}\,\epsilon$ will be called the \emph{augmentation ideal} and it will be denoted by $I$. 

\begin{example}\label{ex:unial}
$A=\mathscr{U}(\mathfrak g)$ is an example of augmented algebra. In fact the co-unit $\epsilon:\mathscr{U}(\mathfrak g)\rightarrow\mathbb F$ is an augmentation map and its kernel, $I=\cup_{k>0}\mathcal U_k(\mathfrak g)$, is the corresponding augmentation ideal.
\end{example}

A \emph{decreasing filtration} of $A$ is a decreasing sequence $A = F_0A \supset F_1A \supset \cdots$ of sub-vector spaces, such that $F_pA \cdot F_qA \subset F_{p+q}A$ and $\text{gr}\,A=\oplus_{n=0}^\infty F_nA/F_{n+1}A$ has a natural structure of a graded algebra. Note that for each $k$, $F_kA$ is a two-side ideal of $A$. We can now define the notion of a \emph{complete augmented algebra}. 

\begin{definition}\label{def:caa} A \emph{complete augmented algebra} is an augmented algebra $A$ endowed with a decreasing filtration $\{F_kA\}_{k\in\mathbb N}$ such that:

\begin{enumerate}
\item[1)] $F_1A=I$,

\item[2)] $\text{gr}\,A$ is generated as an algebra by $\text{gr}_1\,A$,

\item[3)] As an algebra, $A$ is the \emph{inverse limit} $A=\varprojlim\, A/F_nA$.\footnote{Let $(\mathcal I,\leq)$ be a \emph{ directed poset}. Recall that a pair $(\{A_i\}_{i\in \mathcal I},\{f_{ij}\}_{i,j\in \mathcal I})$ is called an \emph{inverse} or \emph{projective system} of sets over $\mathcal I$, if $A_i$ is  a set for each $i\in\mathcal I$, $f_{ij}:A_i\rightarrow A_j$ is a map defined for all $j\leq i$ such that $f_{ij}\circ f_{jk}=f_{ik}:A_i\rightarrow A_k$, every time the corresponding maps are defined and $f_{ii}=\operatorname{id}_{A_i}$. Then the \emph{inverse} or the \emph{projective limit} of the inverse system $(\{A_i\}_{i\in\mathcal I},\{f_{ij}\}_{i,j\in \mathcal I})$ is 
\begin{equation*}
\varprojlim A_i=\{\xi\in\prod_{i\in \mathcal I}A_i\,\vert\, f_{ij}(p_i(\xi))=p_j(\xi),\,\forall j\leq i\},
\end{equation*}
where, for each $i\in \mathcal I$, $p_i:\prod_{i\in \mathcal I} A_i\rightarrow A_i$ is the canonical projection. This definition is easily specialized to define the inverse limit in the category of algebras, co-algebras and Hopf algebras.}

\end{enumerate}
\end{definition}

\begin{example}\label{ex:caa1}
Let $A$ be an augmented algebra. Then $\hat A=\varprojlim A/I^n$ is a complete augmented algebra where, for each $n\geq 0$, $F_n\hat A=\widehat {I^n}=\varprojlim\, I^n/I^k$, $k\geq n$. It is worth to recall that $\hat{A}$ is also called the \emph{$I$-adic completion} of $A$. Note that in this case $F_nA:=I^n$ if $n\geq 1$ and $F_0A=A$ and the inverse system defining the completion is given by the data $(\{A_i\}_{i\in I},\{f_{ij}\}_{i,j\in\mathcal I})$ where $\mathcal I=\mathbb N$, $A_n=A/I^n$ and $f_{ij}:A_j\rightarrow A_i$ is the application that, for all $a\in A$, maps $a\,\text{mod}\, A_j$ to $a\,\text{mod}\, A_i$, for all $j\leq i$.

Since $A/I^n\simeq \hat A/\widehat{I^n}$, one has that $\text{gr}\,\hat A\simeq \text{gr}\,A=\oplus_{n\geq 0}I^n/I^{n+1}$, which implies that $\hat A$ satisfies property $2)$ in the definition above. Property $1)$ is clear from the definition of the filtration of $\hat A$, while Property $3)$ follows from the isomorphism $\hat{\hat A}\simeq\hat A$, for each $\hat A=\varprojlim A_n$, where $\hat{\hat A}=\varprojlim\widehat{A_n}$ and $\widehat{A_n}=\varprojlim A_k$, $k\geq n$.
In particular, taking $A=\mathscr{U}(\mathfrak g)$ and $I=\cup_{k>0}\mathscr{U}_{k}(\mathfrak g)$, see Example \ref{ex:unial}, one can define the complete augmented algebra
\begin{equation}
\hat{\mathscr{U}}(\mathfrak g)=\varprojlim\mathscr{U}(\mathfrak g)/I^n,\label{eq:comunial}
\end{equation}
which will be simply called in the following the \emph{completion} of $\mathscr{U}(\mathfrak g)$.
\end{example}

%\begin{remark}
%Note that both the class of augmented algebras and the one of complete augmented algebras form categories, whose morphisms are obvious. The former is denoted by $\mathcal A\mathcal A$, and the latter by $\mathcal C\mathcal A\mathcal A$. Furthermore, note that between $\mathcal A\mathcal A$ and $\mathcal C\mathcal A\mathcal A$ a forgetful functor is naturally defined. The functor from $\mathcal A\mathcal A$ to $\mathcal C\mathcal A\mathcal A$ described in the previous example is its left-adjoint.
%\end{remark}

Let $V=F_0V\supset F_1V\supset F_2V\supset\cdots$ be a filtered vector space and $\pi_n:F_nV\rightarrow \text{gr}_n\,V$ be the canonical surjection. If $W$ is another filtered vector space, then one can define a filtration on $V \otimes W$ declaring that $F_n(V\otimes W)=\sum_{i+j=n}F_iV \otimes F_jV \subset V \otimes W$, for all $n\geq 0$, where one identifies $F_iV \otimes F_jV$ with its image in $V \otimes W$ via the canonical injection. If $V$ and $W$ are complete, i.e. if $V=\varprojlim V/F_nV$ and $W=\varprojlim W/F_nW$, then we denote by $V\hat \otimes W$ the completion of $V\otimes W$ with respect to the filtrations defined above, and we denote with $x \hat \otimes y$ the image of $x\otimes y$ via the canonical morphism between $V \otimes W$ and $V \hat \otimes W$. Note that, since $F_{2n}(V\otimes W)\subset F_nV\otimes W+V\otimes F_nW\subset F_n(V\otimes W)$, one has that 
\[
	V\hat\otimes W=\varprojlim (V_n\otimes W_n),
\]
where, given the filtered vector space $V=F_0 V\supset F_1 V\supset F_2 V\supset \cdots $, $V_n=V/F_n V$.

\begin{definition}
The vector space $V \hat \otimes W$ so defined is called the \emph{complete tensor product} of the complete vector spaces $V$ and $W$.
\end{definition}

\begin{remark}\label{rem:remahopf} A couple of remarks are in order.
\begin{enumerate}
\item Let $V$ and W be two filtered vector spaces. Then the map $\rho: \text{gr}\,V\otimes\text{gr}\,W\rightarrow \text{gr}\,(V\otimes W)$, defined by $\rho(\pi_p x\otimes\pi_ q y)=\pi_{p+q}(x\otimes y)$, is an isomorphism, which, if $V$ and $W$ are complete, induces an isomorphism, still denoted by $\rho$, between $\text{gr}\,V\otimes\text{gr}\,W$ and $\text{gr}\,V\hat\otimes\,\text{gr}\,W$, and which takes $\pi_p x\otimes\pi_q y$ to $\pi_{p+q} (x\hat\otimes y)$, for all $p,q\in\mathbb N$ and for all $x\in V$ and $y\in W$. 

\item If $A$ and $A'$ are two complete augmented algebras, then $F_n(A\otimes A')$ is a filtration of $A\otimes A'$ and the corresponding completed tensor product $A\hat\otimes A'$ becomes a complete augmented algebra. The complete tensor product of complete algebras has the following property. If $A$ and $B$ are augmented algebras then, $\hat A\hat\otimes\hat B=\widehat{A\otimes B}$.
\end{enumerate}
\end{remark}
Finally we can introduce the following concept.

\begin{definition}
A \emph{complete Hopf algebra} $(H,m,i,\Delta,\epsilon,S)$ is a complete augmented algebra $(H,m,i)$, where $\Delta: H\rightarrow H \hat\otimes H$ and $S:H\rightarrow H$ are morphisms of complete augmented algebras, and $\epsilon:H\rightarrow\mathbb F$ is the augmentation map. These morphisms satisfy the same properties as in the usual definition of Hopf algebra, with the usual tensor product replaced by the complete tensor product.
\end{definition}

Note that $(H,m,i,\Delta,\epsilon,S)$ is co-commutative if $\tau \circ \Delta=\Delta$. Furthermore, to every Hopf algebra $(H,\hat m,\hat i,\hat\Delta,\hat\epsilon,\hat S)$ one can associate a complete Hopf algebra by considering $\hat H$ and $\hat\Delta:\hat H\rightarrow \widehat{H\otimes H}\simeq \hat H\hat\otimes\hat H$, see Remark \ref{rem:remahopf} above.

\begin{example}
Let $\mathfrak g$ be a finite dimensional Lie algebra. Then $\hat{\mathcal U}(\mathfrak g)$ carries a structure of complete Hopf algebra, see Example \ref{ex:caa1}.
\end{example}

Given a complete Hopf algebra $(H,\hat m,\hat i,\hat\Delta,\hat\epsilon,\hat S)$, one can define:
\begin{eqnarray*}
	\mathcal P(H)&:=&\{x\in I_H\,\vert\,\Delta x=x\hat\otimes 1+1\hat\otimes x\}\\
	\mathcal G(H)&:=&\{x\in 1+I_H\,\vert\,\Delta x=x\hat\otimes x\},
\end{eqnarray*}
i.e. the set of \emph{primitive}, respectively, of  \emph{group-like} elements of $H$. 

Note that if $A$ is a complete augmented algebra and if $x\in A$, $\operatorname{e}^x=\sum_{n\geq 0}\frac{x^n}{n!}$ belongs to $A$. This follows from Property $3)$ in Definition \ref{def:caa}, noticing that, if $S_n=\sum_{k=0}^n\frac{x^k}{k!}$ for all $n\geq 0$, then the sequence $\{S_n\}_{n\in\mathbb N}$ is \emph{convergent}, since it is \emph{Cauchy}.\footnote{Recall that if $M$ is a $\mathbb Z$-module endowed with a decreasing filtration, $M=M_0\supset M_1\supset M_2\supset\cdots$, then a sequence $(x_k)_{k\in\mathbb N}$ is called a \emph{Cauchy sequence} if for each $r$ there exists $N_r$, such that, if $n,m>N_r$, then $x_n-x_m \in M_r$. This amounts to saying, that if $n,m$ are sufficiently large, then $x_n+M_r=x_m+M_r$. This implies that $(x_k)_{k\in\mathbb N}$ is a coherent sequence, i.e., it belongs to $\hat M=\varprojlim M/M_k$. In other words, every Cauchy sequence is convergent in $\hat M$. These considerations can be extended verbatim to the case of complete augmented algebras.} 

Let $(H,m,i,\Delta,\epsilon,S)$ be a complete Hopf algebra. Then

\begin{proposition}[\cite{Quillen}]
	$x\in\mathcal P(H)\iff e^x\in\mathcal G(H)$.
\end{proposition}

\begin{proof}
In fact  $x\in\mathcal P(H)\iff \Delta x=x\hat\otimes 1+1\hat\otimes x\iff \operatorname{e}^{\Delta x}=\operatorname{e}^{x\hat\otimes 1+1\hat\otimes x}$ and, since $(x\hat\otimes 1)(1\hat\otimes x)-(1\hat\otimes x)(x\hat\otimes 1)=0$, one has that
\[
	\operatorname{e}^{x\hat\otimes 1+1\hat\otimes x}\\
		=\operatorname{e}^{x\hat\otimes 1}\cdot \operatorname{e}^{1\hat\otimes x}
		=(\operatorname{e}^x\hat\otimes 1)(1\hat\otimes \operatorname{e}^x)
		=\operatorname{e}^x\hat\otimes \operatorname{e}^x,
\]
which implies the statement since $\Delta \operatorname{e}^x=\operatorname{e}^{\Delta x}$.
\end{proof}

\begin{corollary}
The exponential map $\exp:\mathcal P(H)\rightarrow\mathcal G(H)$, $\exp: x\rightarrow \operatorname{e}^x$, defines an isomorphism of sets, whose inverse is the \emph{logarithmic series}, defined by  $\log (1+x)=\sum_{n\geq 1}(-1)^{n-1}\frac{x^n}{n}$, $\forall x\in I_H$.
\end{corollary}

\begin{example}
Let $\mathfrak g$ be a finite dimensional Lie algebra and let $\hat{\mathcal U}(\mathfrak g)$ be the corresponding complete universal enveloping algebra, see Example \ref{ex:caa1}. Then, given $\xi\in\hat{\mathcal U}(\mathfrak g)$, $\operatorname{e}^\xi$ is a group-like element \emph{if and only if} $\xi\in\mathfrak g$. Moreover, from the previous corollary, one knows that if $x\in\cup_{k\geq 1}\hat{\mathcal U}_k(\mathfrak g)$ and $y=1+x$ such that $\Delta y=y\hat\otimes y$, then there exists $z\in\mathcal P(\hat{\mathcal U}(\mathfrak g))$ such that $y=\operatorname{e}^z$, see Example \ref{ex:caa1}.
\end{example}

We conclude this part by noticing that on every complete Hopf algebra, both the Lie algebra of primitive elements and the group of  group-like elements inherit a filtration. More precisely one has the

\begin{proposition}[\cite{Quillen}] If for all $k\geq 0$
\begin{eqnarray*}
	F_k\mathcal G(H)&=&\{x\in\mathcal G(H)\,\vert\,x-1\in F_kH\}\\
	F_k\mathcal P(H)&=&\mathcal P(H)\cap F_kH
\end{eqnarray*}
then $\{F_k\mathcal G(H)\}$ and $F_k\mathcal P(H)$ are filtrations of $\mathcal G(H)$ respectively $\mathcal P(H)$. Moreover:
\begin{enumerate}
\item The exponential map induces a set bijection $\text{gr}\,\mathcal P(H)\rightarrow\text{gr}\,\mathcal G(H)$.
\item 
\begin{eqnarray*}
	\mathcal P(H)&\simeq & \varprojlim \mathcal P(H)/F_k\mathcal P(H)\\
	\mathcal G(H)&\simeq & \varprojlim\mathcal G(H)/F_k\mathcal G(H).
\end{eqnarray*}
\end{enumerate}
\end{proposition}

\begin{example}%[see Subsection \ref{ss:bch}]
\label{ex:bchi}
If $H=\hat{\mathcal U}(\mathfrak g)$, the previous proposition implies that, for all $x,y\in\mathfrak g$, $\operatorname{BCH}(x,y)\in\mathcal P\big(\hat{\mathcal U}(\mathfrak g)\big)$, i.e., $\operatorname{BCH}(x,y)$ is \emph{convergent} for all $x,y\in\mathfrak g$. The proof of this statement is based on two observations. First, $\operatorname{BCH}(x,y)$ is a Lie series in $x,y$ that, seen as an element of $\hat{\mathcal U}(\mathfrak g)$ can be written as
\begin{equation}
	\operatorname{BCH}(x,y)=\sum_{n=0}^\infty z_n(x,y),
\end{equation}
where, for each $n\geq 0$, $z_n(x,y)$ is the \emph{non-commutative homogeneous polynomial} of degree $n$ in $x,y$, obtained from the corresponding Lie polynomial in $\operatorname{BCH}(x,y)$ using the relation $[x,y]=xy-yx$. Second, the sequence $\{S_n\}_{n\geq 0}$, where $S_n=\sum_{k=0}^nz_k(x,y)$ is \emph{Cauchy}.
\end{example}

%%%%%%%%%%%%%%%%%%%%%%%%%%%%%

\section{Pre- and post--Lie algebras}
\label{sec:lieapospre}

In this section we will introduce the definitions and the main properties of a pre- and post-Lie algebra, stressing the relevance of these notions in the theory of smooth manifolds and Lie groups.

An algebra $(A,\cdot)$ is called {\it{Lie admissible}} if the bracket $[\cdot,\cdot]:A\otimes A\rightarrow A$ defined by anti-symmetrization, $[x,y]:=x\cdot y-y\cdot x$, for all $x,y\in A$, is a Lie bracket, i.e., if it satisfies the Jacobi identity. For example every associative algebra is Lie admissible.  Given $(A,\cdot)$, let 
\begin{equation}
\label{eq:associator}
	a_\cdot(x,y,z):=(x\cdot y)\cdot z-x\cdot (y\cdot z),\quad \forall x,y,z\in A
\end{equation}
be the \emph{associator} defined for the product $\cdot$. Note that $(A,\cdot)$ is associative if and only if $a_\cdot(x,y,z)=0$, for all $x,y,z\in A$. In the next two subsections the notions of pre- and post-Lie algebras are introduced. Such algebras are rather natural from the viewpoint of geometry. Moreover, later we will see that they are closely related to solutions of classical Yang--Baxter equations.

%%%%%%%%%%%%%%%%%%%%%%%%%%%%%

\subsection{Pre-Lie algebra}
\label{ssect:pre-lie}

Weakening the condition $a_\cdot(x,y,z)=0$, one arrives at a class of Lie admissible algebras, called pre-Lie algebras, which is more general than that of associative algebras. 
For a nice and general introduction to pre-Lie algebras see \cite{Burde} and \cite{Manchon} and the references therein.

\begin{definition}\label{def:pre-Lie}
$(A,\cdot)$ is a \emph{left pre-Lie algebra} if, for all $x,y,z \in A$
\begin{equation}
	a_{\cdot}(x,y,z)=a_\cdot(y,x,z).
	\label{eq:lpre-Lie}
\end{equation}
\end{definition}

Note that together with the notion of left pre-Lie algebra one can introduce that of a \emph{right pre-Lie algebra} where condition \eqref{eq:lpre-Lie} is traded for $a_\cdot(x,y,z)=a_\cdot(x,z,y)$, for all $x,y,z\in A$. The notions of right and left pre-Lie algebras are equivalent. Indeed, if $(A,\cdot)$ is a left (right) pre-Lie algebra, then $(A,\cdot^{\text{op}})$ is a right (left) pre-Lie algebra, where $x \cdot^{\text{op}} y = y \cdot x$. For this reason, from now on, we will focus on the case of left pre-Lie algebras, which will be called simply pre-Lie algebras. 

Let $(A,\cdot)$ be a pre-Lie algebra and let $\nabla: A\rightarrow\operatorname{End}(A)$ be the morphism defined by $\nabla(x):=\nabla_x: A\rightarrow A,\,\nabla_xy=x\cdot y$. Then, the pre-Lie condition implies that 
\[
	[\nabla_x,\nabla_y]=\nabla_{[x,y]},\quad \forall x,y\in A,
\]
that is, $\nabla:A\rightarrow\operatorname{End}(A)$ is a morphism of Lie algebras, where the Lie brackets of $A$ and of $\operatorname{End}(A)$ are defined by skew-symmetrizing the pre-Lie product of $A$, respectively, the associative product of $\operatorname{End}(A)$. It is worth to recall that the Lie algebra structure on $A$ defined by skew-symmetrizing the pre-Lie product is called \emph{subordinate} to it, or equivalently, that the pre-Lie algebra structure is \emph{compatible} with the Lie algebra structure so defined. Furthermore, defining for $x,y \in A$ the expression $\mathrm{T}(x,y)=\nabla_xy-\nabla_yx-[x,y]$, it is obvious from the definition that $\mathrm{T}(x,y)=0$. From these observations, as it was already remarked in the Introduction, a source of examples of pre-Lie algebras can be found looking at locally flat manifolds, i.e. manifolds endowed with a linear flat and torsion free connection, see for example \cite{CK,Vinberg,DM,Medina} and references therein.
%From these observations it follows that a source of examples of pre-Lie algebras can be found by looking at manifolds carrying a \emph{flat} and \emph{torsion free} (linear) connection. Recall that a linear connection on a manifold $M$ is a bilinear map $\nabla:\mathfrak X_M\times\mathfrak X_M \rightarrow\mathfrak X_M$, where $\mathfrak X_M$ is the Lie algebra of vector fields on $M$, such that (i) $\nabla_{fX}Y=f\nabla_XY$ and (ii) $\nabla_X(fY)=\langle df,X\rangle Y+f\nabla_XY$, for all $X,Y\in\mathfrak X_M$ and $f\in C^\infty(M)$. The torsion $\mathrm{T}_\nabla$ and the curvature $\mathrm{R}_\nabla$ of $\nabla$ are the tensors defined by $\mathrm{T}_\nabla(X,Y)=\nabla_XY-\nabla_YX-[X,Y]$ and $\mathrm{R}_\nabla(X,Y)=\nabla_X\nabla_Y-\nabla_Y\nabla_X-\nabla_{[X,Y]}$ for all $X,Y\in\mathfrak X_M$. Note that given $X,Y\in\mathfrak X_M$, $\mathrm{T}_\nabla(X,Y)\in\mathfrak X_M$, while $\mathrm{R}_\nabla(X,Y)\in\operatorname{End}(TM)$. A linear connection $\nabla$ is called \emph{torsion free} if $\mathrm{T}_\nabla(X,Y)=0$ for all $X,Y \in \mathfrak X_M$, while it is called \emph{flat} (or with $0$-curvature) if $\mathrm{R}_\nabla(X,Y)=0$, for all $X,Y \in \mathfrak X_M$. If $\nabla$ is a (linear) torsion-free and flat connection on $M$, the product $\cdot:\mathfrak X_M \otimes\mathfrak X_M \rightarrow \mathfrak X_M$ defined by $X \cdot Y=\nabla_XY$ is pre-Lie. 
It is worth to note that a $n$-dimensional manifold $M$ admits a (linear) torsion-free and flat connection \emph{if and only if} it admits an \emph{affine structure}, i.e., a (maximal) atlas whose transition functions are constant and take values in $\operatorname{GL}_n(\mathbb F)\ltimes\mathbb F^n$. In fact, given such a $\nabla$, for all $m\in M$ one can find an open neighborhood $m\in U$ and $X_1,\dots,X_n\in\mathfrak X_M(U)$ a local frame for $TM$ such that $\nabla_{X_i}X_j=0$ for all $i,j=1,\dots, n$. Then, if $\alpha_1,\dots,\alpha_n$ is the \emph{dual} local frame, one has that $d\alpha_i=0$ for all $i$. Indeed, one verifies that 
\[
	d\alpha_i(X_j,X_k)=X_j\alpha_i(X_k)-X_k\alpha_i(X_j)-\alpha_i([X_j,X_k])
				    =-\alpha_i([X_j,X_k])=0,
\]
since $\alpha_i(X_j)=\delta_{ij}$, and $\alpha_i([X_j,X_k])=0$ due to the fact that $[X_j,X_k]=\nabla_{X_j}X_k - \nabla_{X_k}X_j -\mathrm{T}_\nabla(X_j,X_k)=0$. Then, on a neighborhood $V$ of $m \in M$, eventually contained in $U$, one can find $x_1,\dots,x_n\in C_M^\infty(V)$, such that $dx_i=\alpha_i$, for all $i=1,\dots,n$. The local functions $x_1,\dots,x_n$ so defined form a system of local coordinates on (a neighborhood of $M$ eventually smaller than) $V$. In this way one defines a system of local coordinates on $M$ such that, if $(V,x_1,\dots,x_n)$ and $(W,y_1,\dots,y_n)$ are two overlapping local charts, $dy_i=\sum_{k=1}^n T^k_idx_k$, where $T^k_i$, $k,i=1,\dots,n$, are the transition functions between the two local charts. Then $0=\nabla dy_i=\sum_{k=1}^ndT^k_i\wedge dx_k$, which implies that $dT^k_i=0$, for all $i,k=1,\dots,n$. From this it follows that the functions $T^k_i$ are (locally) constant, i.e., $T^k_i=\frac{\partial y_i}{\partial x_k}\in\mathbb F$ for all $i,k=1,\dots,n$, which implies that $y_i=\sum_{k=1}^nT^k_ix_k+C_i$, $C_i\in\mathbb F$, proving the statement. To prove that to every affine structure corresponds a flat and torsion-free linear connection one should follow backward all the steps of the argument just presented. A class of examples of manifolds endowed with an affine structure is presented in the following example.

\begin{example}[Invariant affine structures on Lie groups, see \cite{DM,Medina}]\label{ex:simplecticLieG} First, recall that given a vector field $X$ on a smooth manifold $M$, one can define the \emph{Lie derivative} $\mathscr L_X$ and the \emph{interior product} $i_X$, which are derivations of the full tensor algebra of $M$. Once restricted to the exterior algebra defined by $T^\ast M$, they become derivations of degree $0$ and degree $-1$, respectively. They are related by the formula $\mathscr L_X=i_X\circ d+d\circ i_X,$ where $d$ is Cartan's differential. In particular, given a differential $k$-form $\eta\in\Omega^k(M)$, then $\mathscr L_X \eta \in \Omega^k(M)$ and for all $m\in M$	
\[
	(\mathscr L_X\eta)_m=\left.\frac{d}{dt}\right\vert_{t=0}(\varphi_{X,t}^\ast\eta)_m,
\]
where $\{\varphi_{X,t}\}_{t\in\mathbb R}$ is the \emph{local $1$-parameter group of diffeomorphisms} defined by $X$.

A \emph{symplectic form} on a manifold $M$ is a  $2$-form which is \emph{closed} and \emph{non-degenerate}. The pair $(M,\omega)$ is called a \emph{symplectic manifold}. Given a symplectic manifold $(M,\omega)$ and Lie group $G$ acting on $M$ via $\varphi:G\times M\rightarrow M$, $\omega$ will be called $G$-invariant if $\varphi_g^\ast\omega=\omega$, for all $g\in G$. In particular,  a \emph{symplectic Lie group} is a pair $(G,\omega)$ consisting of a Lie group and a \emph{left-invariant symplectic form}, i.e., a symplectic form invariant with respect to \emph{left-translations}. Let $(G,\omega)$ be a symplectic Lie group and let $x,y$ be elements in the Lie algebra $\mathfrak g$ of $G$. Then, %\footnote {denoting by $\mathscr L_X$ and the by $i_X$ the \emph{Lie derivative} and, respectively, the \emph{interior product} defined by the vector field $X$, one has} 
$\mathscr{L}_{X_x}(i_{X_y}\omega)$ is a a left-invariant $1$-form on $G$, to which corresponds the unique left-invariant vector field $X_z$, such that $-i_{X_z}\omega=\mathscr{L}_{X_x}i_{X_y}\omega.$ Note that, since $\mathscr{L}_{X_x}i_{X_y}\omega=i_{X_x}di_{X_y}\omega$, for each $f\in C^\infty(G)$ and for all $x,y\in\mathfrak g$, 
\[
	\mathscr{L}_{fX_x}i_{X_y}\omega
	=f\mathscr{L}_{X_x}i_{X_y}\omega\quad\text{and}\quad\mathscr{L}_{X_x}i_{fX_y}\omega
	=\langle df,X_x\rangle i_{X_y}\omega+f\mathscr{L}_{X_x}i_{X_y}\omega.
\]
In other words, defining $\nabla_{X_x}X_y$ as the unique left-invariant vector field such that
\begin{equation}
	-i_{\nabla_{X_x}X_y}\omega=\mathscr{L}_{X_x}i_{X_y}\omega,\label{eq:defcon}
\end{equation}
for all $X_x,X_y$ left-invariant vector fields, one sees that $\nabla$ admits a unique extension to a $G$-invariant linear connection on $G$. If one denotes still with $\nabla$ this connection, then $\nabla$ is \emph{flat} and \emph{torsion-free}. To prove this statement it suffices to show that $T_\nabla(X_x,X_y)=0$ and $R_\nabla(X_x,X_y)=0$ for all $x,y\in\mathfrak g$. Let us compute
\begin{eqnarray*}
	&&\mathscr{L}_{X_x}i_{X_y}\omega-\mathscr{L}_{X_y}i_{X_x}\omega-i_{[X_x,X_y]}\omega\\
	&=&\mathscr{L}_{X_x}i_{X_y}\omega-\mathscr{L}_{X_y}i_{X_x}\omega-\mathscr{L}_{X_x}i_{X_y}
	\omega+i_{X_y}\mathscr{L}_{X_x}\omega\\
	&=&-i_{X_y}di_{X_x}\omega+i_{X_y}di_{X_x}\omega\\
	&=&0,
\end{eqnarray*}
where we used that $d\omega=0$ and that $\mathscr{L}_X\alpha=di_X\alpha+i_Xd\alpha$, for all forms $\alpha$ and all vector fields $X$. Since $T_\nabla(X_x,X_y)$ is the unique left-invariant vector field such that $-i_{T_\nabla(X_x,X_y)}\omega=\mathscr{L}_{X_x}i_{X_y}\omega-\mathscr{L}_{X_y}i_{X_x}\omega-i_{[X_x,X_y]}\omega$, the non-degeneracy of $\omega$ forces $T_\nabla(X_x,X_y)=0$. Let us now observe that if $x,y,z\in\mathfrak g$ then $\nabla_{X_x}\nabla_{X_y}X_z$ and  $\nabla_{[X_x,X_y]}X_z$ are the unique left-invariant vector fields such that 
\[
	-i_{\nabla_{X_x}\nabla_{X_y}X_z}\omega=i_{X_x}d\big(i_{X_y}d(i_{X_z}\omega)\big)
\]
and, respectively, 
\[
	-i_{\nabla_{[X_x,X_y]}X_z}\omega=i_{[X_x,X_y]}di_{X_z}\omega.
\]
One sees that
\[
	i_{X_x}d\big(i_{X_y}d(i_{X_z}\omega)\big)-i_{X_y}d\big(i_{X_x}d(i_{X_z}\omega)\big)
	=i_{[X_x,X_y]}di_{X_z}\omega,\quad \forall x,y,z\,\mathfrak g,
\]
which again, by the non-degeneracy of $\omega$, is equivalent to
\[
	\nabla_{X_x}\nabla_{X_y}X_z-\nabla_{X_x}\nabla_{X_y}X_z
	=\nabla_{[X_x,X_y]}X_z,\quad \forall x,y,z\in\,\mathfrak g,
\]
proving the flatness of $\nabla$. In other words we have shown that

\begin{theorem}[\cite{DM}]
Every symplectic Lie group $(G,\omega)$ admits an affine structure.
\end{theorem}

In particular, since for all $x,y\in\mathfrak g$ there exists a (unique) $z\in\mathfrak g$ such that $\nabla_{X_x}X_y=X_z$, the underlying vector space of the Lie algebra $\mathfrak g$ results being endowed with a product $\cdot:\mathfrak g\otimes\mathfrak g\rightarrow\mathfrak g$ defined by 
\begin{equation}
	x\cdot y=z, \quad \forall\,x,y,z \quad\text{s.t.}\quad X_z
	=\nabla_{X_x}X_y.\label{eq:pre-Lies}
\end{equation}
Since $\nabla$ is flat and torsion-free, it is easy to show that $\cdot$ is a \emph{pre-Lie} product on the vector space underlying $\mathfrak g$ and that, for all $x,y\in\mathfrak g$, $x\cdot y-y\cdot x=[x,y]$. In other words

\begin{corollary}
The Lie algebra of a symplectic Lie group is subordinate to the pre-Lie product defined in \eqref{eq:pre-Lies}.
\end{corollary}

Finally, since $d\omega=0$, $\omega_e \in \mathscr Z^2(\mathfrak g,\mathbb F)$, where $\mathscr Z^2(\mathfrak g,\mathbb F)$ is the group of $2$-cocycles of $\mathfrak g$ with values in the trivial $\mathfrak g$-module $\mathbb F$, with respect to the cohomology of Cartan--Eilenberg of $\mathfrak g$ with coefficients in the trivial $\mathfrak g$-module $\mathbb F$. See for example~\cite{Fuks}. Hence, $\omega_e\in\operatorname{Hom}_{\mathbb F}(\Lambda^2\mathfrak g,\mathbb F)$ such that
\[
	\omega_e(x,[y,z])+\omega_e(z,[x,y])+\omega_e(y,[z,x])=0,\quad \forall x,y,z\in\mathfrak g,
\]
and since $\omega$ is non-degenerate, $\omega_e$ is also non-degenerate. On the other hand, if $\eta\in \mathscr Z^2(\mathfrak g,\mathbb F)$ is non-degenerate, it defines a unique left-invariant symplectic form $\omega_\eta$ on $G$ via the formula: 
\[
	{\omega_\eta}_g=(L_g)^\ast_e\eta,\quad \forall g\in G.
\]  
In other words, the left-invariant symplectic forms on $G$ are in one-to-one correspondence with the non-degenerate elements of $\mathscr Z^2(\mathfrak g,\mathbb F)$. See also Subsection \ref{ss:YBPP} for a more general approach to this kind of structures.
\end{example}

%%%%%%%%%%%%%%%%%%%%%%%%%%%%%

\subsection{Post-Lie algebra}
\label{ssect:post-lie}

The second class of algebras playing an central role in the present work is introduced in the following definition. 

\begin{definition}[\cite{Vallette,LMK1}] \label{def:post-Lie} 
Let $(\mathfrak g, [\cdot,\cdot])$ be a Lie algebra, and let $\triangleright  : {\mathfrak g} \otimes {\mathfrak g} \rightarrow \mathfrak g$ be a binary product such that for all $x,y,z \in \mathfrak g$
\begin{equation}
\label{post-Lie1}
	x \triangleright [y,z] = [x\triangleright y , z] + [y , x \triangleright z],
\end{equation}
and
\begin{equation}
\label{post-Lie2}
	[x,y] \triangleright z = {\rm{a}}_{\triangleright  }(x,y,z) - {\rm{a}}_{\triangleright  }(y,x,z).
\end{equation}
Then $(\mathfrak{g}, [\cdot,\cdot], \triangleright)$ is called a \emph{left post-Lie algebra}. 
\end{definition}

Relation \eqref{post-Lie1} implies that for every left post-Lie algebra the natural linear map $d_ \triangleright: \mathfrak g \rightarrow\operatorname{End}_{\mathbb F}(\mathfrak g)$ defined by $d_ \triangleright(x)(y) \rightarrow x \triangleright y$ takes values in the derivations of the Lie algebra $(\mathfrak g,[\cdot,\cdot])$.

Together with the notion of left post-Lie algebra one can introduce that of \emph{right post-Lie algebra} $(\mathfrak g,[\cdot,\cdot],\triangleleft)$. Also in this case $(\mathfrak g,[\cdot,\cdot])$ is a Lie algebra and $\triangleleft:\mathfrak g\otimes\mathfrak g\rightarrow\mathfrak g$ is a binary product such that for each $x\in\mathfrak g$, $d_ \triangleleft(x)(y)=x \triangleleft y$ is a derivation of $(\mathfrak g,[\cdot,\cdot])$ and
the analogue of \eqref{post-Lie2} is 
\[
	[x,y] \triangleleft z = {\rm{a}}_{\triangleleft}(y,x,z) - {\rm{a}}_{\triangleleft}(x,y,z),\quad \forall x,y,z\in\mathfrak g.
\]

\begin{proposition}[\cite{LMK1}]\label{pro:postL1}
If $(\mathfrak{g}, [\cdot,\cdot], \triangleright)$ is a left post-Lie algebra, then $(\mathfrak g,[\cdot,\cdot],\triangleleft)$, where
\[
	x \triangleleft y := x \triangleright y -[x,y]
\]
is a right post-Lie algebra.	
\end{proposition}
\begin{proof}
First, we show that 
\allowdisplaybreaks
\begin{eqnarray*}
	x\triangleleft [y,z]&=&x\triangleright [y,z]-[x,[y,z]]\\
				&=&[x\triangleright y,z]+[y,x\triangleright z]-[[x,y],z]-[y,[x,z]]\\
				&=&[x\triangleright y-[x,y],z]+[y,x\triangleright z-[x,z]]\\
				&=&[x\triangleleft y,z]+[y,x\triangleleft z].
\end{eqnarray*}
From
\begin{equation}
	[x,y]\triangleleft z=[x,y]\triangleright z-[[x,y],z],\label{eq:a}
\end{equation}
and
\allowdisplaybreaks
\begin{eqnarray}
	(y\triangleleft x)\triangleleft z
	&=&(y\triangleright x)\triangleright z-[y\triangleright x,z]-
	[y,x]\triangleright z+[[y,x],z] \label{eq:b}\\
	y\triangleleft (x\triangleleft z)
	&=&y\triangleright (x\triangleright z)-[y,x\triangleright z]-y\triangleright [x,z]+[y,[x,z]].\label{eq:c}
\end{eqnarray}
one deduces that
\[
	a_{\triangleleft}(y,x,z)-a_{\triangleleft}(x,y,z)=[x,y]\triangleright z-[[x,y],z],
\]
which is what we needed to show, see formula \eqref{eq:a}.
\end{proof}

Moreover, though post-Lie algebras are not Lie-admissible, one can prove the following proposition. 
  
\begin{proposition}[\cite{LMK1}]
\label{prop:post-lie}
Let $(\mathfrak g, [\cdot, \cdot], \triangleright)$ be a left post-Lie algebra. The bracket
\begin{equation}
\label{post-Lie3}
	\llbracket x,y \rrbracket := x \triangleright y - y \triangleright x - [x,y]
\end{equation}
satisfies the Jacobi identity for all $x, y \in \mathfrak g$, and it defines on $\mathfrak g$ the structure of a Lie algebra. 
\end{proposition}

\begin{proof}
It follows from a direct computation using the identities \eqref{post-Lie1} and \eqref{post-Lie2}.
\end{proof}

In particular, as consequence of the previous result one has 

\begin{corollary} Given a left post-Lie algebra $(\mathfrak g,[\cdot,\cdot],\triangleleft)$, the product $\succ:\mathfrak g\otimes\mathfrak g\rightarrow\mathfrak g$, defined by
\[
	x \succ y := x \triangleright y + \frac{1}{2}[x,y],\quad \forall x,y\in\mathfrak g
\] 
defines on $\mathfrak g$ the structure of Lie admissible algebra.	
\end{corollary}

Clearly, both the proposition and the corollary can be easily adapted to the case of right post-Lie algebra.

\begin{remark}\label{rem:notrem} A few remarks are in order. 

\begin{enumerate}
\item From now on, given a post-Lie algebra $(\mathfrak g,\triangleright,[\cdot,\cdot])$, we will denote by $\mathfrak g$ the Lie algebra with bracket $[\cdot,\cdot]$ and by $\overline{\mathfrak g}$ the Lie algebra with bracket $\llbracket\cdot,\cdot\rrbracket$.

\item Pre- and post-Lie algebras are important in the theory of numerical methods for differential equations. We refer the reader to \cite{Cartier11, ChaLiv, EFLMK,LMK1, Manchon} for background and details.

\item It is worth noting that if $(\mathfrak g,\triangleright,[\cdot,\cdot])$ is an \emph{abelian} left post-Lie algebra, i.e., $[\cdot,\cdot] \equiv 0$, then it reduces to the pre-Lie algebra $(\mathfrak g,\triangleright)$, whose underlying Lie algebra is $(\mathfrak g,\llbracket\cdot,\cdot\rrbracket)$, see \eqref{post-Lie2} and Definition \ref{def:pre-Lie}.
\end{enumerate}
\end{remark}
%\newpage
As for the case of pre-Lie algebras, differential geometry is a natural place to look for examples of post-Lie algebras, see for example \cite{BD,BDV,ZBG}. This is based on the well known result, see \cite{KobayashiNomizu} for example, saying that if $\nabla$ is a linear connection on $M$ then
%\newpage
\begin{proposition}\footnote{Formula \eqref{eq:bianchi} is known as Bianchi's $1$st identity. Among many other identities fulfilled by the covariant derivatives of the torsion and curvature of a linear connection, the so-called Bianchi's $2$nd identity is worth to recall:
\[
	\sum_{\circlearrowleft}\big((\nabla_X \mathrm{R})(Y,Z)+\mathrm{R}(\mathrm{T}(X,Y),Z)\big)
	=0,\quad \forall X,Y,Z\in\mathfrak X_M .
\]
}
\begin{equation}
	\sum_{\circlearrowleft}(\mathrm{R}(X,Y)Z-\mathrm{T}(\mathrm{T}(X,Y),Z)-(\nabla_X\mathrm{T})(Y,Z))=0,\label{eq:bianchi}
\end{equation}
for all $X,Y,Z\in\mathfrak X_M$
\end{proposition}

\begin{proof}
Since all the terms in \eqref{eq:bianchi} are tensors, it suffices to prove it for $X=\partial_i$, $Y=\partial_j$ and $Z=\partial_k$ where $\partial_i,\partial_j$ and $\partial_k$ are elements of a local frame. The formula follows now by a direct computation, noticing that $[\partial_i,\partial_j]=[\partial_i,\partial_k]=[\partial_j,\partial_k]=0$ and that $(\nabla_X \mathrm{T})(Y,Z)=\nabla_X\mathrm{T}(Y,Z)-\mathrm{T}(\nabla_XY,Z)-\mathrm{T}(Y,\nabla_XZ)$, for all $X,Y,Z\in\mathfrak X_M$.
\end{proof}

Then, if $\nabla$ is flat and has constant torsion, this formula implies that $[\cdot,\cdot]_\mathrm{T}: \mathfrak X_M\times\mathfrak X_M\rightarrow\mathfrak X_M$, defined by $[X,Y]_\mathrm{T}={\mathrm T}(X,Y)$, for all $X,Y\in\mathfrak X_M$ is a Lie bracket on $\mathfrak X_M$. In particular, defining $X\triangleright Y:=\nabla_XY$ for all $X,Y\in\mathfrak X_M$, then
\begin{equation*}
	X\triangleright [Y,Z]_\mathrm{T}=\nabla_X\mathrm{T}(Y,Z)
	=\mathrm{T}(\nabla_XY,Z)+\mathrm{T}(Y,\nabla_XZ)
	=[X\triangleright Y,Z]_\mathrm{T}+[Y,X\triangleright Z]_\mathrm{T}\label{eq:01}
\end{equation*}
and
\begin{eqnarray*}
	[X,Y]_\mathrm{T}\triangleright Z
	&=&\nabla_{\mathrm{T}(X,Y)}Z=\nabla_{\nabla_XY}Z-\nabla_{\nabla_YX}Z-\nabla_{[X,Y]}Z\\
	&=& \nabla_{\nabla_XY}Z-\nabla_{\nabla_YX}Z-\nabla_{[X,Y]}Z\\
	&=& \nabla_{\nabla_XY}Z-\nabla_{\nabla_YX}Z-\nabla_X\nabla_YZ+\nabla_Y\nabla_XZ\\
	&=& (X\triangleright Y)\triangleright Z-(Y\triangleright X)\triangleright Z
		-X\triangleright(Y\triangleright X)+Y\triangleright (X\triangleright Z)\\
	&=& a_{\triangleright}(X,Y,Z)-a_{\triangleright}(Y,X,Z),
\end{eqnarray*}
for all $X,Y,Z\in\mathfrak X_M$. In the second equality we used $\mathrm{R}_\nabla=0$. Moreover
\begin{eqnarray*}
	&&[X,Y]_\mathrm{T}=\mathrm{T}(X,Y)
	=\nabla_XY-\nabla_YX-[X,Y]
	=X\triangleright Y-Y\triangleright X-[X,Y].
\end{eqnarray*}
Summarizing, under the assumptions on the linear connection $\nabla$, one has that:
\begin{eqnarray*}
	&&X\triangleright [Y,Z]_\mathrm{T}
	=[X\triangleright Y,Z]_\mathrm{T}+[Y,X\triangleright Z]_\mathrm{T}\\ 
	&&[X,Y]_\mathrm{T}\triangleright Z=a_{\triangleright}(X,Y,Z)-a_{\triangleright}(Y,X,Z)\\
	&&[X,Y]=X\triangleright Y-Y\triangleright X-[X,Y]_\mathrm{T},
\end{eqnarray*}
for all $X,Y, Z\in\mathfrak X_M$. In other words

\begin{proposition}{\rm{\cite{LMK1}}}\label{pro:ge}
If $\nabla$ is a flat linear connection on the manifold $M$, with constant-torsion, then $(\mathfrak X_M,\triangleright,[\cdot,\cdot]_T)$ is a left post-Lie algebra.
\end{proposition}

\begin{remark} A few remarks are in order.
\begin{enumerate}
	\item Note that in the previous proposition the Lie-Jacobi bracket between vector fields, plays the role of the Lie bracket $\llbracket\cdot,\cdot\rrbracket$ in the post-Lie structure, while the role of the bracket $[\cdot,\cdot]$ is taken by $[\cdot,\cdot]_\mathrm{T}$, i.e., the one defined by the torsion tensor.
	
	\item If one had defined $[X,Y]_\mathrm{T}=-\mathrm{T}(X,Y)$ and $X\triangleleft Y=\nabla_XY$, which is the same product one has in the previous proposition, then $(\mathfrak X_M,\triangleleft,[\cdot,\cdot]_\mathrm{T})$ is a right post-Lie algebra.
\end{enumerate}
\end{remark}

At this point, it is worth recalling a classical result from differential geometry due to Cartan and Schouten. See references \cite{C-S} and \cite{PostGeo}. Let $G$ be a Lie group and $\mathfrak g$ its corresponding Lie algebra. A linear connection $\nabla$ on $G$ is called left-invariant if for all left-invariant vector fields, $X,Y$, $\nabla_XY$ is a left-invariant vector field.  Then

\begin{proposition}
There is a one-to-one correspondence between the set of  left-invariant connections on $G$ and the set $\operatorname{Hom}_{\mathbb F}(\mathfrak g\otimes\mathfrak g,\mathfrak g)$.
\end{proposition}

Given a left-invariant connection, $\nabla$, let $\alpha\in\operatorname{Hom}_\mathbb F(\mathfrak g\otimes\mathfrak g,\mathfrak g)$ be the corresponding bilinear form, and let $s$ and $a$ be the symmetric, respectively skew-symmetric summands of $\alpha$, i.e., $s=\frac{\alpha+\sigma\alpha}{2}$ and $a=\frac{\alpha-\sigma\alpha}{2}$, where $\sigma\alpha(x,y):=\alpha(y,x)$ for all $x,y\in\mathfrak g$.

\begin{corollary}
The connection $\nabla$ is torsion-free if and only if $a(\cdot,\cdot)=\frac{1}{2}[\cdot,\cdot],$
where $[\cdot,\cdot]$ is the Lie bracket on $\mathfrak g$.
\end{corollary}

A left-invariant connection $\nabla$ is called a \emph{Cartan connection} if there exists a one-to-one correspondence between the set of the geodesics of $\nabla$ going through the unit $e$ and the $1$-parameter subgroups of $G$.  

\begin{theorem}
A left-invariant connection $\nabla$ on $G$ is a Cartan connection if and only if the symmetric part of the bilinear form $\alpha$ corresponding to $\nabla$ is zero. In other words, Cartan's connections on $G$ are in one-to-one correspondence with $\operatorname{Hom}(\Lambda^2\mathfrak g,\mathfrak g)$.
\end{theorem}

 Let  $\lambda\in\mathbb F$ and define $\alpha_\lambda:\mathfrak g\otimes\mathfrak g\rightarrow\mathfrak g$ by $\alpha_\lambda(x,y)=\lambda[x,y]$ for all $x,y\in\mathfrak g$. Then the curvature and the torsion of the left-invariant connection defined by $\alpha_\lambda$ are 
\begin{eqnarray*}
	\mathrm{R}_\lambda(X,Y)Z&=&(\lambda^2-\lambda)[[X,Y],Z]\\
 	\mathrm{T}_\lambda(X,Y)  &=&(2\lambda-1)[X,Y],
\end{eqnarray*}
for all $X,Y,Z\in\mathfrak X_G$. In particular, the Cartan connection defined by $\alpha_\lambda(\cdot,\cdot)=\lambda[\cdot,\cdot]$ is flat if and only if $\lambda=1$ or $\lambda=0$.
Then, going back to our main topic, one finds 
 
\begin{corollary}
The Cartan connections defined by $\nabla_XY=[X,Y]$ and $\nabla_XY=0$, for all $X,Y\in\mathfrak X_G$ define a (left) post-Lie algebra structure on $\mathfrak X_G$.
\end{corollary}

\begin{proof}
Note that the two cases correspond to $\lambda=1$, $\lambda=0$, respectively. For $\lambda=1$ one has $\mathrm{T}(\cdot,\cdot)=[\cdot,\cdot]$, while for $\lambda=0$ one has $\mathrm{T}(\cdot,\cdot)=-[\cdot,\cdot]$. On the other hand,
\begin{equation*}
(\nabla_X\mathrm{T})(Y,Z)=\nabla_X(\mathrm{T}(Y,Z))-\mathrm{T}(\nabla_XY,Z)-\mathrm{T}(Y,\nabla_XZ).
\end{equation*}
Then, when $\lambda=0$, one has 
$(\nabla_X\mathrm{T})(Y,Z)=0$, since $\nabla_XY=0$ for all $X,Y$ and when $\lambda=1$, one has
\begin{eqnarray*}
	(\nabla_X\mathrm{T})(Y,Z)&=&\nabla_X(\mathrm{T}(Y,Z))-\mathrm{T}(\nabla_XY,Z)-\mathrm{T}(Y,\nabla_XZ)\\
	&=&\nabla_X([Y,Z])-[[X,Y],Z]-[Y,[X,Z]]\\
	&=&[X,[Y,Z]]-[[X,Y],Z]-[Y,[X,Z]]
	=0,
\end{eqnarray*}
thanks to the Jacobi identity. Then the statement follows from Proposition \ref{pro:ge}, observing that $[\cdot,\cdot]_\mathrm{T}=[\cdot,\cdot]$.
\end{proof}

%%%%%%%%%%%%%%%%%%%%%%%%%%%%%

\section{Poisson structures and $r$-matrices}
\label{sec:PoisYB}

This section has two main goals. First to introduce the theory of classical $r$-matrices and, second, the one of isospectral flows. To this end, we will introduce the reader to the theory of the classical integrable systems, where both classical $r$-matrices and isospectral flows play a central role. Classical $r$-matrices will be used to produce examples of pre and post-Lie algebras, while isospectral flows will be studied in the last section from the point of view of post-Lie algebra. We will also discuss in some details the factorization of (suitable) elements of a Lie group whose Lie algebra is endowed with a classical $r$-matrix. The analogue of this construction, applied to the group-like elements of (the $I$-adic completion of) the universal enveloping algebra of a finite dimensional Lie algebra, will be discuss at the end of these notes.

%%%%%%%%%%%%%%%%%%%%%%%%%%%%%

\subsection{Poisson manifolds and isospectral flows}
\label{ssec:PoisIso}

We start recalling the definition of a Poisson algebra. A commutative and associative algebra $(A,\cdot)$ is called a \emph{Poisson algebra} if $A$ is endowed with a skew-symmetric bi-derivation $\{\cdot,\cdot\}:A\times A\rightarrow A$ which fulfils the Jacobi identity. The bi-derivation $\{\cdot,\cdot\}$ is called a \emph{Poisson bracket}. In particular, a smooth manifold $P$ is called a \emph{Poisson manifold} if its algebra of smooth functions $C^\infty(P)$ is a Poisson algebra. For an extensive review of the theory of Poisson algebras and Poisson manifolds we refer the reader to the monograph \cite{Poisson}.

\begin{example}[Symplectic structures vs. Poisson structures]\label{ex:symstr}
Let $M$ be a smooth manifold and let $\omega\in\Omega^2(M)$ be a \emph{non-degenerate} $2$-form, i.e., a smooth section of $\Lambda^2T^\ast M$, the second exterior power of $T^\ast M$, such that $\omega_m:T_m^\ast M\times T_m^\ast M\rightarrow\mathbb F$  is non-degenerate, for all $m\in M$. Then, $\omega$ defines an isomorphism between $\Omega^1(M)$ and $\mathfrak X(M)$ which associates to each $f\in C^{\infty}(M)$ its \emph{Hamiltonian vector field} $X_f$, defined by the condition $df=-\omega(X_f,\cdot)$. In this way, for each $f,g\in C^\infty(M)$ one can define $\{f,g\}_\omega=\omega(X_f,X_g)$, which is a skew-symmetric bi-derivation of $C^\infty(M)$. Furthermore, $\{\cdot,\cdot\}_\omega$ is a Poisson bracket if and only if $d\omega=0$. In this case $\omega$ is called a \emph{symplectic form} and the Poisson tensor corresponding to $\{\cdot,\cdot\}_\omega$ is the \emph{inverse} of $\omega$.
\end{example}

Let $\mathfrak{g}$ be a finite dimensional Lie algebra and $\mathfrak{g}^\ast$ its dual vector space. Let $\mathbb F[\mathfrak{g}]$ be the algebra of polynomial functions on $\mathfrak{g}^\ast$. For $x,y\in\mathfrak{g}$, let
\begin{equation}\label{eq:linpoi}
	\{x,y\}(\alpha):=\langle\alpha,[x,y]\rangle,\quad \forall\alpha\in\mathfrak{g}^\ast.
\end{equation}
Since $\mathbb F[\mathfrak{g}]$ is generated in degree one by $\mathfrak{g}$, the bracket in \eqref{eq:linpoi} admits a unique extension to a skew-symmetric bi-derivation of $\mathbb F[\mathfrak{g}]$. This bracket satisfies the Jacobi identity, and therefore yields a Poisson bracket on the algebra of polynomial functions on $\mathfrak{g}^\ast$. On a basis $x_1,\dots,x_n$ of $\mathfrak g$, \eqref{eq:linpoi} reads $ \{x_i,x_j\}=\sum_{k=1}^nC_{ij}^kx_k$, for $i,j=1,\dots,n,$ where  $\{C^k_{ij}\}_{i,j,k=1,\dots,n}$, are the \emph{structure constants} of $\mathfrak{g}$, defined by $[x_i,x_j]=C^k_{ij}x_k$, for $i,j=1,\dots,n$. This implies that $\mathfrak{g}$, seen as the vector sub-space of $\mathbb F[\mathfrak{g}]$ of  linear functions on $\mathfrak g^\ast$, is closed with respect to \eqref{eq:linpoi}, i.e., it implies that the Poisson bracket of two linear functions is still a linear function. For this reason, the Poisson bracket induced on $\mathbb F[\mathfrak{g}]$ by \eqref{eq:linpoi} is called a \emph{linear Poisson bracket}. From the discussion above it follows that giving a Lie bracket on $\mathfrak{g}$ is equivalent to giving a linear Poisson structure on $\mathfrak g^\ast$. Let us recall that, given a Poisson manifold $(P,\{\cdot, \cdot\})$ and a smooth function $H:P\rightarrow\mathbb F$, one can define the Hamiltonian vector field $X_H\in \mathfrak X(P)$ whose Hamiltonian is $H$:
\[
	X_H(m)=\{H,\cdot\}(m),\quad \forall m\in P,
\]
or, equivalently, $X_H(m)=\Pi(dH,\cdot)(m)$, for $ m\in P$.

Now let $H\in C^{\infty}(\mathfrak{g}^\ast)$. Then the Hamiltonian vector field $X_H$ with respect to the linear Poisson bracket defined on $\mathfrak{g}^\ast$ is given by:
\begin{equation}\label{eq:hamv}
	X_H(\alpha)=-\operatorname{ad}^\sharp_{dH_\alpha}(\alpha),\quad \forall \alpha\in\mathfrak{g}^\ast,
\end{equation}
where $\operatorname{ad}^\sharp$ is the co-adjoint representation of $\mathfrak{g}$. The Hamiltonian equations, corresponding to the integral curves of the vector field $X_H$ in \eqref{eq:hamv} can be written as
\begin{equation}\label{eq:hame}
	\dot{\alpha}=-\operatorname{ad}^\sharp_{dH_\alpha}(\alpha).
\end{equation}

Recall now that $f\in C^{\infty}(P)$ is called a \emph{Casimir} of $(P,\{\cdot,\cdot\})$ if $\{f,g\}=0$, for all $g\in C^\infty(P)$, which forces $X_f$ to be the zero-vector field, i.e., $X_f$ is such that $X_f(m)=0$ for all $m\in P$.
Moreover, this condition implies that Casimir functions are constant along the leaves of the \emph{symplectic foliation} associated to $(P,\{\cdot,\cdot\})$. \footnote{The symplectic foliation of a Poisson manifold $(P,\{\cdot,\cdot\})$ is the generalized distribution in the sense of Sussmann defined on $P$ by the Hamiltonian vector fields. Each leaf of this distribution is an immersed symplectic manifold, i.e., it is an immersed submanifold of $P$ which carries a symplectic structure. See Example \ref{ex:symstr}, which is defined by the restriction to the leaf of the Poisson structure.}
In the particular case of a linear Poisson structure, if $G$ is assumed to be connected, one can prove that $f \in C^{\infty}(\mathfrak g^\ast)$ is a Casimir if and only if $f$ is $G$-invariant, i.e., if and only if $\operatorname{Ad}_g^\sharp f = f$, for all $g \in G$. In this case, $f$ is a Casimir if and only if for all $\alpha \in \mathfrak g^\ast$
\begin{equation}\label{eq:casin}
	\operatorname{ad}^{\sharp}_{df_\alpha}(\alpha)=0.
\end{equation}

The vector space of the Casimirs of the Poisson manifold $(P,\{\cdot,\cdot\})$ is denoted by $\text{Cas}(P,\{\cdot,\cdot\})$. It is a commutative Poisson sub-algebra (actually the center) of $(C^\infty(P),\{\cdot,\cdot\})$.

%%%%%%%%%%%%%%%%%%%%%%%%%%%%%

\subsubsection{Lax-type equations} 

Let $H\in C^{\infty}(\mathfrak{g})$. Then $dH_{\alpha}\in\mathfrak{g}$, for all $\alpha\in\mathfrak{g}^\ast$, and $dH$ is a (smooth) map between $\mathfrak{g}^\ast$ and $\mathfrak{g}$. Suppose now that $\mathfrak g$ is a quadratic Lie algebra, i.e., a Lie algebra endowed with a non-degenerate, symmetric bilinear form $B:\mathfrak g\otimes\mathfrak g\rightarrow\mathbb F$, which is invariant with respect to the adjoint action of $\mathfrak g$, i.e., $B(\operatorname{ad}_xy,z)+B(y,\operatorname{ad}_xz)=0$, for all $x,y,z$ in $\mathfrak g$. Then, if $x_\alpha\in\mathfrak g$ is the (unique) vector such that $B(x_\alpha,y)=\langle\alpha,y\rangle$, for all $y\in\mathfrak g$, the integral curves of the Hamiltonian vector field $X_H$ correspond to the integral curves of the vector field defined on $\mathfrak g$ by the system of ordinary differential equations:
\begin{equation}\label{eq:isf} 
	\dot{x}_\alpha=[x_\alpha,dH_\alpha],
\end{equation}
where the bracket on the right-hand side of \eqref{eq:isf} is the Lie bracket of $\mathfrak g$.  Evolution equations of type \eqref{eq:isf} are known as \emph{Lax type equations} or \emph{isospectral flow equations}, see for example \cite{Lax}, since, if $\mathfrak g$ is a matrix Lie algebra\footnote{Note that this is not a restriction, since, by the Ado's theorem, every finite dimensional Lie algebra admits a faithful finite dimensional representation.}, then writing $F_k=\frac{\operatorname{tr}x_\alpha^k}{k}$, one has
\[
	\frac{dF_k}{dt}=\:\operatorname{tr}\Big(\:\frac{dx_\alpha^k}{dt}\Big)
	=k\:\operatorname{tr}\Big(\frac{dx_\alpha}{dt}x_\alpha^{k-1}\Big)
	=k\:\operatorname{tr}([x_\alpha,dH_\alpha]x_\alpha^{k-1})=0,
\]  
implying that, generically, the eigenvalues of $x_\alpha$ are conserved quantities along the flow defined by \eqref{eq:isf}.

%%%%%%%%%%%%%%%%%%%%%%%%%%%%%

\subsection{$r$-matrices, factorization in Lie groups and integrability}
\label{ss:Rfact}

We follow references \cite{Faybusovich,STS1, STS2}. Let $\mathfrak{g}$ be a finite dimensional Lie algebra over $\mathbb F$, and let $R\in\operatorname{End}_{\mathbb F}(\mathfrak{g})$. Then the bracket
\begin{equation} \label{doubleLiebracket}
	[x,y]_R := \frac{1}{2}([Rx,y]+[x,Ry]),\quad \forall x,y\in\mathfrak{g}
\end{equation}
is skew-symmetric. Moreover, if:
\begin{equation}\label{eq:B}
	B(x,y):=R([Rx,y]+[x,Ry])-[Rx,Ry],
\end{equation}
then $[\cdot,\cdot]_R$ satisfies the Jacobi identity \emph{if and only if}:
\begin{equation}\label{eq:idB}
	[B(x,y),z]+[B(z,x),y]+[B(y,z),x]=0,\quad \forall x,y,z\in\mathfrak{g}.
\end{equation}
In fact, the Jacobi identity for $[\cdot,\cdot]_R$ is equivalent to:
\begin{equation}
	\sum_{\circlearrowleft}\big( \big[R([Rx,y]+[x,Ry]),z\big]+\big[[Rx,y]+[x,Ry],Rz\big] \big)= 0
	\label{eq:NI}
\end{equation}
where $\sum_{\circlearrowleft}$ denote cyclic permutations of $(x,y,z)$. On the left-hand side of the previous equation, the following three-terms sum appears:
\[
	\big[[Rx,y]+[x,Ry],Rz\big]+\big[[Rz,x]+[z,Rx],Ry\big]+\big[[Ry,z]+[y,Rz],Rx\big]
\]
which, using the Jacobi identity for the bracket $[\cdot,\cdot]$, becomes
\[
	-[[Rx,Ry],z]-[[Rz,Rx],y]-[[Ry,Rz],x].
\]
From the previous computation it follows that if, for some $\theta\in\mathbb F$, $B(x,y)=\theta [x,y]$, which amounts to the following identity 
\begin{equation}\label{eq:impc}
	[Rx,Ry]=R([Rx,y]+[x,Ry])-\theta [x,y],
\end{equation}
for all $x,y\in\mathfrak{g}$, then identity (\ref{eq:idB}) will be fulfilled.

\begin{definition}[Classical $r$-matrix and modified CYBE]
\label{def:r-matrix}
Equation (\ref{eq:impc}) is called {\it modified Classical Yang--Baxter Equation} (mCYBE). Its solution is called {\it classical $r$-Matrix}. For $\theta=0$, equation (\ref{eq:impc}) reduces to what is called {\it classical Yang--Baxter Equation} (CYBE). The Lie algebra with classical $r$-matrix, $(\mathfrak{g},R)$, defines a {\it{double Lie algebra}}. The Lie algebra with bracket $[\cdot,\cdot]_R$ defined in \eqref{doubleLiebracket} is denoted $\mathfrak{g}_R$.
\end{definition}

\begin{remark}
Note that on the underlying vector space $\mathfrak g$ of a double Lie algebra $(\mathfrak{g},R)$ are defined two Lie brackets, the original one, $[\cdot,\cdot]$, and the Lie bracket $[\cdot,\cdot]_R$ defined in \eqref{doubleLiebracket}. Correspondingly we have two linear Poisson structures, i.e., the bracket $\{\cdot,\cdot\}$, defined in \eqref{eq:linpoi}, and the bracket $\{\cdot,\cdot\}_R$, defined by
\[
	\{f,g\}_R(\alpha) :=\langle\alpha,[df_\alpha,dg_\alpha]_R\rangle
				=\frac{1}{2}\langle\alpha,([Rdf_\alpha,dg_\alpha]+[df_\alpha,R dg_\alpha])\rangle.
\]

Furthermore, the two Lie algebra structures yield two co-adjoint actions, $\operatorname{ad}^\sharp$ and $\operatorname{ad}^{\sharp,R}$, defined for all $x,y\in\mathfrak g$ and $\alpha\in{\mathfrak g}^{\ast}$ by 
\begin{eqnarray*}
	\operatorname{ad}^{\sharp}_{x}(\alpha)(y) &:=&-\langle\alpha,[x,y]\rangle\\ 
	\operatorname{ad}^{\sharp,R}_x(\alpha)(y)&:=&-\frac{1}{2}\langle\alpha,[Rx,y]+[x,Ry]\rangle.
\end{eqnarray*}
The definition of $\operatorname{ad}^{\sharp, R}$, together with a simple calculation shows that 
\begin{equation}\label{eq:adR}
	\{f,g\}_R(\alpha)=\frac{1}{2}\big(\operatorname{ad}^{\sharp}_{dg_{\alpha}}(\alpha)(Rdf_{\alpha})
			- \operatorname{ad}^{\sharp}_{df_{\alpha}}(\alpha)(Rdg_{\alpha})\big).
\end{equation}
\end{remark}

Let $R$ be a solution of the mCYBE with $\theta=1$ and let $R_\pm\in\operatorname{End}_{\mathbb F}(\mathfrak g)$ be two maps defined by:
\begin{equation}\label{eq:Rpm}
	R_\pm := \frac{1}{2}(R\pm\operatorname{id}_{\mathfrak g}),
\end{equation}

\begin{proposition}\label{prop:prel} 
\begin{enumerate}
\item The maps $R_\pm:\mathfrak g\rightarrow\mathfrak g$ are homomorphisms of Lie algebras from $\mathfrak g_R$ to $\mathfrak g$, so that  $\mathfrak g_\pm=\text{im}\,R_\pm$ are Lie sub-algebras of $\mathfrak g$. 

\item Let $\mathfrak k_\pm=\text{ker}\,R_\mp$. Then $\mathfrak k_\pm\subset\mathfrak g_\pm$ are ideals, and denoting by $\overline{w}$ the class of the element $w$, the map $C_R:\mathfrak g_+/\mathfrak k_+\rightarrow\mathfrak g_-/\mathfrak k_-$, defined by $C_R\big(\,\overline{(R+\operatorname{id}_{\mathfrak g})x}\,\big)=\overline{(R-\operatorname{id}_{\mathfrak g})x}$, is an isomorphism of Lie algebras. Note that, with a slight abuse of language, we include also the case when $\mathfrak g_\pm/\mathfrak k_\pm$ are the zero-Lie algebras. 

\item Let $\Delta:\mathfrak g\rightarrow\mathfrak g\oplus\mathfrak g$ be the diagonal morphism, and let $i_R=(R_+,R_-)\circ\Delta:\mathfrak g\rightarrow\mathfrak g\oplus\mathfrak g$, defined by $i_R x=(R_+x,R_-x)$, for all $x\in\mathfrak g$. Then, if $\mathfrak g\oplus\mathfrak g$ is endowed with the direct-product Lie algebra structure, $i_R$ is an injective Lie algebra homomorphism of Lie algebras whose image consists of all pairs $(x,y)\in\mathfrak g\oplus\mathfrak g$ such that $C_R\overline x=C_R\overline y$.

\item Each element $x\in\mathfrak g$ has a unique decomposition as $x=x_+-x_-$, where $(x_+,x_-)=i_R x$.
\end{enumerate}

\end{proposition}
\begin{proof}
We will sketch only the proof of item 1. To this end, observe that, for all $x,y\in\mathfrak g$, the following identities hold:
\begin{equation}\label{eq:identity}
	[R_\pm x,R_\pm y]=R_\pm\big([R_\pm x,y]+[x,R_\pm y]\mp [x,y]\big).
\end{equation}
Via a simple computation they yield the equality $R_\pm[x,y]_R=[R_\pm x,R_\pm y]$, for $x,y\in\mathfrak g$.
\end{proof}

The map $C_R$ is called the \emph{Cayley transform} of $R$. In the following example we will introduce an important class of solutions of the mCYBE (with $\theta=1$).

\begin{example}\label{aks} Let $\mathfrak{g}_+,\mathfrak{g}_-$ be two Lie subalgebras of $\mathfrak{g}$ such that $\mathfrak{g}=\mathfrak{g}_+\oplus\mathfrak{g}_-$, and let $i_{\pm}:\mathfrak{g}_{\pm}\rightarrow\mathfrak g\oplus\mathfrak g$, $i_+:x\rightarrow (x,0)$ and $i_-:x\rightarrow (0,x)$ the two canonical embeddings. In particular $\mathfrak g_+$ and $\mathfrak g_-$, as Lie sub-algebras of $\mathfrak g$, centralizing each other, i.e., $[\mathfrak{g}_+,\mathfrak g_-]=0$. Finally, let $\pi_\pm:\mathfrak{g}\rightarrow \mathfrak{g}$ be the corresponding projections, and define
\begin{equation}
	R:=\pi_+-\pi_-.\label{eq:eP}
\end{equation}
First note that 
\begin{equation}\label{eq:pir}
	R+2\pi_-=\operatorname{id}_\mathfrak g = 2\pi_+-R .
\end{equation}
Now let us show that $R$ defined in (\ref{eq:eP}) satisfies \eqref{eq:impc}, i.e., the mCYBE with $\theta=1$. To this end it suffices to observe that $[x,y]_R=[x_+,y_+]-[x_-,y_-]$, which implies
\[
	R\big([Rx,y]+[x,Ry\big])-[Rx,Ry]=[x_+,y_+]+[x_-,y_-]+[x_-,y_+]+[x_+,y_-]=[x,y].
\]
Since $R$ satisfies the mCYBE, $\pi_\pm:=\pm R_\pm$ satisfy the following identity:
\begin{equation*}
	[\pi_\pm x,\pi_\pm y]=\pi_\pm\big([\pi_\pm x,y]+[x,\pi_\pm y]-[x,y]\big),\quad \forall x,y\in\mathfrak g,
\end{equation*}
and they are homomorphisms of Lie algebras from $\mathfrak g_R$ to $\mathfrak g$.
\end{example}

%\begin{remark}\label{rem:isolin}
%Recall that, as vector spaces, $\mathfrak g_R=\mathfrak g$. Let $\Delta:\mathfrak g_R\rightarrow\mathfrak g\oplus\mathfrak g$ be the \emph{diagonal map}, i.e., the linear map $\Delta x=(x,x)$ for all $x \in \mathfrak g_R$, and let $S:\mathfrak g\rightarrow \mathfrak g$ be the map defined by $Sx=-x$, for all $x\in\mathfrak g$. Note that $S$ is an \emph{anti-isomorphism} of Lie algebras, i.e., $S$ is a linear isomorphism such that $S[x,y]=-[Sx,Sy]$, for all $x,y\in\mathfrak g$. Then, if $R$ is a solution of the mCYBE and $R_\pm$ are the two associated linear operators defined in \eqref{eq:identity}, then the identity operator $\operatorname{Id}:\mathfrak g_R\rightarrow\mathfrak g$, defined by $\operatorname{Id} x=x$ for all $x\in\mathfrak g$, can be written as
%\begin{equation}
%	\operatorname{Id}=\mathfrak s\circ (\operatorname{id}_{\mathfrak g},S)\circ (R_+,R_-)\circ\Delta,\label{eq:facid}
%\end{equation}
%where $\mathfrak s:\mathfrak g\times\mathfrak g\rightarrow\mathfrak g$ is the defined by $\mathfrak s(x,y)=x+y$, for all $x,y\in\mathfrak g$.
%\noindent
%This formula could be interpreted as a factorization of $\operatorname{Id}\in\operatorname{Hom}_{\mathbb F}(\mathfrak g_R,\mathfrak g)$. Note that the operator $\operatorname{Id}$ is not a morphism of Lie algebras.
%\end{remark}

%%%%%%%%%%%%%%%%%%%%%%%%%%%%%

\subsection{Factorization in Lie groups.}
%Let $\mathfrak g$ and $\mathfrak g_R$ be as above, and let $G$ respectively $G_R$ be the corresponding connected and simply-connected Lie groups. $R_\pm:\,\mathfrak g_R\rightarrow\mathfrak g$ are defined in \eqref{eq:identity}, and $r_\pm: G_R\rightarrow G$ are the corresponding Lie group homomorphisms.\footnote{Recall that if $G_1$ and $G_2$ are two Lie groups, with Lie algebras $\mathfrak g_1$ respectively $\mathfrak g_2$, such that $G_1$ is connected and simply-connected. Then every morphism $\psi:\mathfrak g_1\rightarrow \mathfrak g_2$ can be integrated to a unique homomorphism of Lie groups $\Psi: G_1\rightarrow G_2$.} Let $G_\pm=\text{im}\,r_\pm$, and $K_\pm=\text{ker}\,r_\pm$ are the normal Lie sub-groups of $G_\pm$, whose Lie algebras are $\mathfrak k_\pm$. Then the Cayley transform $C_R$ gives rise to a morphism of Lie groups $c_R: G_+/K_+ \rightarrow G_-/K_-$,  and one can prove the following result. Let $\delta: G\rightarrow G\times G$ be the diagonal map, $\delta g=(g,g)$, for all $g\in G$, and let $i:G\rightarrow G$ be the \emph{inversion map}. Let $m$ denote the multiplication map of $G$, and define $\tilde m: G\times G\rightarrow G$ to be the map $m \circ (\operatorname{id}_G,i)$, i.e., the map such that $\tilde m(g,h)=gh^{-1}$, for all $(g,h)\in G\times G$.

Let $R$ be a solution of the mCYBE and let $G_\pm\subset G$ be the unique (up to isomorphism), connected and simple Lie groups whose Lie algebras are $\mathfrak g_\pm=R_\pm\mathfrak g$.

\begin{theorem}[\cite{STS1,STS2,STS4,Faybusovich}] \label{thm:factorizationtheorem}
Then, every element $g'$ in a suitable neighborhood of the identity element of $G$ admits a factorization as
\begin{equation}
	g'=h_1h_2^{-1},\label{eq:fac}
\end{equation}
where $h_1\in G_+$ and $h_2\in G_-$.
\end{theorem}

\begin{proof}
Recall that, as vector spaces, $\mathfrak g_R=\mathfrak g$. Let $\Delta:\mathfrak g_R\rightarrow \mathfrak g_R\oplus\mathfrak g_R$ be the diagonal map, i.e., $\Delta(x)=(x,x)$ for all $x \in \mathfrak g_R$. Let $i:\mathfrak g\oplus\mathfrak g\rightarrow \mathfrak g$ the linear map defined by $i(x,y)=x-y,$ for all $x,y \in\mathfrak g$.

Consider the linear map defined by:
\begin{equation}
\begin{CD}
 \mathfrak g_R@>{\Delta}>>\mathfrak g_R\oplus\mathfrak g_R@>(R_+,R_-)>>\mathfrak g\oplus\mathfrak g@>{i}>>\mathfrak g.\label{eq:comp}
\end{CD}
\end{equation}
Then $i\circ (R_+,R_-)\circ\Delta(x)=x$, for all $x\in\mathfrak g$. Since $\Delta:\mathfrak g_R\rightarrow\mathfrak g_R \oplus \mathfrak g_R$ and $(R_+,R_-):\mathfrak g_R\oplus\mathfrak g_R\rightarrow\mathfrak g\oplus\mathfrak g$ are homomorphism of Lie algebras they integrate to homomorphisms of Lie groups, which will be denoted as $\delta:G_R\rightarrow G_R\times G_R$ and $(r_+,r_-):G_R\times G_R\rightarrow G\times G$, respectively. In particular, for each $g\in G_R$, $g_\pm=r_\pm g$. Furthermore, note that, even though $i:\mathfrak g\oplus\mathfrak g\rightarrow\mathfrak g$ is not a homorphism of Lie algebras, it is the differential (at the identity) of the map $j:G\times G\rightarrow G$, defined by $j(g,h)=gh^{-1}$. Then the map defined in \eqref{eq:comp} is the differential (at the identity $e\in G_R$) of the map $\psi:G_R\rightarrow G$ defined by
\begin{equation}
	\psi=j\circ (r_+,r_-)\circ\delta.\label{eq:map} 
\end{equation}
Observe now that since $\psi_{\ast,e}=\operatorname{id}$, the map $\psi$ is a local diffeomorphism, i.e., there exist neighborhoods $V$ and $U$ of the identities elements of both $G_R$ and $G$ such that $\psi\vert_V:V\rightarrow U$ is a diffeomorphism. Moreover, just applying the definition of the map $\psi$ given in formula \eqref{eq:map}, one has that 
\begin{equation}
	\psi(g)=g_+g_-^{-1},\label{eq:formexp} 
\end{equation}
for all $g\in G_R$. Taking now any element $g'$ in a suitable neighborhood of the identity of $G$ (eventually contained in $U$), then 
\[
	g'=\psi(\psi^{-1}g')=(\psi^{-1}(g'))_+(\psi^{-1}(g'))_-^{-1}.
\]
The statement now follows taking $h_1=(\psi^{-1}(g'))_+$ and $h_2=(\psi^{-1}(g'))_-$.
\end{proof}

To simplify statement and notation, let us suppose that the map $\psi$ is a global diffeomorphism. As remarked above, the map $\psi$ is not a homomorphism of Lie groups. On the other hand one can prove that

\begin{corollary} [\cite{STS1,STS2,STS4,Faybusovich}]
The map $\ast:G\times G\rightarrow G$ defined by
\begin{equation}
	g\ast h=(\psi^{-1}(g))_+h(\psi^{-1}(g))_-^{-1}\label{eq:newpro}
 \end{equation}
 for all $g,h\in G$, defines a new structure of Lie group on the underlying manifold of the Lie group $G$. Let $G_\ast$ be the Lie group whose product is $\ast$ and whose underlying manifold is $G$. Then $\psi:G_R\rightarrow G_\ast$ is an isomorphism of Lie groups.
\end{corollary}

Note that the product $\ast$ defined in the previous corollary can be obtained as the \emph{push-forward} via $\psi$ of the product defined on $G_R$. More precisely one can prove that

\begin{proposition}[\cite{STS1,STS2,STS4,Faybusovich}]
For all $g,h\in G$, one has that
\[
	g\ast h=\psi(\psi^{-1}(g)\psi^{-1}(h)).  
\]
\end{proposition}

%\begin{enumerate}

%\item The map $I_R:G_R\rightarrow G\times G$, defined by $I_R g=(r_+,r_-)\circ\delta=(r_+g,r_-g)$ for all $g \in G_R$, is an \emph{embedding} of Lie groups whose image consists of the pairs 
%$(g,h)$ such that $c_R\overline g=\overline h$. Here $\overline g$ and $\overline h$ are the classes of $g$ and $h$ in $G_+ /K_+$ respectively $G_-/K_-$.

%\item The map $\tilde m\circ I_R: G_R \rightarrow G$, $g\rightarrow (r_+g,{r_-g}^{-1})$, is a \emph{local diffeomorphism}, such that every $g \in G$ sufficiently close to the identity admits a unique factorization as
%\begin{equation}
%	g=g_+{g_-}^{-1},\label{eq:factlie}
%\end{equation}
%where $(g_+,g_-)\in\text{im}\,I_R$.

%\end{enumerate}
%\end{theorem} 

%\begin{proof}
%The proof of this statement follows from the decomposition of the Lie algebra $\mathfrak g$ into the components $\mathfrak g_+ $ and $\mathfrak g_-$, together with the observation, that the differential at the identity $e \in G$ of the map $\tilde m \circ I_R$ is an isomorphism of vector spaces.\\
%\end{proof}

\noindent
\emph{Applications to dynamics and integrability.} 
Classical $r$-matrices and double Lie algebras play an important role in the theory of classical integrable systems, both finite and infinite dimensional. The relevance of these object to this theory stems from the next result. 

\begin{theorem}[\cite{STS1,STS2,STS4}]\label{thm:STS1}
Let $\mathfrak g$ be a finite dimensional Lie algebra and $R$ a solution of the mCYBE. Let $G$ be the connected and simply-connected Lie group corresponding to $\mathfrak g$. Let $\{\cdot,\cdot\}$ and $\{\cdot,\cdot\}_R$ be the linear Poisson brackets defined on $\mathfrak g$. Then:
\begin{enumerate}
\item The elements of $\text{Cas}({\mathfrak g}^{\ast},\{\cdot,\cdot\})$ Poisson commute with respect to the Poisson bracket $\{\cdot,\cdot\}_R$.

\item For every $f \in \text{Cas}({\mathfrak g}^{\ast},\{\cdot,\cdot\})$, the Hamiltonian vector field $X_f^R$, defined by $\{\cdot,\cdot\}_R$, equals:
\begin{equation}
	X_f^R(\alpha)=-\frac{1}{2}\operatorname{ad}^{\sharp}_{Rdf_{\alpha}}\alpha.\label{eq:HamSTS1}
\end{equation}

\item If $(\mathfrak g,(\cdot\,\vert\,\cdot))$ is a quadratic Lie algebra, then to $X_f^R$ corresponds a vector field $\widetilde{X_f^R}$ on $\mathfrak g$, defined by the following Lax (type) equation:
\begin{equation} 
	\widetilde{X_f^R}(x)=\frac{1}{2}[x,Rdf_x].\label{eq:HamSTS2}
\end{equation}
The vector field $\widetilde{X_f^{\operatorname R}}$ is obtained using the diffeomorphism between $\mathfrak g$ and $\mathfrak g^{\ast}$ induced by the bilinear form $(\cdot\,\vert\,\cdot)$.
\end{enumerate}
\end{theorem}

\begin{proof}
The proof of the first statement of the theorem follows from \eqref{eq:casin} and \eqref{eq:adR}. The second statement follows by a direct computation
\begin{eqnarray}
	\langle X_f^R(\alpha),x\rangle
	&=&-\langle\operatorname{ad}^{\sharp,R}_{df_{\alpha}}\alpha,x\rangle
	=\langle\alpha,\operatorname{ad}^{R}_{df_{\alpha}}x\rangle\nonumber\\
	&=&\frac{1}{2}\langle\alpha,[Rdf_{\alpha},x]\rangle+\frac{1}{2}{\langle\alpha,[df_{\alpha},Rx]}\rangle 
	=-\frac{1}{2}\langle\operatorname{ad}^{\sharp}_{Rdf_{\alpha}}\alpha,x\rangle\nonumber,
\end{eqnarray}
where we used that ${\langle\alpha,[df_{\alpha},Rx]}\rangle =0$, see (\ref{eq:casin}). This proves formula (\ref{eq:HamSTS1}). Finally, using the non-degenerate, bilinear $\operatorname{ad}$-invariant form $(\cdot\,\vert\,\cdot)$ defined on $\mathfrak g$ we can write for $\alpha\in\mathfrak g^{\ast}$ and $y \in\mathfrak g$
\[
	\langle X_f^R(\alpha),y\rangle	=\frac{1}{2}\langle\alpha,[Rdf_{\alpha},y]\rangle
							=\frac{1}{2}\big(x_{\alpha}\,\vert\,[Rdf_{\alpha},y]\big).
\]
In the previous formula, $x_{\alpha}$ is the element in $\mathfrak g$ corresponding to $\alpha \in \mathfrak g^{\ast}$ via the isomorphism between $\mathfrak g^{\ast}$ and $\mathfrak g$ induced by $(\cdot\,\vert\,\cdot)$. Using now the $\operatorname{ad}$-invariance of $(\cdot\,\vert\,\cdot)$, the last term of the previous formula can be written as:
\[
	\frac{1}{2}\big([x_{\alpha},Rdf_{\alpha}]\,\vert\,y\big)
\]
implying that:
\[
	\langle X_f^R(\alpha),y\rangle=\frac{1}{2}\big([x_{\alpha},Rdf_{\alpha}]\,\vert\,y\big).
\]
Using again the isomorphism defined by $(\cdot\,\vert\,\cdot)$, we can write:
\[ 
	\widetilde{X_f^R}(x_{\alpha})=\frac{1}{2}[x_{\alpha},Rdf_{\alpha}],
\]
which proves the last part of the theorem.
\end{proof}

\begin{remark}
The Hamiltonian equations corresponding to $X_f^R$ have the following form:
\begin{equation}
	\dot{\alpha}=-\frac{1}{2}{\operatorname{ad}}^{\sharp}_{R(df_{\alpha})}\alpha.\label{eq:lhe}
\end{equation}
Moreover, the Hamiltonian vector field $X_f$ corresponding to a Casimir $f \in C^{\infty}({\mathfrak g}^{\ast})$ is identically zero.
\end{remark}

Finally, using the notations of Theorem \ref{thm:STS1}, where $G_\pm$ are the connected and simply-connected Lie groups whose Lie algebras are $\mathfrak g_\pm=\text{im}\,R_\pm$, see Theorem \ref{thm:factorizationtheorem}, one can prove that

\begin{theorem}[\cite{STS2,STS4}]\label{thm:STS2}
Let $f\in\text{Cas}(\mathfrak g,\{\cdot,\cdot\})$ and let $t \rightarrow g_\pm(t)$ be two smooth curves in $G_{\pm}$ solving the factorization problem
\[
	\exp (tdf_\alpha)=g_+(t)g_-(t)^{-1},\quad g_\pm(0)=e.
\]
The integral curve $\alpha=\alpha(t)$ of the vector field \eqref{eq:HamSTS1}, solving \eqref{eq:lhe} with $\alpha(0)=\alpha$, is
\begin{equation}
	\alpha(t)=\operatorname{Ad}^\sharp_{g_+^{-1}(t)}\alpha=\operatorname{Ad}^\sharp_{g_-^{-1}(t)}\alpha.
\end{equation}

\end{theorem}

%%%%%%%%%%%%%%%%%%%%%%%%%%%%%

\subsection{Pre- and post-Lie algebras from classical $r$-matrices}
\label{ss:YBPP}

Let $R \in \operatorname{End}_{\mathbb F}(\mathfrak g)$ be a solution of the CYBE. The next result is well-known.

\begin{proposition}[\cite{GuoBaiNi}]
The binary product $\cdot :\mathfrak g\otimes\mathfrak g\rightarrow\mathfrak g$ defined by
\begin{equation}
	x\cdot y=[Rx,y],\quad \forall x,y\in\mathfrak g\label{eq:preR}
\end{equation}
defines a left pre-Lie algebra on $\mathfrak g$.
\end{proposition}

\begin{proof}
Indeed, for all $x,y,z\in\mathfrak g$ we have
\begin{eqnarray*}	
	(x\cdot y)\cdot z-x\cdot (y\cdot z)
	&=&[R[Rx,y],z]-[Rx,[Ry,z]]\\\
	&\stackrel{(a)}{=}&[R[Rx,y],z]-[[Rx,Ry],z]-[Ry,[Rx,z]]\\
	&\stackrel{(b)}{=}&-[R[x,Ry],z]-[Ry,[Rx,z]]\\
	&=&[R[Ry,x],z]-[Ry,[Rx,Rz]]\\
	&=&(y\cdot x)\cdot z-y\cdot (x\cdot z).
\end{eqnarray*}
In (a) we used the Jacobi identity. In (b) we used \eqref{eq:impc} with $\theta=0$
\end{proof}

Note that the Lie bracket \eqref{doubleLiebracket} defined by a solution  $R$ of the CYBE is, up to a numerical factor, subordinate to the pre-Lie product \eqref{eq:preR}. In fact, since $x\cdot y-y\cdot x=[Rx,y]-[Ry,x]=[Rx,y]+[x,Ry]$, 
\[
	[\cdot,\cdot]_R=\frac{1}{2}(x\cdot y-y\cdot x).
\]
In particular, the Lie bracket \eqref{doubleLiebracket} is subordinate to the pre-Lie product $\bullet :\mathfrak g\otimes\mathfrak g\rightarrow\mathfrak g$ defined by $x\bullet y=\frac{1}{2}x\cdot y$, for all $x,y\in\mathfrak g$. Let $R$ be a solution of the CYBE and let $\bullet:\mathfrak g\otimes\mathfrak g\rightarrow\mathfrak g$ be the pre-Lie product defined by  $x \bullet y=\frac{1}{2}[Rx,y]$ for all $x,y\in\mathfrak g$. Then, if for every $x\in\mathfrak g_R=(\mathfrak g,[\cdot,\cdot]_R)$ one denotes with $X_x$ the left-invariant vector field on $G_R$, the unique connected and simply-connected Lie group whose Lie algebra is $\mathfrak g_R$, then one can prove the next result.

\begin{proposition}
The left-invariant linear connection on $G_R$ defined by
\begin{equation}
	\nabla_{X_x}X_y=\frac{1}{2}X_{[Rx,y]},\label{eq:lipr1}
\end{equation}
is flat and torsion free. In particular, the product
\begin{equation}
	X\cdot Y=\nabla_XY\label{eq:lipr2}
\end{equation}
defines a left pre-Lie algebra on $\mathfrak X_{G_R}$.
\end{proposition}

\begin{proof}
First let us compute the torsion of the connection defined in \eqref{eq:lipr1}.
\allowdisplaybreaks
\begin{eqnarray*}
	\mathrm{T}(X_x,X_y)&=&\nabla_{X_x}X_y-\nabla_{X_y}X_x-[X_x,X_y]\\
			&=&\frac{1}{2}\big(X_{[Rx,y]}-X_{[Ry,x]}\big)-X_{[x,y]_R}
			=0,
\end{eqnarray*}
for all $x,y\in\mathfrak g$. On the other hand, computing the curvature of $\nabla$ one gets:
\allowdisplaybreaks
\begin{eqnarray*}
	\mathrm{R}(X_x,X_y)X_z
	&=&\nabla_{X_x}\nabla_{X_y}X_z-\nabla_{X_y}\nabla_{X_x}X_z-\nabla_{[X_x,X_y]}X_z\\
	&=&\frac{1}{4}\big(X_{[Rx,[Ry,z]]}-X_{[R_y,[Rx,z]]}\big)-\frac{1}{2}\nabla_{X_{([Rx,y]+[x,Ry])}}X_z\\
	&=&\frac{1}{4}\big(X_{[Rx,[Ry,z]]-[Ry,[Rx,z]]-[R([Rx,y]+[x,Ry]),z]}\big)\\
	&\stackrel{(a)}{=}&\frac{1}{4}\big(X_{[[Rx,Ry]-R([Rx,y]+[x,Ry]),z]}\big)
	=0,
\end{eqnarray*}
for all $x,y,z\in\mathfrak g$, proving the first claim. Note that in (a) we used the Jacobi identity, and the last equality follows from $R$ being a solution of CYBE. For the second one, note that it suffices to prove it for the left-invariant vector fields. Then, given $x,y,z\in\mathfrak g$, one has that
\allowdisplaybreaks
\begin{eqnarray*}
	a_\cdot(X_x,X_y,X_z)
	&=&\nabla_{\nabla_{X_x}X_y}X_z-\nabla_{X_x}\nabla_{X_y}X_z\\
	&=&\frac{1}{2}(\nabla_{X_{[Rx,y]}}X_z-\nabla_{X_x}{X_{[Ry,z]}})\\
	&=&\frac{1}{4}X_{([R[Rx,y],z]-[Rx[Ry,z]])}\\
	&=&\frac{1}{4}X_{([R[Rx,y],z]-[[Rx,Ry],z]-[Ry,[Rx,z]])}\\
	&=&\frac{1}{4}X_{[R[Ry,x],z]}-\frac{1}{4}X_{[Ry,[Rx,z]]}
	=a_\cdot(X_y,X_x,X_z),
\end{eqnarray*}
for all $x,y,z\in\mathfrak g$.
\end{proof}

We will now prove the following result, which completes Example \ref{ex:simplecticLieG}. Let $(\mathfrak g,B)$ be a quadratic Lie algebra\footnote{Recall that a quadratic Lie algebra $(\mathfrak g,B)$ is a Lie algebra endowed with a non-degenerate, $\mathfrak g$-invariant bilinear form $B: \mathfrak g\otimes\mathfrak g \rightarrow \mathfrak g$, i.e., $B$ is a bilinear form such that 1) if $x \in\mathfrak g$ is such that $B(x,y)=0$ for all $y \in \mathfrak g$, then $x=0$ and 2) $B([x,y],z)+B(y,[x,z])=0$ for all $x,y,z\in\mathfrak g$.}. Then we have the next result.

\begin{proposition}[Drinfeld \cite{Drinf}]
The set of \emph{invertible} and \emph{skew-symmetric} solutions of the CYBE on $(\mathfrak g,B)$ is in one-to-one correspondence with set of the invariant symplectic structures on the corresponding connected and simply-connected Lie group $G_R$. In particular, every invertible solution of the CYBE defines a left pre-Lie algebra structure on $\mathfrak X_{G_R}$.
\end{proposition}

\begin{proof}
Let $R$ be an invertible solution of the CYBE on $\mathfrak g$ and let $\omega(\cdot,\cdot)=B(R\cdot,\cdot):\mathfrak g\otimes\mathfrak g\rightarrow\mathfrak g$. Then $\omega$ is non-degenerate and skew-symmetric. In fact, since $R$ is skew-symmetric one has that 
\[
	 \omega(y,x)=B(Ry,x)=-B(y,Rx)=-B(Rx,y)=-\omega(x,y),
\]
 for all $x,y\in\mathfrak g$, and if $x\in\mathfrak g$ is such that $\omega(x,y)=0$ for all $y\in\mathfrak g$, then $B(Rx,y)=-B(x,Ry)=0$ for all $y\in\mathfrak g$, which implies that $B(x,z)=0$ for all $z\in\mathfrak g$ since $R$ is invertible. Let us now prove that $\omega\in\mathscr Z^2(\mathfrak g_R,\mathbb F)$, i.e., that 
\[
	\omega(x,[y,z]_R)+\omega(z,[x,y]_R)+\omega(y,[z,x]_R)=0,\quad \forall x,y,z\in\mathfrak g.
\]
To this end, first compute
\begin{eqnarray*}
	\omega(x,[y,z]_R)
	&=&\frac{1}{2}B(Rx,[Ry,z])+\frac{1}{2}B(Rx,[y,Rz])\\
 	&=&\frac{1}{2}B([Rx,Ry],z)-\frac{1}{2}B(R[Rx,y],z),
\end{eqnarray*}
then compute

\begin{eqnarray*}
	\omega(y,[z,x]_R)&=&\frac{1}{2}B(Ry,[Rz,x])+\frac{1}{2}B(Ry,[z,Rx])\\
 	&=&\frac{1}{2}B([Rx,Ry],z)-\frac{1}{2}B(R[x,Ry],z).
\end{eqnarray*}
On the other hand,
\begin{eqnarray*}
	\omega(z,[x,y]_R)&=&B(Rz,[x,y]_R)\\
 	&=&-B(R[x,y]_R,z)
	=-\frac{1}{2}B\big(R([Rx,y]+[x,Ry]),z\big).
\end{eqnarray*}
Using the results of these partial computations one has:
\allowdisplaybreaks
\begin{eqnarray*}
	&& \omega(x,[y,z]_R)+\omega(z,[x,y]_R)+\omega(y,[z,x]_R)\\
 	&=&\frac{1}{2}B([Rx,Ry],z)-\frac{1}{2}B(R[Rx,y],z)\\
  	&=&\frac{1}{2}B([Rx,Ry],z)-\frac{1}{2}B(R[x,Ry],z)\\
  	&=&-\frac{1}{2}B\big(R([Rx,y]+[x,Ry]),z\big)
	=0,
\end{eqnarray*}
since $R$ is a solution of the CYBE. On the other hand, suppose that $\omega\in\mathscr Z^2(\mathfrak g_R,\mathbb F)$ is non-degenerate. Then, using $B$ and $\omega$ one can define $B^v:\mathfrak g\rightarrow\mathfrak g^\ast$ and, respectively, $\omega^v:\mathfrak g\rightarrow\mathfrak g^\ast$ to be the linear isomorphisms such that $\langle B^v(x),y\rangle=B(x,y)$ and $\langle\omega^v(x),y\rangle=\omega(x,y)$ for all $x,y\in\mathfrak g$. Then if $R:=(B^v)^{-1}\circ\omega^v:\mathfrak g\rightarrow\mathfrak g$, one has that
\begin{eqnarray*}
	B(Ry,x)&=&\langle\omega^v(y),x\rangle\\
  	&=&\omega(y,x)=-\omega(x,y)=-\langle\omega^v(x),y)=-B(Rx,y)=-B(y,Rx),
\end{eqnarray*}
showing that $R$ is skew-symmetric. On the other hand, if $x\in\mathfrak g$ is such that $Rx=0$, then 
\begin{equation*}
	0=(Rx,y)=\langle\omega^v(x),y\rangle=\omega(x,y),
\end{equation*}
for all $y\in\mathfrak g$, which implies that $x=0$, proving that $R$ is an isomorphism. Furthermore, since $\omega\in\mathscr Z^2(\mathfrak g_R,\mathbb F)$, 
\[
	\omega(z,[x,y]_R)+\omega(y,[z,x]_R)+\omega(x,[y,z]_R)=0\quad \forall x,y,z\in\mathfrak g,
\]
which implies that 
\[
	B(R([Rx,y]+[x,Ry],z)=B([Rx,Ry],z),\quad \forall x,y,z\in\mathfrak g,
\]
\end{proof}
proving that $R$ is a solution of the CYBE. Finally, if $\omega\in\mathscr Z^2(\mathfrak g_R,\mathbb F)$ is non-degenerate, its extension on $G_R$ by left-translations defines a left-invariant symplectic form on $G_R$. Then the last part of the statement of the proposition follows now from the discussion in Example \ref{ex:simplecticLieG}.
 %\end{example}

We will now see how given a solution of the mCYBE on $\mathfrak g$ one can define a structure of a post-Lie algebra on $\mathfrak X_{G_R}$. In spite of the fact that the relation between post-Lie algebra structures on $\mathfrak X_{G_R}$ and solutions of the mCYBE is completely analogous to the one just discussed between the solutions of the CYBE and pre-Lie algebra structures on the $\mathfrak X_{G_R}$, we will give full details also in this case. Before moving to this more geometrical topic, let us make a few observations of algebraic flavor. Let $R \in \operatorname{End}_{\mathbb{F}}(\mathfrak g)$ be a solution of the mCYBE, Equation \eqref{eq:impc}, and let $R_\pm$ be defined as in \eqref{eq:Rpm}. Then: 

\begin{theorem}[\cite{GuoBaiNi}]
\label{thm:post-Lie1}
The binary product 
\begin{equation}
\label{def:RBpost-Lie}
	x \triangleright_\pm y := [R_\pm(x),y].
\end{equation}
 defines a left (right) post-Lie algebra structure on $\mathfrak g$.
\end{theorem} 

\begin{proof}
The axiom \eqref{post-Lie1} holds true since $[R_\pm,\cdot ]$ is a derivation with respect to $[\cdot,\cdot]$. The axiom \eqref{post-Lie2} follows from \eqref{eq:identity} and the Jacobi identity.
%\allowdisplaybreaks{
%\begin{align*}
%	a \triangleright   [b,c] &= - [\pi_+(a),[b,c]]\\
			 		%&= -[[\pi_+(a),b],c] - [b,[\pi_+(a),c]] 
			 		   %= [a \triangleright   b,c] + [b,a \triangleright   c],
%\end{align*}}
%and
%\allowdisplaybreaks{
%\begin{align*}
%	[a,b] \triangleright   c 
%			&= - [\pi_+([a,b]),c]\\
%			&=  [[\pi_+(a),\pi_+(b)],c] - [\pi_+([\pi_+(a),b]),c] - [\pi_+([a,\pi_+(b)]),c]\\
%			&=  [\pi_+(a),[\pi_+(b),c]] - [\pi_+([\pi_+(a),b]),c] - [\pi_+(b),[\pi_+(a), c]] + [\pi_+([\pi_+(b),a]),c]  \\
%			&=  a \triangleright   (b \triangleright   c) -  (a \triangleright   b) \triangleright   c
%			 			 		- b \triangleright   (a \triangleright   c)  + (b \triangleright   a) \triangleright   c\\
%			&= {\rm{a}}_ \triangleright  (a,b,c) - {\rm{a}}_ \triangleright  (b,a,c).
%\end{align*}}
\end{proof}

%It turns out that the new Lie bracket defined in terms of this post-Lie product and the original Lie bracket is the one given in \eqref{doublebracket}. Indeed, for $\pi_+:= \operatorname{id} - \pi_-$,
%\[
%	x \triangleright y - y \triangleright x + [x,y] = [\pi_-(x),y] + [x,\pi_-(y)] - [x,y] = \llbracket x,y \rrbracket.
%\] 
Note that 
\[
	x\triangleright_-y=[R_-x,y]=[(R_+-\operatorname{id}_{\mathfrak g})x,y]=x\triangleright_+y-[x,y]
\]
which is the content of Proposition \ref{pro:postL1}. 
%\[
%x\triangleright_+ y=[R_+x,y]=\frac{1}{2}[Rx,y]+\frac{1}{2}[x,y]=\frac{1}{2}[R_+x,y]+\frac{1}{2}[R_-x,y]+\frac{1}{2}[x,y],
%\]
%which implies that:
%\[
%x\blacktriangleright y=x\triangleright_+ y=x\triangleright_- y+[x,y],\,\forall x,y\in\mathfrak g.
%\]
In particular, a computation shows that:
\[
	x\triangleright_-y-y\triangleright_-x+[x,y]=[x,y]_R=x\triangleright_+y-y\triangleright_+x-[x,y],
\]
for all $x,y\in\mathfrak g$. See Proposition \ref{prop:post-lie}, i.e., 
\begin{equation}
	\llbracket\cdot,\cdot\rrbracket=[\cdot,\cdot]_R.
	\label{eq:brack}
\end{equation}

Moreover, one finds the Lie-admissible algebras $(\mathfrak g,\succ_\pm)$ with binary compositions 
\[
	x \succ_\pm y := x \triangleright_\pm y + \frac{1}{2}[x,y].
\]
 The Lie bracket \eqref{post-Lie3} is then given by $\llbracket x,y\rrbracket=[x,y]_R= x \succ y - y \succ x$, 
 for all $x,y\in\mathfrak g$. Writing $\tilde{R}:=\frac{1}{2}R$, one can deduce from $[\tilde{R}x,\tilde{R}y] - \tilde{R}\big([\tilde{R}x,y]+[x,\tilde{R}y]\big) = -\frac{1}{4}[x,y]$, that   
\[
	 {\rm{a}}_{\succ}(x,y,z) - {\rm{a}}_{\succ}(y,x,z) = -\frac{1}{4}[[x,y],z].
\]

Let us now move to the geometric side and discuss the post-Lie structure defined on $\mathfrak X_{G_R}$ by a any solution of the mCYBE. 

Let $R$ be a solution of the mCYBE and let $R_+$ be as defined in \eqref{eq:Rpm}. Then denoting by $X_x$ the left-invariant vector field on $G_R$ defined by $x\in\mathfrak g$ we have the result

\begin{theorem}
The formula:
\begin{equation}
	\nabla_{X_x}X_y=X_{[R_+x,y]},\quad \forall x,y\in\mathfrak g\label{eq:connpost-LieG}
\end{equation}
defines a flat left-invariant linear connection on $G_R$ with constant torsion. In particular, the product
\begin{equation}
	X\triangleright Y=\nabla_XY\label{eq:prodpost-LieG}
\end{equation}
defines a left post-Lie algebra on $(\mathfrak X_{G_R},[\cdot,\cdot])$, where $[\cdot,\cdot]:\mathfrak X_{G_R}\otimes\mathfrak X_{G_R}\rightarrow\mathfrak X_{G_R}$ is the usual Lie bracket on the set of vector field on the smooth manifold $G_R$.
\end{theorem}
\begin{proof}
The first statement follows from a direct computation. More precisely
\begin{eqnarray*}
	\nabla_{X_x}\nabla_{X_y}X_z-\nabla_{X_y}\nabla_{X_x}X_z
	&=&X_{[R_+x,[R_+y,z]]-[R_+y[R_+x,z]]}\\
	&=&X_{[[R_+x,R_+y],z]},\quad \forall x,y,z\in\mathfrak g.
\end{eqnarray*}
On the other hand,
\begin{eqnarray*}
	\nabla_{[X_x,X_y]}X_z
	=\nabla_{X_{[x,y]_R}}X_z
	=\frac{1}{2}\nabla_{X_{[Rx,y]+[x,Ry]}}X_z	
	&=&X_{[R_+[R_+x,y]-R_+[x,y]+R_+[x,R_+y],z]}\\
	&=&X_{[[R_+x,R_+y],z]}
\end{eqnarray*}
where we used that $R=2R_+-\operatorname{id}_\mathfrak g$ and \eqref{eq:identity} which, together, prove that $\mathrm{R}(X_x,X_y)X_z=0$ for all $x,y,z\in\mathfrak g.$
Let us now compute
\begin{eqnarray*}
	\mathrm{T}(X_x,X_y)
	&=&\nabla_{X_x}X_y-\nabla_{X_y}X_x-[X_x,X_y]\\
	&=&X_{[R_+x,y]+[x,R_+y]-[x,y]_R}\\
	&=&X_{[x,y]}
\end{eqnarray*}
for all $x,y,z\in\mathfrak g$. Then
\begin{eqnarray*}
	(\nabla_{X_z}\mathrm{T})(X_x,X_y)
	&=&\nabla_{X_z}\mathrm{T}(X_x,X_y)-\mathrm{T}(\nabla_{X_z}X_x,X_y)-\mathrm{T}(X_x,\nabla_{X_z}X_y)\\
	&=&\nabla_{X_z}X_{[x,y]}-X_{[[R_+z,x],y]}-X_{[x,[R_+z,y]]}=0
\end{eqnarray*}
Because of its definition, $\nabla$ is left-invariant and since every $X\in\mathfrak X_{G_R}$ can be written as $X=\sum_{i=1}^{\text{dim}\,\mathfrak g}f_iX_{x_i}$ where $f_i\in C^\infty(G_R)$ for all $i=1,\dots,\text{dim}\,\mathfrak g$ and $x_1,\dots,x_{\text{dim}\,\mathfrak g}$ is a basis of $\mathfrak g$, it follows that $\nabla$ has the properties stated in the theorem. The last part of the statement follows now from Proposition \ref{pro:ge} in Section \ref{sec:lieapospre}.
\end{proof}

%%%%%%%%%%%%%%%%%%%%%%%%%%%%%
\section{Post-Lie algebras, factorization theorems and isospectral flows}

In this section we will study the properties of the universal enveloping algebra of a post-Lie algebra. For post-Lie algebras coming from classical $r$-matrices, we will discuss in details the factorization of group-like elements of the relevant $I$-adic completion. In the last part,  we will discuss how this factorization can be applied to find solutions of particular Lax-type equations. 

\subsection{The universal enveloping algebra of a post-Lie algebra}

Proposition \ref{prop:post-lie} above shows that any post-Lie algebra comes with two Lie brackets, $[\cdot,\cdot]$ and $\llbracket \cdot,\cdot\rrbracket$, which are related in terms of the post-Lie product by identity \eqref{post-Lie3}. The relation between the corresponding universal enveloping algebras was explored in \cite{EFLMK}. In \cite{OudomGuin} similar results in the context of pre-Lie algebras and the symmetric algebra $S_\mathfrak g$ appeared.

The next proposition summarizes the results relevant for the present discussion of lifting the post-Lie algebra structure to $\mathcal U(\mathfrak g)$. Denoting the product induced on $\mathcal U(\mathfrak g)$ by the post-Lie product defined on $(\mathfrak g,\triangleright, [\cdot,\cdot])$ with the same symbol $\triangleright$, one can show the next proposition.

\begin{proposition}\cite{EFLMK}\label{prop:post1}
Let $A,B,C\in\mathcal U(\mathfrak g)$ and $x,y\in\mathfrak g \hookrightarrow \mathcal U(\mathfrak g),$ then there exists a unique extension of the post-Lie product from $\mathfrak g$ to $\mathcal U(\mathfrak g)$, given by:
\allowdisplaybreaks{
\begin{align}
	1\triangleright A &= A \label{eq:pha1}\\
	xA\triangleright y &= x\triangleright (A\triangleright y)-(x\triangleright A)\triangleright y\nonumber\\
	A\triangleright BC &=(A_{(1)}\triangleright B)(A_{(2)}\triangleright C).	\label{eq:pha2}
\end{align}}
\end{proposition}

\begin{proof}
The proof of Proposition \ref{prop:post1} goes by induction on the length of monomials in $\mathcal U(\mathfrak g)$. 
\end{proof}

Note that \eqref{eq:pha1} together with \eqref{eq:pha2} imply that the extension of the post-Lie product from $\mathfrak g$ to $\mathcal U(\mathfrak g)$  yields a linear map $d:\mathfrak g \rightarrow \operatorname{Der}\big(\mathcal U(\mathfrak g)\big),$ defined via $d(x)(x_1\cdots x_n):=\sum_{i=1}^nx_1\cdots (x\triangleright x_i)\cdots x_n$, for any word $x_1\cdots x_n \in \mathcal U(\mathfrak g)$. A simple computation shows that, in general, this map is not a morphism of Lie algebras. Together with Proposition \ref{prop:post1} one can prove 

\begin{proposition}\label{prop:post2}
\allowdisplaybreaks{
 \begin{align}
	A \triangleright 1 &= \epsilon (A),\\
	\epsilon(A\triangleright B) &=\epsilon(A)\epsilon (B),\\
	\Delta (A\triangleright B)&=(A_{(1)}\triangleright B_{(1)}) \otimes (A_{(2)}\triangleright B_{(2)}),\\
	xA\triangleright B&=x\triangleright (A\triangleright B)-(x\triangleright A)\triangleright B,\\
	A\triangleright (B\triangleright C)&=(A_{(1)}(A_{(2)}\triangleright B))\triangleright C. \label{last}
\end{align}}
 \end{proposition}

\begin{proof}
These identities follow by induction on the length of monomials in $\mathcal U(\mathfrak g)$.
\end{proof}
 
It turns out that identity \eqref{last} in Proposition \ref{prop:post2} can be written $A\triangleright (B\triangleright C)=(A*B)\triangleright C$, where the product $m_\ast:\mathcal U(\mathfrak g)\otimes\mathcal U(\mathfrak g)\rightarrow\mathcal U(\mathfrak g)$ is defined by
\begin{equation}
\label{eq:post-LieU}
	m_\ast(A\otimes B)= A\ast B:=A_{(1)}(A_{(2)}\triangleright B).
\end{equation}

\begin{theorem}\cite{EFLMK}\label{thm:KLM0}
The product defined in \eqref{eq:post-LieU} is non-commutative, associative and unital. Moreover, $\mathcal U_*(\mathfrak g):=(\mathcal U(\mathfrak g),m_\ast,1,\Delta,\epsilon,S_\ast)$ is a co-commutative Hopf algebra, whose unit, co-unit and co-product coincide with those defining the usual Hopf algebra structure on $\mathcal U(\mathfrak g)$. The antipode $S_\ast$ is given uniquely by the defining equations:
\[
	m_\ast\circ(\operatorname{id}\otimes S_\ast)\circ\Delta
	=1\circ\epsilon
	=m_\ast\circ(S_\ast\otimes\operatorname{id})\circ\Delta.
\]
More precisely
\begin{equation}
	S_\ast (x_1\cdots x_n)=-x_1\cdots x_n-\sum_{k=1}^{n-1}
	\sum_{\sigma\in\Sigma_{k,n-k}}x_{\sigma(1)}\cdots x_{\sigma(k)}\ast 
	S(x_{\sigma(k+1)}\cdots x_{\sigma(n)}),\label{eq:antipodests}
\end{equation}
for every $x_1\cdots x_n\in\mathcal U_n(\mathfrak g)$ and for all $n\geq 1$. 
\end{theorem}

\noindent Here $\Sigma_{k,n-k} \subset \Sigma_n$ denotes the set of $(k,n-k)$-\emph{shuffles}, i.e. the elements $\sigma\in\Sigma_n$  %of n elements $[n]:=\{1, 2, \ldots, n\}$ 
such that ${\sigma(1)}< \cdots <\sigma(k)$ and ${\sigma(k+1)}< \cdots <{\sigma(n)}$. Note that since elements $x \in \mathfrak g$ are primitive and $\Delta$ is a $\ast$-algebra morphism, one deduces 

\begin{lemma}\label{lem:coprodast}
\allowdisplaybreaks{
\begin{eqnarray*}
	\Delta(x_1\ast\cdots\ast x_n)
	&=&x_1\ast\cdots\ast x_n\otimes 1 \nonumber +1\otimes x_1\ast \cdots \ast x_n\\
	&+&\sum_{k=1}^{n-1}\sum_{\sigma\in\Sigma_{k,n-k}}
			x_{\sigma(1)} \ast \cdots \ast x_{\sigma(k)}\otimes x_{\sigma(k+1)}\ast \cdots \ast x_{\sigma(n)}.
\end{eqnarray*}}
\end{lemma}

The relation between the Hopf algebra $\mathcal U_*(\mathfrak g)$ in Theorem \ref{thm:KLM0} and the universal enveloping algebra $\mathcal U(\overline{\mathfrak g})$ corresponding to the Lie algebra $\overline{\mathfrak g}$ is the content of the following theorem.

\begin{theorem}\cite{EFLMK}\label{thm:KLM}
$\mathcal U_*(\mathfrak g)$ is isomorphic, as a Hopf algebra, to $\mathcal U(\overline{\mathfrak g})$. More precisely, the identity map $\operatorname{id}:\overline{\mathfrak g}\rightarrow\mathfrak g$ admits a unique extension to an isomorphism of Hopf algebras $\phi:\mathcal U(\overline{\mathfrak g})\rightarrow \mathcal U_*(\mathfrak g)$.
\end{theorem} 
\begin{proof}
First, let us verify the existence of an algebra morphism $\phi:\mathcal U(\overline{\mathfrak g})\rightarrow\mathcal U_\ast(\mathfrak g)$. To this end, note that the inclusion map $i:\mathfrak g \hookrightarrow \mathcal U_\ast(\mathfrak g)$, via the universal property of the tensor algebra $T\mathfrak g$, guarantees the existence of an algebra morphism $I:T\mathfrak g\rightarrow\mathcal U_\ast(\mathfrak g)$ making the following diagram commutative:
\begin{equation*}
\xy\xymatrix{
&&\mathfrak g\ar[dl]_{i_T}\ar[dr]^{i}&&\\
&{T\mathfrak g}\ar[rr]^{I} &&\mathcal U_\ast(\mathfrak g)&
}\endxy
\end{equation*}
where $i_T:\mathfrak g \hookrightarrow T\mathfrak g$ is an inclusion map. Note that, since $i(x)=x\in\mathcal U_\ast(\mathfrak g)$ and $i_T(x)=x\in T\mathfrak g$ for all $x\in\mathfrak g$, one has $I(x)=x$ for all $x\in\mathfrak g$, i.e., the map $I$ restricts to the identity on $\mathfrak g$. Then, for all monomials $x_1\otimes\cdots\otimes x_n\in T\mathfrak g$, one has
\[
	I(x_1\otimes\cdots\otimes x_n)=x_1\ast\cdots\ast x_n,
\]
and, since $x\ast y-y\ast x=\llbracket x,y\rrbracket$ for $x,y \in \mathfrak g$,
\[
	I(x\otimes y-y\otimes x-\llbracket x,y\rrbracket)=0.
\]
It follows then that the map $I:T\mathfrak g\rightarrow\mathcal U_\ast(\mathfrak g)$ factors through the (bilateral) ideal $J=\langle x\otimes y-y\otimes x-\llbracket x,y\rrbracket\rangle\subset T\mathfrak g$, defining a morphism of (filtered) algebras $\phi:\mathcal U(\overline{\mathfrak g})\rightarrow\mathcal U_\ast(\mathfrak g)$ which makes the following diagram commutative:
\begin{equation*}
	\xy\xymatrix{
	&&\mathfrak g\ar[dl]_{i_T}\ar[dr]^{i}&&\\
	&T{\mathfrak g}\ar[rr]^{I}\ar[dr]_{\pi} &&\mathcal U_\ast(\mathfrak g)&\\
	&&\mathcal U(\overline{\mathfrak g})\ar[ur]_{\phi}&&}\endxy
\end{equation*}
where $\pi:T\mathfrak g\rightarrow\mathcal U(\overline{\mathfrak g})$ is the canonical projection, i.e., $\pi (A)=A\ \text{mod}\ J$, for all $A\in T\mathfrak g$. Note that since $\pi (x)=x$ for all $x\in\mathfrak g$, the map $\phi$ restricts to the identity on $\mathfrak g$. Now, using a simple inductive argument on the length of monomials, one can show that for all $A\in\mathcal U_n(\mathfrak g)$ and $B\in\mathcal U_m(\mathfrak g)$
\[
	m_\ast(A\otimes B)=AB\  \text{mod}\ \mathcal U_{n+m-1}(\mathfrak g),
\]
which implies that the graded map $\text{gr}(\phi):\operatorname{gr}(\mathcal U(\overline{\mathfrak g}))\rightarrow\operatorname{gr}(\mathcal U_\ast(\mathfrak g))$, defined, at the level of the homogeneous components, by
\[
	\text{gr}_n(\phi)\big(x_1\cdots x_n\ \text{mod}\ \mathcal U_{n-1}(\overline{\mathfrak g})\big)
	=\phi(x_1 \cdots x_n)\text{mod}\ \mathcal U_{\ast,n-1}(\mathfrak g)
\]
is an isomorphism, proving that $\phi:\mathcal U(\overline{\mathfrak g})\rightarrow\mathcal U_\ast(\mathfrak g)$ is an isomorphism of filtered algebras. It is easy now to show that this morphism is compatible with the Hopf algebra structure maps, which implies the statement of the theorem.
\end{proof}

\begin{remark}
Note that to prove the theorem above one could argue as follows. First note that $\mathcal U_\ast(\mathfrak g)$ is a co-commutative and connected Hopf algebra, which implies, by the Cartier--Quillen--Milnor--Moore's theorem \cite{Cartier07}, that it is the enveloping algebra  of the Lie algebra of its primitive elements. Furthermore, since the co-product of $\mathcal U_\ast(\mathfrak g)$ is the same of the one of $\mathcal U(\mathfrak g)$, one can conclude that the Lie algebra of the primitive elements of $\mathcal U_\ast(\mathfrak g)$ is $\mathfrak g$. Finally, since $x\ast y-y\ast x=\llbracket x,y\rrbracket$, for all $x,y\in\mathfrak g$, $\mathcal U_\ast(\mathfrak g)$ is isomorphic to $\mathcal U(\overline{\mathfrak g})$ and $\phi$ is an isomorphism. 
\end{remark}

In the general case, on the other hand, it is difficult to say more about the isomorphism $\phi:\mathcal U(\bar{\mathfrak g})\rightarrow\mathcal U_*(\mathfrak g)$. One has the following nice combinatorial description.
If $m_\cdot: \mathcal U(\overline{\mathfrak g}) \otimes \mathcal U(\overline{\mathfrak g}) \to \mathcal U(\overline{\mathfrak g})$ denotes the product in $\mathcal U(\overline{\mathfrak g})$, i.e., $m_\cdot(A \otimes B)=A \cdot B$ for any $A,B \in \mathcal U(\overline{\mathfrak g})$, then the Hopf algebra isomorphism $\phi: \mathcal U(\overline{\mathfrak g}) \to \mathcal U_*(\mathfrak g)$ in Theorem \ref{thm:KLM} can be described as follows. From the proof of Theorem \ref{thm:KLM} it follows that $\phi$ restricts to the identity on $\mathfrak g \hookrightarrow \mathcal U(\mathfrak g)$. Moreover, for $x_1,x_2,x_3 \in \mathfrak g$ we find
$$
	\phi(x_1 \cdot x_2) = \phi(x_1) * \phi(x_2) = x_1 * x_2 =x_1x_2 + x_1 \triangleright  x_2,
$$
and 
\allowdisplaybreaks{
\begin{align}
	\phi(x_1 \cdot x_2 \cdot x_3) &= x_1 * x_2 * x_3 \nonumber \\	
	&= x_1(x_2 * x_3) + x_1 \triangleright (x_2 * x_3) \label{recursion}\\
	&=x_1x_2x_3 + x_1(x_2 \triangleright x_3) + x_2(x_1 \triangleright x_3) 
			+ (x_1 \triangleright  x_2)x_3 + x_1 \triangleright(x_2 \triangleright x_3). \nonumber 
\end{align}}
Equality \eqref{recursion} can be generalized to the following simple recursion for words in $\mathcal U(\overline{\mathfrak g})$ with $n>0$ letters
\begin{equation}
	\phi(x_1 \cdot \cdots \cdot x_n) = x_1\phi(x_2 \cdot \cdots \cdot x_n)  
								+ x_1 \triangleright \phi(x_2 \cdot \cdots \cdot x_n)  \label{eq:PHIrecursion1}.
\end{equation}
Recall that $x \triangleright {1}=0$ for $x \in \mathfrak g$, and $\phi({1})={1}$. From the fact that the post-Lie product on $\mathfrak g$ defines a linear map $d:\mathfrak g \rightarrow \operatorname{Der}\big(\mathcal U(\mathfrak g)\big),$ we deduce that the number of terms on the righthand side of the recursion \eqref{eq:PHIrecursion1} is given with respect to the length $n=1,2,3,4,5,6$ of the word $x_1 \cdot \cdots \cdot x_n \in \mathcal U_*(\mathfrak g)$ by 1, 2, 5, 15, 52, 203, respectively. These are the Bell numbers $B_i$, for $i=1,\ldots,6$, and for general $n$, these numbers satisfy the recursion $B_{n+1} = \sum_{i=0}^n {n \choose i} B_i$. Bell numbers count the different ways the set $[n]$ can be partition into disjoint subsets. 
From this we deduce the general formula for $x_1 \cdot \cdots \cdot x_n \in \mathcal U(\overline {\mathfrak g})$
\begin{equation}
	\phi(x_1 \cdot \cdots \cdot x_n) = x_1 * \cdots * x_n = \sum_{\pi \in P_n} X_\pi \in  \mathcal U( {\mathfrak g})    \label{eq:PHIrecursion2},
\end{equation}
where $P_n$ is the lattice of set partitions of the set $[n]=\{1,\dots,n\}$, which has a partial order of refinement ($\pi \leq \kappa$ if $\pi$ is a finer set partition than $\kappa$). Remember that a partition $\pi$ of the (finite) set $[n]$ is a collection of (non-empty) subsets $\pi=\{\pi_1,\dots,\pi_b\}$ of $[n]$, called blocks, which are mutually disjoint, i.e., $\pi_i \cap \pi_j=\emptyset$ for all $i\neq j$, and whose union $\cup_{i=1}^b \pi_i =[n]$. We denote by $|\pi|:=b$ the number of blocks of the partition $\pi$, and $|\pi_i|$ is the number of elements in the block $\pi_i$. Given $p,q \in [n]$ we will write that $p \sim_{\pi} q$ if and only if they belong to same block. The partition $\hat{1}_n = \{\pi_1\}$ consists of a single block, i.e., $|\pi_1|=n$. It is the maximum element in $P_n$. The partition $\hat{0}_n=\{\pi_1,\dots,\pi_n\}$ has $n$ singleton blocks, and is the minimum partition in $P_n$.

The element $X_\pi$ in \eqref{eq:PHIrecursion2} is defined as follows
\begin{equation}
	X_{\pi} := \prod_{\pi_i \in \pi} x(\pi_i), \label{eq:PHIrecursion2a} 
\end{equation}
where $x(\pi_i):= \ell^{\triangleright }_{x_{k_1^i}} \circ \ell^{\triangleright }_{x_{k_2^i}} \circ\cdots \circ  \ell^{\triangleright }_{x_{k_{l-1}^i}}(x_{k_l^i})$ for the block $\pi_i=\{k_1^i,k_2^i,\ldots ,k_l^i\}$ of the partition $\pi=\{\pi_1, \ldots, \pi_m\}$, and $\ell^{\triangleright}_{a}(b):= a \triangleright b$, for $a,b$ elements in the post-Lie algebra  $\mathfrak g \hookrightarrow \mathcal U(\mathfrak g)$. Recall that $k_l^i \in \pi_i$ is the maximal element in this block. 

%For instance
%$$
%	X_{\scalebox{0.6}{\strich\strich\strich}}=x_1x_2x_3, 
%	\quad
%	X_{\scalebox{0.6}{\strich \n}}=x_1(x_2 \triangleright x_3),
%	\quad
%	X_{\scalebox{0.6}{\nin}}=x_2(x_1 \triangleright x_3),
%$$
%$$
%	X_{\scalebox{0.6}{\n\hspace{0.05cm} \strich }}=(x_1 \triangleright x_2)x_3,
%	\quad
%	X_{\scalebox{0.6}{\n\hspace{0.00cm} \n}}=x_1 \triangleright(x_2 \triangleright x_3)
%$$

\begin{remark}
Defining $m_i:=\phi(x^{\cdot i})$ and $d_i := \ell^{\triangleright i-1}_x(x):=x \triangleright (\ell^{\triangleright i-2}_x(x))$, $ \ell^{\triangleright 0}:=\mathrm{id}$, we find that \eqref{eq:PHIrecursion2} is the $i$-th-order non-commutative Bell polynomial, $m_i = {\mathrm{B}}^{nc}_i(d_1,\ldots,d_i)$. See \cite{ELM14,LMK2} for details. 
\end{remark}

Next we state a recursion for the compositional inverse $\phi^{-1}(x_1 \cdots x_n)$ of the word $x_1 \cdots x_n \in \mathcal U(\mathfrak g)$. First, it is easy to see that $\phi^{-1}(x_1x_2)=x_1 \cdot x_2 - x_1 \triangleright  x_2 \in \mathcal U(\overline{\mathfrak g})$. Indeed, since $\phi$ is linear and the identity on $\mathfrak g \hookrightarrow \mathcal U(\mathfrak g)$, we have
\[
	\phi(x_1 \cdot x_2 - x_1 \triangleright  x_2)= x_1 * x_2 - x_1 \triangleright  x_2 = x_1x_2,
\]	
and 
\allowdisplaybreaks{
\begin{eqnarray*}
	\phi^{-1}(x_1x_2x_3 ) &=& x_1 \cdot x_2 \cdot x_3 
		- \phi^{-1}(x_1(x_2 \triangleright x_3)) 
		- \phi^{-1}(x_2(x_1 \triangleright x_3)) 
		- \phi^{-1}((x_1 \triangleright  x_2)x_3)\\
		&-& x_1 \triangleright(x_2 \triangleright x_3)
\end{eqnarray*}}
which is easy to verify. In general, we find the recursive formula for $\phi^{-1}(x_1 \cdots x_n) \in \mathcal U(\overline{\mathfrak g})$
\begin{equation}
	\phi^{-1}(x_1 \cdots x_n) = x_1 \cdot \cdots \cdot x_n - \sum_{\hat{0}_n < \pi \in P_n} \phi^{-1}(X_\pi) \label{eq:PHIrecursion3}.
\end{equation}
This is well-defined since in the sum on the righthand side all partitions have less than $n$ blocks. 

Observe now that since $\phi$ maps the augmentation ideal of $\mathcal U(\bar{\mathfrak g})$ to the one of $\mathcal U_*(\mathfrak g)$ it extends to an isomorphism between the completions of the two universal enveloping algebras $\hat{\phi}:\hat{\mathcal U}(\bar{\mathfrak g})\rightarrow\hat{\mathcal U}_*(\mathfrak g)$, see Section \ref{sec:LieTheory}. We are interested in the inverse of the group-like element $\exp(x) \in \mathcal G(\hat{\mathcal U}(\mathfrak g))$, $x\in\mathfrak g$, with respect to $\hat{\phi}$. It follows from the inverse of the word $x^n \in \hat{\mathcal U}(\mathfrak g)$, i.e., $\hat{\phi}^{-1}(\exp(x))=\sum_{n \ge 0} \frac{1}{n!} \hat\phi^{-1}(x^n)$. 

\begin{theorem}\label{thm:FinverseChi}
For each $x \in \mathfrak g$, there exists an unique element $\chi(x) \in \mathfrak g$, such that 
\begin{equation}
	\exp(x) = \exp^*(\chi(x)). \label{group-like}
\end{equation}
\end{theorem}

\begin{proof}
For $x \in \mathfrak g$ the exponential $\exp(x)$ is a group-like element in $\mathcal G(\hat{\mathcal U}(\mathfrak g))$. The proof of Theorem \ref{thm:FinverseChi} involves calculating the inverse of the group-like element $\exp(x) \in \mathcal G(\hat{\mathcal U}(\mathfrak g))$ with respect to the map $\hat{\phi}$. Indeed, we would like to show that $\hat{\phi}^{-1}(\exp(x)) = \exp^\cdot(\chi(x)) \in \mathcal G(\hat{\mathcal U}(\bar{\mathfrak g}))$, from which identity \eqref{group-like} follows
$$
	\hat{\phi}\circ\hat{\phi}^{-1}(\exp(x)) = \exp(x) = \hat{\phi}\circ \exp^\cdot(\chi(x)) = \exp^*(\chi(x)),
$$
due to $\hat\phi$ being an algebra morphism from $\hat{\mathcal U}(\overline{\mathfrak g})$ to $\hat{\mathcal U}_*(\mathfrak g)$, which reduces to the identity on~${\mathfrak g}$.

First we show that for $x\in \mathfrak g$, the element $\chi(x)$ is defined inductively. For this we consider the expansion $\chi(xt):=xt + \sum_{m>0} \chi_m(x)t^m$ in the parameter $t$. Comparing $\exp^*(\chi(xt))$ order by order with $\exp(xt)$ yields at second order in $t$
$$
	 \chi_2(x) := \frac{1}{2}x^2 - \frac{1}{2}x * x=- \frac{1}{2}x \triangleright x \in \mathfrak g. 
$$
At third order we deduce from \eqref{group-like} that 
\allowdisplaybreaks{
\begin{align*}
	\lefteqn{\chi_3(x) := -\frac{1}{3!} \sum_{\hat{0}_3 < \pi \in P_3} X_\pi   - \frac{1}{2}  \chi_2(x) * x - \frac{1}{2}  x * \chi_2(x)} \\
		      &= -\frac{1}{3!} \sum_{\hat{0}_3 < \pi \in P_3} X_\pi  
		      			+ \frac{1}{4}  \big((x \triangleright x) x + (x \triangleright x) \triangleright x\big)
					+ \frac{1}{4}  \big(x (x \triangleright x) + x \triangleright (x \triangleright x)\big)\\
		      &= -\frac{1}{3!}\big( 2x(x \triangleright x) + (x \triangleright  x)x + x \triangleright (x \triangleright x)\big)
		      			+ \frac{1}{4}  \big((x \triangleright x) x + (x \triangleright x) \triangleright x 
					+ x (x \triangleright x) + x \triangleright (x \triangleright x)\big)\\
		     &= \frac{1}{12} [(x \triangleright x), x] 
		     			+ \frac{1}{4} (x \triangleright x) \triangleright x 
		     				+  \frac{1}{12} x \triangleright (x \triangleright x) \\
		     &=  \frac{1}{6} [\chi_1(x), \chi_2(x)] 
		     			- \frac{1}{2} \chi_2(x)\triangleright x 
		     				-  \frac{1}{6} x \triangleright \chi_2(x),		 
\end{align*}}
where we defined $\chi_1(x):=x$. The $n$-th order term is given by
\allowdisplaybreaks{
\begin{align}
\label{eq:nth-order}
	\chi_n(x) &:= -\frac{1}{n!} \sum_{\hat{0}_n < \pi \in P_n} X_\pi  
		- \sum_{k=2}^{n-1} \frac{1}{k!} \sum_{p_1 + \cdots + p_k = n \atop p_i > 0}  \chi_{p_1}(x) *  \chi_{p_2}(x) * \cdots *  \chi_{p_k}(x)\\
		      &= \frac{1}{n!} x^n -\frac{1}{n!} x^{*n} 
		      - \sum_{k=2}^{n-1} \frac{1}{k!} \sum_{p_1 + \cdots + p_k = n \atop p_i > 0}  \chi_{p_1}(x) *  \chi_{p_2}(x) * \cdots *  \chi_{p_k}(x).\end{align}}
From this we derive an inductive description of the terms $\chi_n(x) \in \hat{\mathcal U}_*({\mathfrak g})$ depending on the $\chi_p(x)$ for $1 \le p \le n-1$
\begin{equation}
	\chi_n(x) := \frac{1}{n!} x^n 
		      - \sum_{k=2}^{n} \frac{1}{k!} \sum_{p_1 + \cdots + p_k = n \atop p_i > 0}  \chi_{p_1}(x) *  \chi_{p_2}(x) * \cdots *  \chi_{p_k}(x).
		      \label{chi-map}
\end{equation} 

We have verified directly that the first three terms, $\chi_i(x)$ for $i=1,2,3$, in the expansion $\chi(xt):=xt + \sum_{m>0} \chi_m(x)t^m$ are in $ \mathfrak g$. However, showing that $\chi_n(x) \in \mathfrak g$ for $n>3$ is more difficult using formula \eqref{chi-map}. We therefore follow another strategy. At this stage \eqref{chi-map} implies that $\chi(x) \in \hat{\mathcal U}_*({\mathfrak g})$ exists. Since $x \in \mathfrak g$, we have that $\exp(x)$ is group-like, i.e., $\hat{\Delta} (\exp(x)) = \exp(x) \hat\otimes \exp(x)$. Recall that $\hat{\mathcal U}_*({\mathfrak g})$ is a complete Hopf algebra with the same coproduct $\hat{\Delta}$. Hence
$$
	\hat{\Delta}(\exp^*(\chi(x)))
			= \hat{\Delta} (\exp(x))
			= \exp(x) \hat\otimes \exp(x) 
			= \exp^*(\chi(x)) \hat\otimes \exp^*(\chi(x)).
$$
Using $\hat\phi$ we can write $\hat\phi \hat\otimes \hat\phi \circ \hat{\Delta}_{\overline{\mathfrak g}}(\exp^\cdot(\chi(x))) = \hat\phi \hat\otimes \hat\phi \circ (\exp^\cdot(\chi(x)) \hat\otimes \exp^\cdot(\chi(x))),$ which implies that $\exp^\cdot(\chi(x))$ is a group-like element in $\hat{\mathcal U}(\overline{\mathfrak g})$
$$
	 \hat{\Delta}_{\overline{\mathfrak g}}(\exp^\cdot(\chi(x))) = \exp^\cdot(\chi(x)) \hat\otimes \exp^\cdot(\chi(x)).
$$ 
Since $\hat{\mathcal U}(\overline{\mathfrak g})$ is a complete filtered Hopf algebra, the relation between group-like and primitive elements is one-to-one, see Section 
\ref{sec:LieTheory}. This implies that $\chi(x) \in \overline{\mathfrak g} \simeq {\mathfrak g}$, which proves equality \eqref{group-like}. Note that $\chi(x)$ actually is an element of the completion of the Lie algebra $\mathfrak g$. However, the latter is part of $\hat{\mathcal U}({\mathfrak g})$. 
\end{proof}

\begin{corollary}\label{cor:diffeqChi}
Let $x \in {\mathfrak g}$. The following differential equation holds for $\chi(xt) \in \mathfrak g[[t]]$
\begin{equation}
\label{proof-key2}
	\dot \chi(xt) =  {\rm dexp}^{*-1}_{-\chi(xt)}\Big( \exp^*\big(-\chi(xt)\big) \triangleright   x\Big).
\end{equation}
The solution $\chi(xt)$ is called post-Lie Magnus expansion. 
\end{corollary}

\begin{proof} Recall the general fact for the $\rm{dexp}$-operator \cite{Blanes}
$$
	\exp^*({-\beta(t)}) \ast \frac{d }{dt}\exp^*({\beta(t)}) 
	= \exp^*({-\beta(t)}) \ast {\rm{dexp}}^\ast _{\beta}(\dot{\beta}) *\exp^*({\beta(t)}) 
	={\rm{dexp}}^\ast _{-\beta}(\dot{\beta}),
$$
where 
$$
	{\rm{dexp}}^\ast _{\beta}(x):= \sum_{n \ge 0} \frac{1}{(n+1)!}\operatorname{ad}^{(\ast n)}_\beta(x)
	\qquad {\rm{and}} \qquad
	{\rm dexp}^{\ast  -1}_{\beta}(x):=\sum_{n \ge 0} \frac{b_n}{n!} \operatorname{ad}^{(\ast n)}_\beta(x).
$$
Here $b_n$ are the Bernoulli numbers and $\operatorname{ad}^{(\ast k)}_a(b):=[a,\operatorname{ad}^{(\ast k-1)}_a(b)]_\ast$. This together with the differential equation $\frac{d}{dt}\exp^*(\chi(xt)) = \exp(xt)x$ deduced from \eqref{group-like}, implies
\allowdisplaybreaks{ 
\begin{eqnarray*}
	 &&{\rm{dexp}}^{*}_{-\chi(xt)}\big(\dot \chi(xt)\big) 
	  			 = \exp^*\big(-\chi(xt)\big)* (\exp(xt)x) \nonumber\\ 
	 			 &=&  \exp^*\big(-\chi(xt)\big)
				 	\Big(\exp^*\big(-\chi(xt)\big) \triangleright   (\exp(xt)x)\Big) \label{step1}\\					 								&=&	 \exp^*\big(-\chi(xt)\big)
					 \bigg(
						\big(\exp^*\big(-\chi(xt)\big) \triangleright   \exp(xt)\big) 
						 \big(\exp^*\big(-\chi(xt)\big) \triangleright   x\big)
					 \bigg) \label{step2}\\ 
				 &=&	 \exp^*\big(-\chi(xt)\big)
					\bigg(
						 \big(\exp^*\big(-\chi(xt)\big) \triangleright   \exp^*\big(\chi(xt)\big)\big) 
				 		\big(\exp^*\big(-\chi(xt)\big) \triangleright   x\big)
					\bigg) \label{step3}\\ 
				  &=&
				  	\bigg( \exp^*\big(-\chi(xt)\big)
				 		\Big(\exp^*\big(-\chi (xt)\big) \triangleright 
						\exp^*\big(\chi(t a)\big)\Big)  
					\bigg)
						 \big(\exp^*\big(-\chi (xt)\big) \triangleright x\big)
									 \label{step4}\\ 
				 &=&
			 		 \Big( \exp^*\big(-\chi(xt)\big) * 
					 \exp^*\big(\chi(xt)\big) \Big)
				 	\big(\exp^*\big(-\chi(xt)\big) \triangleright   x\big) 	\nonumber\\ 
				 &=& \exp^*\big(-\chi(xt)\big) \triangleright x. 			\nonumber
\end{eqnarray*}} 
The claim in \eqref{proof-key2} follows after inverting $ {\rm{dexp}}^{*}_{-\chi(xt)}\big(\dot \chi(xt)\big) $. Note that we used successively \eqref{eq:post-LieU}, \eqref{eq:pha2} and \eqref{group-like}
\end{proof}

Let us return to point $3.$ of Remark  \ref{rem:notrem} in Section \ref{sec:lieapospre}, and assume that the post-Lie algebra $(\mathfrak g,\triangleright,[\cdot,\cdot])$ is equipped with an abelian Lie bracket. This implies that $(\mathfrak g,\triangleright)$ reduces to a left pre-Lie algebra. The complete universal enveloping algebra  $\hat{\mathscr U}(\mathfrak g)$ becomes the complete symmetric algebra ${\hat S}_\mathfrak g$. This is the setting of \cite{OudomGuin}. Identity \eqref{group-like} was analyzed in the pre-Lie algebra context in \cite{ChapPat}.  

\begin{corollary}\label{cor:pre-LieMag}
For the pre-Lie algebra $(\mathfrak g,\triangleright,[\cdot,\cdot]=0)$, identity \eqref{group-like} in $\hat{S}_\mathfrak g$ is solved by the pre-Lie Magnus expansion
$$
	\chi(x) = \frac{\ell^{\triangleright}_{-\chi(x)}}{e^{\ell^{\triangleright}_{-\chi(x)}} - 1}(x)
	=\sum_{k\ge 0} \frac{(-1)^kb_k}{k!} \ell^{\triangleright k}_{\chi(x)}(x),
$$ 
where $b_n$ is the $n$-th Bernoulli number.
\end{corollary}

\begin{proof}
The proof of this result was given in \cite{ChapPat} and follows directly from identity \eqref{group-like} in Theorem \ref{thm:FinverseChi}, i.e., by calculating the Lie algebra element $\chi(x)$ as the $\log^*(\exp(x))$ in $\hat{S}_\mathfrak g$.
\end{proof}

The next proposition will be useful in the context of Lie bracket flow equations. 

\begin{proposition}\label{prop:star-sol}
\begin{equation}
\label{LieflowSol}
	a(t) :=\exp^*\big(- \chi(a_0t)\big)\triangleright a_0.
\end{equation}
solves the non-linear post-Lie differential equation with initial value $a(0)=a_0$
\begin{equation}
\label{post-Lie-flow}
	\dot{a}(t) = - a(t) \triangleright  a(t).
\end{equation}
\end{proposition}

\begin{proof}
We calculate 
\allowdisplaybreaks{ 
\begin{eqnarray*}
	\dot{a}(t) 	&=& \Big(-{\rm{dexp}}_{-\chi(a_0t)}\big(\dot {\chi}(a_0t)\big) * \exp^*\big(-\chi(a_0t)\big) \Big)\triangleright a_0\\
			&=& -{\rm{dexp}}_{-\chi(a_0t)}\big(\dot {\chi}(a_0t)\big) \triangleright  \Big(\exp^*\big(-\chi(a_0t)\big) \triangleright a_0\Big)\\
			&=& - a(t) \triangleright  a(t),
\end{eqnarray*}}
where we used that $\exp^*\big(-\chi(a_0t)\big)\triangleright a_0 = {\rm{dexp}}_{-\chi(a_0t)}\big(\dot {\chi}(a_0t)\big) = a(t)$.
\end{proof}

%%%%%%%%%%%%%%%%%%%%%%%%%%%%%

\subsection{Factorization theorems and $r$-matrices}

In this subsection we will suppose that the post-Lie algebra structure on $\mathfrak g$ is defined in terms of a solution of the mCYBE, see Subsection \ref{ss:YBPP}. Recall that in this case $\bar{\mathfrak g}=\mathfrak g_R$ implying that $\mathcal U(\bar{\mathfrak g})=\mathcal U(\mathfrak g_R)$ and, correspondingly, that $\hat{\mathcal U}(\bar{\mathfrak g})=\hat{\mathcal U}(\mathfrak g_R)$. In what follows we will prove that for this particular class of post-Lie algebras, the isomorphism $\phi$ admits an explicit description in terms of the structure of the two Hopf algebras of the universal enveloping algebras $\mathcal U(\mathfrak g_R), \mathcal U_*(\mathfrak g)$. To this end first we prove the following result.

\begin{theorem}\label{thm:lineariso}
The map $F:\mathscr U(\mathfrak g_R)\rightarrow\mathscr U_{\ast}(\mathfrak g)$ defined by
\begin{equation}
\label{eq:sigma}
	F=m_{\mathfrak g}\circ (\operatorname{id}\otimes S_{\mathfrak g})\circ (R_+\otimes R_-)\circ\Delta_{\mathfrak g_R},
\end{equation}
is an isomorphism of Hopf algebras. 
\end{theorem}
\begin{proof}
First note that $m_{\mathfrak g}$ and $S_{\mathfrak g}$ denote respectively the product and the antipode of $\mathcal U(\mathfrak g)$, whereas $\Delta_{\mathfrak g_R}$ denotes the co-product in $\mathcal U(\mathfrak g_R)$. Also, recall that $\mathcal U_\ast(\mathfrak g)$ denotes the Hopf algebra $(\mathcal U(\mathfrak g),m_\ast,1,\Delta,\epsilon,S_\ast)$. The slightly more cumbersome notation is applied in order to make the presentation more traceable. Given an element $x \in\mathfrak g_R \hookrightarrow \mathcal U(\mathfrak g_R)$, one has that
\allowdisplaybreaks{ 
\begin{align*}
	F(x) 	&=m_{\mathfrak g}\circ (\operatorname{id}\otimes S_{\mathfrak g})\circ (R_+\otimes R_-)\circ\Delta_{\mathfrak g_R}(x)\\
		&= m_{\mathfrak g}\circ (\operatorname{id}\otimes S_{\mathfrak g})\circ (R_+\otimes R_-)(x \otimes {1} + {1} \otimes x)\\
		&= m_{\mathfrak g}\circ (\operatorname{id}\otimes S_{\mathfrak g})(R_+(x)\otimes {1} + {1} \otimes R_-(x))\\
		&= m_{\mathfrak g}(R_+(x)\otimes  {1} - {1} \otimes R_-(x))\\
		&= R_+(x)-R_-(x)
		= x \in \mathfrak g \hookrightarrow \mathcal U(\mathfrak g),
\end{align*}}
showing that $F$ restricts to the identity map between $\mathfrak g_R$ and $\mathfrak g$. As in Lemma \ref{lem:coprodast} we have
\[
	\Delta_{\mathfrak g_R}(x_1\cdots x_n) = x_1 \cdots x_n \otimes {1} + {1} \otimes x_1 \cdots x_n
	+ \sum_{k=1}^{n-1}\sum_{\sigma\in\Sigma_{k,n-k}} x_{\sigma(1)}\cdots x_{\sigma(k)} \otimes x_{\sigma(k+1)}\cdots x_{\sigma (n)}.
\]
Since $R_\pm$ are homomorphisms of unital associative algebras, one can easily show that for every $x_{1}\cdots x_k\in\mathcal U_k(\mathfrak g_R)$:
 \begin{eqnarray*}
	\lefteqn{F (x_1\cdots x_k)=R_+(x_1)\cdots R_+(x_k) + (-1)^k R_-(x_k)\cdots R_-(x_1) +} \\
				&\sum_{l=1}^{k-1}\sum_{\sigma\in\Sigma_{l,k-l}}(-1)^{k-l}R_+(x_{\sigma(1)})\cdots R_+(x_{\sigma(l)})R_-(x_{\sigma(k)})\cdots R_-(x_{\sigma(l+1)}) \in \mathcal U_k(\mathfrak g).
\end{eqnarray*}

This shows, in particular, that $F$ maps homogeneous elements to homogeneous elements. Moreover, a simple computation shows that 
\begin{align}
	F(x_1 \cdots x_k)	
	&= R_+(x_{1}) \cdots R_+(x_{k}) + (-1)^k R_-(x_{k}) \cdots R_-(x_{1}) \nonumber\\
	&+ \sum_{l=1}^{k-1} \sum_{\sigma\in\Sigma_{l.k-l}}(-1)^{k-l}
	R_+(x_{\sigma(1)})\cdots R_+(x_{\sigma(l)})\cdot R_-(x_{\sigma(k)}) \cdots R_-(x_{\sigma(l+1)}). 
	\nonumber
\end{align}
for each monomial $x_1\cdots x_k$. Then, using the definition of the $*$-product, one can easily see that $F(x_1x_2)=x_1x_2 + [R_-(x_1),x_2] =  x_1x_2 + x_1 \triangleright_- x_2$, where $\triangleright$ is defined in \eqref{def:RBpost-Lie} (and lifted to $\mathcal U(\mathfrak g) $),  which implies that 
$F(x_1x_2) = x_1 * x_2 \in \mathcal U_*(\mathfrak g) $. Using a simple induction on the length of the monomials, the above calculation extends to all of $\mathcal U(\mathfrak g_R)$ and shows that $F$ is a morphism of unital, associative algebras.
On the other hand, since $F(x)=x$ for all $x\in\mathfrak g_R$ and $\mathcal U(\mathfrak g)$ is generated by $\mathfrak g$, one can conclude that $F$ is a surjective. Furthermore, 
since
\[
\Delta\circ F(x)=x\otimes 1+1\otimes x=(F\otimes F)\circ\Delta_{\mathfrak g_R}(x)
\]
for all $x\in\mathfrak g_R$ one can conclude that the two algebra morphisms $\Delta\circ F$ and $(F\otimes F)\circ\Delta$ are equal, since they coincide on $\mathfrak g_R$, which implies that $F$ is bialgebra morphism. This, in turn, implies that $F$ is an Hopf algebra morphism since the compatibility with the antipodes is automatically fulfilled. To conclude the proof of the theorem, it suffices to show that $F$ is injective. But since $\text{ker}\,(F)$ is a co-ideal, if non-trivial it should contain a non-zero primitive element, which is not the case since $F(x)=x$ for all $x\in\mathfrak g_R$ and $\mathfrak g_R=\mathcal P(\mathcal U(\mathfrak g_R))$.
\end{proof}

Comparing this result with the Theorem \ref{thm:KLM} of the previous section, one has 

\begin{proposition}\label{prop:idenF}
If the post-Lie algebra $(\mathfrak g,[\cdot,\cdot],\triangleright_-)$ is defined in terms of an $r$-matrix $R$ via formula \eqref{def:RBpost-Lie}, then the isomorphism $\phi$ of Theorem \ref{thm:KLM} assumes the explicit form given in formula \eqref{eq:sigma}, i.e., $\phi=F$.
\end{proposition}
\begin{proof}
In fact note that both $\phi$ and $F$ are isomorphisms of filtered, unital associative algebras taking values in $\mathscr U_*(\mathfrak g)$, restricting to the identity map on $\mathfrak g_R$ which is the generating set of $\mathscr U(\mathfrak g_R)$.
\end{proof}

At this point it is worth making the following observation, which will be useful in what follows.

\begin{corollary}\label{cor:dec}
Every $A\in\mathscr U(\mathfrak g)$ can be written uniquely as
\begin{equation}
	A= R_+(a_{(1)})S_{\mathfrak g}(R_-(a_{(2)})) \label{eq:factinu1}
\end{equation}
for a suitable element $a\in\mathscr U(\mathfrak g_R)$, where we wrote the co-product of this element using the Sweedler's notation, i.e., $\Delta_{\mathfrak g_R}(a)=a_{(1)}\otimes a_{(2)}$.
\end{corollary}

\begin{proof}
The proof follows from Theorem \ref{thm:lineariso}, noticing that for each $a \in \mathscr U(\mathfrak g_R)$,
\[
	F(a)= R_+(a_{(1)})S_{\mathfrak g}(R_-(a_{(2)})).
\]
\end{proof}

Finally, in this more specialized context, we can give the following computational proof of the result contained in Theorem \ref{thm:KLM}.

We conclude this section with the following observation, see Remark \ref{rem:STSR}.

\begin{proposition}\label{prop:prodRSTS}
For all $A,B\in\mathcal U(\mathfrak g)$, one has that:
\begin{equation}
\label{eq:pSTS2}
	A \ast B = R_+(a_{(1)})B S_{\mathfrak g}(R_-(a_{(2)})),
\end{equation}
where $a\in\mathcal U(\mathfrak g_R)$ is the unique element, such that $A=F(a)$, see Corollary \ref{cor:dec}.
\end{proposition}

\begin{proof}
Let $a, b\in\mathcal U(\mathfrak g_R)$ such that $F(a)=A$ and $F(b)=B$. We will use Sweedler's notation for the co-product $\Delta_{\mathfrak g_R}(a)=a_{(1)}\otimes a_{(2)}$, and write $m_{\mathfrak g_R}(a\otimes b):=a\cdot b$ for the product in $\mathcal U(\mathfrak g_R)$.
\allowdisplaybreaks{
\begin{align*}
	A\ast B=F(a\cdot b)
	&=m_{\mathfrak g}\circ (\operatorname{id}\otimes S_{\mathfrak g})\circ (R_+\otimes R_-) \circ\Delta_{\mathfrak g_R}(a\cdot b)\\
	&=m_{\mathfrak g}\circ(\operatorname{id}\otimes S_{\mathfrak g})\circ (R_+\otimes R_-)(a_{(1)}\otimes a_{(2)})\cdot (b_{(1)}\otimes b_{(2)})\\
	&=m_{\mathfrak g}\circ(\operatorname{id}\otimes S_{\mathfrak g})\circ (R_+\otimes R_-)(a_{(1)}\cdot b_{(1)})\otimes (a_{(2)}\cdot  b_{(2)})\\
	&=m_{\mathfrak g}\circ(\operatorname{id}\otimes S_{\mathfrak g})\big(R_+(a_{(1)})R_+(b_{(1)})\otimes R_-(a_{(2)})R_-(b_{(2)})\big)\\
	&\stackrel{(*)}{=}m_{\mathfrak g}\big(R_+(a_{(1)}) R_+(b_{(1)})\otimes S_{\mathfrak g}(R_-(b_{(2)}))S_{\mathfrak g}(R_-(a_{(2)}))\big)\\
	&=R_+(a_{(1)}) R_+(b_{(1)})S_{\mathfrak g}(R_-(b_{(2)}))S_{\mathfrak g}(R_-(a_{(2)}))\\
	&= R_+(a_{(1)}) F(b)S_{\mathfrak g}(R_-(a_{(2)})\\
	&=R_+(a_{(1)}) BS_{\mathfrak g}(R_-(a_{(2)}),
\end{align*}}
which proves the statement. In equality $(*)$ we applied that $S_{\mathfrak g}(\xi\eta)=S_{\mathfrak g}(\eta)S_{\mathfrak g}(\xi)$.
\end{proof}

\begin{remark}[Link to the work of Semenov-Tian-Shansky and Reshetikhin]\label{rem:STSR} 
In this remark we aim to link the previous results to the ones described in the references \cite{STS3} and \cite{RSTS}. The map \eqref{eq:sigma} was first defined in \cite{STS3} (see also \cite{RSTS}), where it was used to push-forward to $\mathcal U(\mathfrak g)$ the associative product of $\mathcal U(\mathfrak g_R)$ using the formula 
\begin{equation}
A\ast B=F(m_{\mathfrak g_R}(F^{-1}(A)\otimes F^{-1}(B))),\label{eq:push}
\end{equation}
for all monomials $A,B\in\mathcal U(\mathfrak g)$. From the equality between the maps $\phi$ and $F$, see Proposition \ref{prop:idenF}, it follows at once that the associative product defined in $\mathcal U(\mathfrak g)$ by the authors of \cite{STS3, RSTS}, \emph{is} the product defined in formula \eqref{eq:post-LieU}. Moreover, to the best knowledge of the authors of the present note, in the references \cite{STS3,RSTS}, the Hopf algebra structure induced on (the underlying vector space of) $\mathcal U(\mathfrak g)$, by the push-forward of the associative product of $\mathcal U(\mathfrak g_R)$ was not disclosed. Via the theory of the post-Lie algebras, on one hand we could extend (part of) the results of \cite{STS3,RSTS} to an Hopf algebraic framework, while on the other, we could get a more computable formula for the product defined in \cite{STS3, RSTS}. In particular, note that, although the result in Proposition \ref{prop:prodRSTS} was stated in \cite{STS3,RSTS}, the product in formula \eqref{eq:pSTS2} is not easily computable, since it supposes the knowledge of the inverse of the map $F$. On the other hand, formula \eqref{eq:post-LieU} provides an explicit way to compute the $\ast$-product between any two monomials of $\mathcal U(\mathfrak g)$.
\end{remark}

In this final part we discuss an application of the result presented above to the problem of the factorization of the group like-elements of the completed universal enveloping algebra of $\mathfrak g_R$. This result should be compared with the one in Theorem \ref{thm:factorizationtheorem} in Subsection \ref{ss:Rfact}.
%Next we consider Theorem \ref{thm:factorizationtheorem} in the context of the universal enveloping algebra of $\mathfrak g$. To this end, one needs first to trade $\mathcal U(\mathfrak g)$ for its completion $\hat{\mathcal U}(\mathfrak g)$. 
We start observerving that, since $R_\pm:\mathcal U(\mathfrak g_R) \rightarrow \mathcal U(\mathfrak g)$ are algebra morphisms, they map the augmentation ideal of $\mathcal U(\mathfrak g_R)$ to the augmentation ideal of $\mathcal U(\mathfrak g)$  and, for this reason, both these morphisms extend to morphisms $R_\pm:\hat{\mathcal U}(\mathfrak g_R) \rightarrow \hat{\mathcal U}(\mathfrak g)$. In particular, the map $F$ extends to an isomorphism of (complete) Hopf algebras $\hat{F}:\hat{\mathcal U}(\mathfrak g_R)\rightarrow \hat{\mathcal U}_{\ast}(\mathfrak g)$, defined by 
\[
	\hat{F}=\hat{m}_{\mathfrak g}\circ (\operatorname{id} 
	\hat\otimes \hat{S}_{\mathfrak g})\circ(R_+\hat\otimes R_-)\circ\hat\Delta_{\mathfrak g_R},
\]
where, $\hat\Delta_{\mathfrak g_R}$ denotes the coproduct of $\hat{\mathcal U}(\mathfrak g_R)$, and with $\hat{m}_{\mathfrak g}$, $\hat{S}_{\mathfrak g}$ the product respectively the antipode of $\hat{\mathcal U}(\mathfrak g)$ are denoted. Let $\operatorname{exp}^{\cdot}(x)\in\mathcal G(\hat{\mathcal U}(\mathfrak g_R))$, $\operatorname{exp}^{\ast}(x)\in\mathcal G(\hat{\mathcal U}_{\ast}(\mathfrak g))$ and
$\exp(x) \in \mathcal G(\hat{\mathcal U}(\mathfrak g))$, the respective exponentials. 

Following \cite{RSTS} we now compare identity \eqref{eq:factinu2} with \eqref{eq:fac}. At the level of the universal enveloping algebra, the main result of Theorem \ref{thm:factorizationtheorem} can be rephrased as follows.

\begin{theorem}\label{thm:factcircled}
Every element $\operatorname{exp}^{\ast}(x)\in \mathcal G(\hat{\mathcal U}_{\ast}(\mathfrak g))$ admits the following factorization:
\begin{equation}
	\operatorname{exp}^{\ast}(x)=\exp({x_+})\exp({-x_-}).\label{eq:factinu2}
\end{equation}
\end{theorem}

\begin{proof}
To simplify notation, write $m_{\mathfrak g_R}(x\otimes y)=x\cdot y$, for all $x,y\in\mathfrak g_R$, so that for each $x \in \mathfrak g_R$, $x^{\cdot n}:=x\cdots x$. Then observe that, for each $n\geq 0$, one has
\[
	\hat F (x^{\cdot n})=R_+(x)^n+\sum_{l=1}^{n-1}(-1)^{n-l}{n\choose l}R_+(x)^lR_-(x)^{n-l}+(-1)^nR_-(x)^n.
\]
Then, after reordering the terms, one gets $\hat{F} (\operatorname{exp}_{\cdot}(x))=e^{x_+}e^{-x_-}.$ On the other hand, since $\hat F:\hat{\mathcal U}(\mathfrak g_R)\rightarrow \hat{\mathcal U}_\ast(\mathfrak g)$ is an algebra morphism, one obtains for each $n\geq 0$
\[
	\hat F(x^{\cdot n})	= \hat F(x)\ast\cdots\ast \hat F (x)
					= x^{\ast n},
\]
from which it follows that
\begin{eqnarray*}
	\hat{F} (\operatorname{exp}^{\cdot}(x))
		&=&\hat F(1)+\hat F(x)+\frac{\hat F(x^{\cdot 2})}{2!}+\cdots + \frac{\hat F(x^{\cdot n})}{n!} + \cdots =\operatorname{exp}^{\ast}(x),
\end{eqnarray*}
giving the result.
\end{proof}

The observation in Theorem \ref{thm:FinverseChi} implies for group-like elements in $\mathcal G(\hat{\mathcal U}(\mathfrak g))$ and $\mathcal G(\hat{\mathcal U}_*(\mathfrak g))$ that $\exp(x) = \exp^*(\chi(x))$, from which we deduce 

\begin{corollary} Group-like elements $\operatorname{exp}(x) \in \mathcal G(\hat{\mathcal U}(\mathfrak g))$ factorize
\begin{equation}
	\operatorname{exp}(x)=\exp({\chi_+(x)})\exp({-\chi_-(x)}).\label{eq:factast}
\end{equation}
\end{corollary}

\begin{proof}
The proof follows from Theorem \ref{thm:FinverseChi} and Theorem \ref{thm:factcircled}.
\end{proof}

\begin{remark} Looking at $\chi(x)$ in the context of $\hat{\mathcal U}(\mathfrak g)$, i.e., with the post-Lie product on $\mathfrak g$ defined in terms of the $r$-matrix, $x \triangleright_- y = [R_-(x),y]$, we find that $\chi_2(x) = -\frac{1}{2} [R_-(x),x]$ and 
$$
	\chi_3(x) = \frac{1}{4} [R_-([R_-(x),x]),x] + \frac{1}{12} ( [[R_-(x),x], x] + [R_-(x),[R_-(x),x]]).   				
$$
This should be compared with Eq.~(7) in \cite{EGM}, as well as with the results in \cite{EFLIM}.
\end{remark}

%%%%%%%%%%%%%%%%%%%%%%%%%%%%%

\subsection{Applications to isospectral flow equations}
\label{ssect:flow}

Recall Proposition \ref{prop:star-sol}. In the context of the post-Lie product $x \triangleright_- y := [R_-(x),y]$ induced on $\mathfrak g$ by an $r$-matrix $R$, this proposition says that the Lie bracket flow 
$$
	\dot{x}(t)=[x,R_-(x)], \quad x(0)=x_0
$$ 
has solution 
\allowdisplaybreaks{ 
\begin{eqnarray*}
	x(t) 	&=& \exp^*(-\chi(x_0t)) \triangleright_- x_0   					\nonumber \\						
		&=&  \exp\big(- R_-(\chi(x_0t))\big) x_0 \exp\big(R_-(\chi(x_0t))\big). 	\label{Lie-Flow}
\end{eqnarray*}}
The last equality follows from general results of post-Lie algebra. Since, $-\chi(x_0t) \in \mathfrak g$ we have  
\begin{eqnarray*}
	\exp^*(-\chi(x_0t)) \triangleright_- x_0
	&=& x_0 -  \chi(x_0t) \triangleright_- x_0 
	+ \frac{1}{2!} (\chi(x_0t) * \chi(x_0t)) \triangleright_- x_0 + \cdots\\
	&=& X_0 -  \chi(X_0t) \triangleright_- X_0 
	+ \frac{1}{2!} \chi(X_0t) \triangleright_- (\chi(x_0t) \triangleright_- x_0) + \cdots \\
	&=& \sum_{n\ge 0} \frac{(-1)^n}{n!} \operatorname{ad}_{R_-(\chi(x_0t))}^{(n)}x_0.
\end{eqnarray*}

%%%%%%%%%%%%%%%%%%%%%%%%%%%%%
%%%%%%%%%%%%%%%%%%%%%%%%%%%%%
%%%%%%%%%%%%%%%%%%%%%%%%%%%%%

\input{referenc}
\end{document}

%% file: referenc.tex
%%%%%%%%%%%%%%%%%%%%%%%% referenc.tex %%%%%%%%%%%%%%%%%%%%%%%%%%%%%%
% sample references
% %
% Use this file as a template for your own input.
%
%%%%%%%%%%%%%%%%%%%%%%%% Springer-Verlag %%%%%%%%%%%%%%%%%%%%%%%%%%
%
% BibTeX users please use
% \bibliographystyle{}
% \bibliography{}
%

%% and use \bibitem to create references.
%%
%% Use the following syntax and markup for your references if 
%% the subject of your book is from the field 
%% "Mathematics, Physics, Statistics, Computer Science"
%%
%% Contribution 
%\bibitem{science-contrib} Broy, M.: Software engineering --- from auxiliary to key technologies. In: Broy, M., Dener, E. (eds.) Software Pioneers, pp. 10-13. Springer, Heidelberg (2002)
%%
%% Online Document
%\bibitem{science-online} Dod, J.: Effective substances. In: The Dictionary of Substances and Their Effects. Royal Society of Chemistry (1999) Available via DIALOG. \\
%\url{http://www.rsc.org/dose/title of subordinate document. Cited 15 Jan 1999}
%%
%% Monograph
%\bibitem{science-mono} Geddes, K.O., Czapor, S.R., Labahn, G.: Algorithms for Computer Algebra. Kluwer, Boston (1992) 
%%
%% Journal article
%\bibitem{science-journal} Hamburger, C.: Quasimonotonicity, regularity and duality for nonlinear systems of partial differential equations. Ann. Mat. Pura. Appl. \textbf{169}, 321--354 (1995)
%%
%% Journal article by DOI
%\bibitem{science-DOI} Slifka, M.K., Whitton, J.L.: Clinical implications of dysregulated cytokine production. J. Mol. Med. (2000) doi: 10.1007/s001090000086 